%% file: PanelForecasting_arxiv.tex
\documentclass[11pt,fleqn]{article}

\usepackage{eurosym}
\usepackage{amsmath,amssymb,amsthm,graphicx,xcolor,rotating}
\usepackage{caption}
\usepackage{amsopn}
\usepackage{floatpag,chngcntr}
\usepackage[normalem]{ulem}
\usepackage[doublespacing]{setspace}

\usepackage[authoryear]{natbib}
\theoremstyle{plain}
\newtheorem{proposition}{Proposition}
\newtheorem{assumption}{Assumption}
\newtheorem{lemma}{Lemma}

\DeclareMathOperator*{\plim}{\mathrm{plim}}
\newcommand{\rowsep}[1]{\vspace*{#1 em}}
\newcommand{\bs}[1]{\boldsymbol{#1}}

\AtBeginEnvironment{thebibliography}{\linespread{1}\selectfont}

\newcommand{\lleft}{\left}
\newcommand{\rrvert}{\vert}
\newcommand{\rright}{\right}
\newcommand{\rrVert}{\Vert}
\newcommand{\llvert}{\vert}
\newcommand{\llVert}{\Vert}

\begin{document}

\title{Forecasting with panel data: Estimation uncertainty versus parameter
heterogeneity\thanks{%
We thank the three anonymous reviewers and the Co-editor, Stephane Bonhomme, for their helpful and
constructive comments. We also thank Laura Liu, Mahrad Sharifvaghefi, Ron
Smith, Cynthia Yang, Liying Yang, and seminar participants at the University of Pittsburgh, the Board of the Federal Reserve, ESEM, VTSS, IAAE, and at 2024
NBER-NSF Time Series Conference held at University of Pennsylvania for
helpful comments.}}

\author{M. Hashem Pesaran\thanks{%
University of Cambridge, UK, and University of Southern California, USA
Email: mhp1@cam.ac.uk} \and Andreas Pick\thanks{%
Erasmus University Rotterdam, Erasmus School of Economics, Burgemeester
Oudlaan 50, 3000DR Rotterdam, and Tinbergen Institute. Email:
andreas.pick@cantab.net} \and Allan Timmermann\thanks{%
UC San Diego, Rady School of Management, 9500 Gilman Drive, La Jolla CA
92093-0553. Email: atimmermann@ucsd.edu.}}
\maketitle

\begin{abstract}
We provide a comprehensive examination of the predictive accuracy of panel
forecasting methods based on individual, pooling, fixed effects, and empirical
Bayes estimation, and propose optimal weights for forecast combination
schemes. We consider linear panel data models, allowing for weakly exogenous
regressors and correlated heterogeneity. We quantify the gains from exploiting
panel data and demonstrate how forecasting performance depends on the degree
of parameter heterogeneity, whether such heterogeneity is correlated with
the regressors, the goodness-of-fit of the model, and the dimensions of
the data. Monte Carlo simulations and empirical applications to house prices
and CPI inflation show that empirical Bayes and forecast combination methods
perform best overall and rarely produce the least accurate forecasts for
individual series. \enlargethispage{3em} \newline
\end{abstract}

\thispagestyle{empty}

\newpage

\setcounter{page}{1}

\section{Introduction}

Panel data sets on economic and financial variables are widely available
at individual, firm, industry, regional, and country granularities and
have been extensively used for estimation and inference. Yet, panel estimation
methods have had a comparatively lower impact on common practices in economic
forecasting, which remain dominated by unit-specific forecasting models
or low-dimensional multivariate models such as vector autoregressions \citep{Hsi2022}. 
The relative shortage of panel applications in the economic forecasting
literature is, in part, a result of the absence of a deeper understanding
of the determinants of forecasting performance for different panel estimation
methods and the absence of guidelines on which methods work well in different
settings.

In this paper, we examine existing approaches and develop novel forecast
combination methods for panel data with possibly correlated heterogeneous
parameters. We conduct a systematic comparison of their predictive accuracy
in settings with different cross-sectional ($N$) and time ($T$) dimensions
and varying degrees of parameter heterogeneity, whether correlated or not.
Our analysis provides a deeper understanding of the determinants of the
performance of these methods across a variety of settings chosen for their
relevance to economic forecasting problems. This includes the important
choice of whether to use pooled versus individual estimates, or perhaps
a combination of the two approaches, with a focus on forecasting rather
than parameter estimation and inference.

We begin by exploring analytically the bias-variance trade-off between
individual, fixed effects (FE), and pooled estimation for forecasting.
Our analysis is conducted in a general setting that allows for weakly exogenous
regressors and correlated heterogeneity, consistent with the type of dynamic
panel models commonly used in empirical applications. We show how such
effects contribute to the mean squared forecast error (MSFE) of forecasts
based on individual, FE, and pooled estimates.

We next examine forecast combination methods. Estimation errors are well
known to lead to imprecisely estimated combination weights for data with
a small time-series dimension. Our main combination scheme assumes homogeneous
weights across individual variables. This allows us to use cross-sectional
information to reduce the effect of estimation error on the combination
weights compared to the conventional combination scheme that lets the weights
be individual-specific, which we also consider. To handle cases where the
pooling estimator imposes too much homogeneity, we also consider combinations
based on forecasts from the individual-specific and fixed effect estimators.

Our theoretical analysis of the individual and pooled estimation schemes
focuses on the case with finite $T$ and $N\rightarrow \infty $ and does
not require that $\sqrt{N}/T\rightarrow 0$ as $N$ and
$T\rightarrow \infty $, jointly, which is often assumed in the literature.
The estimation of the combination weights, however, requires $T$
$\rightarrow \infty $, but at a much slower rate compared to $N$.

Finally, we consider forecasts based on the empirical Bayesian (EB) approach
of \cite{Hsietal1999}. These are related to forecast combination and we
show for the empirical Bayes estimator that it can be thought of as a weighted
average of an estimator that allows for full heterogeneity and a pooled
mean group estimator. The empirical Bayes scheme assigns greater weight
to the pooled estimator, the lower the estimated degree of parameter heterogeneity
and so adapts to the degree of parameter heterogeneity characterizing a
given data set.\footnote{In the Supplemental Appendix, we
also report results based on the hierarchical Bayesian approach of \citet{LinSmi1972},
\cite{LeeGri1979}, and \cite{Madetal1997}.}

We evaluate the predictive accuracy of these alternative panel forecasting
methods through Monte Carlo simulations. The simulations explore the importance
to forecasting performance of the degree of parameter heterogeneity, how
correlated it is with the regressors, whether it affects intercepts or slopes, the value of the regressors in the forecast period, and dimensions of $N$ and $T$. In
the scenario with homogeneous parameters, forecasts based on pooled estimates
are most accurate. Forecasts based on fixed or random effect estimates
perform well, relative to other methods, when parameter heterogeneity is
confined to the intercepts and does not affect slopes. Outside these cases,
empirical Bayes and forecast combinations produce the most accurate forecasts
and are better able to handle parameter heterogeneity, whether correlated
or not, while being more robust in cases with a small $T$ than the individual-specific
approach.

Next, we consider two empirical applications selected to represent varying
degrees of heterogeneity and predictive power of the underlying forecasting
models. Our first application considers predictability of house prices
across 362 US metropolitan statistical areas (MSAs). In this application,
individual-specific forecasts perform quite poorly, producing the highest
MSFE values among all methods for more than 50\% of the MSAs. Forecasts
based on pooled estimates perform notably better and, across all forecasts,
reduce the average MSFE value by 8\% relative to the forecasts based on
individual estimates, though this gets reversed if the regressor set is
close to the sample average. Empirical Bayes and forecast combinations
work even better in this application, beating forecasts based on individual
estimates for over 90\% of MSAs while almost never generating the least
accurate forecasts for individual series.

Our second application considers forecasts for a panel containing 187 subcategories
of CPI inflation. In this application, forecasts based on individual estimates
generate the highest MSFE values for 44\% of the series. Forecasts based
on pooled estimates produce the highest MSFE values for 40\% of the individual
series but, conversely, generate the lowest MSFE-values for 20\% of the
series. Combination forecasts produce lower MSFE values for 73--97\% of
the individual CPI series than the individual forecasts while almost never
generating the largest MSFE value. Even better inflation forecasts are
produced by the empirical Bayes method, which is more accurate (in the
MSFE sense) than the individual forecasts for 98\% of the series and generates
the lowest MSFE values for 39\% of the individual variables while never
generating the largest MSFE value.

Overall, forecasts that use only the information on a given unit tend to
have loss distributions with wide dispersion across units. Their associated
forecasts are therefore sometimes the best but far more often the worst,
and their distribution of MSFE performance is often shifted to the right,
implying larger losses on average than for other methods. Forecasts based
on pooling, random effects, or fixed effects estimation tend to perform
better, on average, than the individual forecasts whose accuracy they beat
for the majority of series. However, relative to the individual-specific
forecasts, these approaches also tend to have a right-skewed MSFE distribution,
suggesting a high risk of poor forecasting performance for individual series
whose model parameters are very different from the average. Combinations
and empirical Bayes forecasts have narrower MSFE distributions across units,
often shifted to the left as they are centered around a smaller average
loss, and rarely produce the largest squared forecast error among all methods
we consider.

\paragraph*{Related literature} The review articles by \citeauthor{Bal2008} (\citeyear{Bal2008,Bal2013})
consider the forecasting performance of the best linear unbiased predictor
(BLUP) of \citet{Gol1962} in models with either fixed effects or random
effects. The BLUP estimator gives rise to a generalized least squares (GLS)
predictor, which Baltagi compares to models that allow for autoregressive
moving average (ARMA) dynamics in innovations as well as models with spatial
dependencies in the errors. \citet{TraUrg2009} use Monte Carlo simulations
to assess the forecasting performance of pooled, individual, and shrinkage
estimators and find that parameter heterogeneity is a key determinant of
the accuracy of different forecasts. \citet{BruSil2006}
consider a similar group of methods to forecast migration data and find
that fixed effects and shrinkage estimators perform best; see \citet{PicTim2024} for a review of the literature.

\cite{Wanetal2019} also propose forecast combination methods. However,
their analysis does not allow for correlation of regressors and parameters
or dynamics in the model. Additionally, their combination weights are determined
from in-sample test statistics rather than the expected out-of-sample performance
that we propose. In this sense, our approach is closer to the forecast-based
test for a structural break of \cite{Pesetal2013} and \citet{BooPic2020},
where the target is also significant improvements in forecast accuracy
rather than a significant change in parameters.

\citet{Liuetal2020} study forecasting for dynamic panel data
models with a short time-series dimension. Though $T$ exceeds the number
of parameters that have to be estimated for each series, such estimates
are typically very noisy and not consistent under large $N$, fixed
$T$ asymptotics. To handle estimation noise, they adopt a nonparametric
Bayesian approach that shrinks the heterogeneous parameters towards local
patterns in the distribution. This is closely related to the idea of using
forecast combinations to reduce the effect on the forecasts of noisy estimates
of individual-specific parameters.

\paragraph*{Outline} The rest of the paper is organized as follows. Section~\ref{SETUP} introduces the model setup and our assumptions, while Section~\ref{sec:theory_msfe} derives
analytical results on the predictive accuracy of individual, pooled, and
FE forecasting schemes. Section~\ref{sec:combinations} introduces our forecast combination schemes.
Section~\ref{sec:MC} describes the empirical Bayes estimator. Section~\ref{sec:applications} presents Monte
Carlo experiments, Section~7 reports our empirical applications,
and Section~8 concludes. Technical details are provided in appendices at the end of
the paper and in the Supplemental Appendix.


\section{Setup and assumptions\label{SETUP}}

We begin by describing the panel regression setup and assumptions used
in our analysis.

\subsection{Panel regression model\label{PanelReg}}

Our analysis considers the following linear panel regression model:
\begin{equation}
y_{it}=\alpha _{i}+\boldsymbol{\beta
}_{i}^{\prime }\boldsymbol{x}%
_{it}+
\varepsilon _{it}=\boldsymbol{\theta }_{i}^{\prime }
\boldsymbol{w}%
_{it}+\varepsilon _{it},\quad
\varepsilon _{it}\sim \bigl(0,\sigma _{i}^{2}
\bigr), \label{eq:model1}
\end{equation}
where $i=1,2,\ldots ,N$ refers to the individual units and
$t=1,2,\ldots ,T$ refers to the time period, $y_{it}$ is the outcome of
unit $i$ at time $t$, $\boldsymbol{x}_{it}$ is a $k\times 1$ vector of 
regressors---or predictors---used
to forecast $y_{it}$ (including, possibly, latent factors),
$\boldsymbol{\beta }_{i}$ is the associated vector of regression coefficients,
and $\varepsilon _{it}$ is the disturbance of unit $i$ in period $t$. The
second equality in (\ref{eq:model1}) introduces the notation
$\boldsymbol{\theta }_{i}=(\alpha _{i},\boldsymbol{\beta }%
_{i}^{\prime })^{\prime }$ and
$\boldsymbol{w}_{it}=(1,\boldsymbol{x}%
_{it}^{\prime })^{\prime }$, which have dimensions $K\times 1$, with
$K=k+1$. For simplicity, we use the time subscript $t$ for
$\boldsymbol{x}_{it}$ and $%
\boldsymbol{w}_{it}$, but it is important to emphasize that this refers
to the predicted time for the outcome variable, $y_{it}$. For a forecast
horizon of $h$ periods, all variables in $\boldsymbol{x}_{it}$ must therefore
be known at time $t-h$. Our notation avoids explicitly referring to
$h$ everywhere, but it should be recalled throughout the analysis that
$\boldsymbol{x}_{it}$ includes suitably lagged predictors. We will focus
on the case of $h=1$ but extensions to larger $h$ are straightforward.


\paragraph*{Notation} Stacking the time series of outcomes, regressors, and
disturbances, define
$\boldsymbol{y}_{i}=(y_{i1},y_{i2},\ldots ,y_{iT})^{\prime }$,
$\boldsymbol{X}_{i}=(\boldsymbol{x}_{i1},%
\boldsymbol{x}_{i2},\ldots ,\boldsymbol{x}_{iT})^{
\prime }$,
$\boldsymbol{W}_{i}=  ( \boldsymbol{\tau }_{T},\boldsymbol{X}%
_{i}  ) $, where $\boldsymbol{\tau }_{T}$ is a $T\times 1$ vector
of ones, and
$\boldsymbol{\varepsilon }_{i}=(\varepsilon _{i1},\varepsilon _{i2},
\ldots ,\varepsilon _{iT})^{\prime }$. Further, let
$\boldsymbol{y}=(%
\boldsymbol{y}_{1}^{\prime },\boldsymbol{y}_{2}^{\prime },\ldots ,%
\boldsymbol{y}_{N}^{\prime })^{\prime }$,
$\boldsymbol{X}=(\boldsymbol{X}%
_{1}^{\prime },\boldsymbol{X}_{2}^{\prime },\ldots ,\boldsymbol{X}%
_{N}^{\prime })^{\prime }$,
$\boldsymbol{W}=(\boldsymbol{W}%
_{1}^{\prime },\boldsymbol{W}_{2}^{\prime },\ldots ,\boldsymbol{W}%
_{N}^{\prime })^{\prime }$, and
$\boldsymbol{\varepsilon }=(\boldsymbol{%
\varepsilon }_{1}^{\prime },\boldsymbol{\varepsilon }_{2}^{\prime },
\ldots ,\boldsymbol{\varepsilon }_{N}^{\prime })^{\prime }$. Generic positive finite
constants are denoted by $C$ when large and $c$ when small. They can take
different values at different instances.
$\lambda _{\max }  ( \boldsymbol{%
A}  ) $ and $\lambda _{\min }  ( \boldsymbol{A}  ) $ denote
the maximum and minimum eigenvalues of matrix $\boldsymbol{A}$.
$\boldsymbol{A}%
\succ \boldsymbol{0}$ and $\boldsymbol{A}\succeq \boldsymbol{0}$ denote
that $\boldsymbol{A}$ is a positive definite and a nonnegative definite
matrix, respectively.
$ \llVert  \boldsymbol{A} \rrVert  =\lambda _{\max }^{1/2}(%
\boldsymbol{A}^{\prime }\boldsymbol{A)}$ and
$ \llVert  \boldsymbol{A}%
 \rrVert  _{1}$ denote the spectral and column sum norms of matrix
$%
\boldsymbol{A}$, respectively,
$ \llVert  \boldsymbol{x} \rrVert  _{p}=%
  [ \mathrm{E}  (  \llVert  \boldsymbol{x} \rrVert  ^{p}
  )   ] ^{1/p}$. If
$  \{ f_{n}  \} _{n=1}^{\infty }$ is any real sequence and
$  \{ g_{n}  \} _{n=1}^{\infty }$ is a sequence of positive real
numbers, then $f_{n}=O(g_{n})$, if there exists a $C$ such that
$ \llvert  f_{n} \rrvert  /g_{n}\leq C$ for all $n$ and
$f_{n}=o(g_{n})$ if $f_{n}/g_{n}\rightarrow 0$ as
$n\rightarrow \infty $. Similarly, $f_{n}=O_{p}(g_{n})$ if
$f_{n}/g_{n}$ is stochastically bounded and $f_{n}=o_{p}(g_{n})$ if
$f_{n}/g_{n}\overset{p}{\rightarrow}0$. The operator
$\overset{p}{\rightarrow}$ denotes convergence in probability, and
$\overset{d}{\rightarrow}$ denotes convergence in distribution.

\subsection{Assumptions}

Our theoretical analysis builds on a set of standard assumptions about
the underlying data generating process.

\begin{assumption}
\label{ass:1}
$\varepsilon _{it}$ is serially independent with mean zero, a fixed variance
$\sigma _{i}^{2}$ $(0<c<\sigma _{i}^{2}<C<\infty)$, and with
$\sup_{i,t}\mathrm{E} \llvert  \varepsilon _{it} \rrvert  ^{4}<C<
\infty $.
\end{assumption}

\begin{assumption}
\label{ass:weak_exogeneity_a}
$  \{ \varepsilon _{it}  \} $ for $i$ $%
=1,2,\ldots ,N$ are martingale difference processes with respect to the
filtration,
$\mathcal{I}_{it}=  ( \boldsymbol{w}_{it},\boldsymbol{w}%
_{i,t-1},\ldots   ) $, so that
\begin{equation*}
\mathrm{E}\lleft( \varepsilon _{it}\vert \boldsymbol{w}_{is}
\rright. ) =0,\quad \text{for }t\geq s,\text{ for }t=1,2,\dots ,T,T+1.
\end{equation*}
\end{assumption}

\begin{assumption}
\label{ass:3}
(a) $  \{ \boldsymbol{w}_{it}  \} $ for $i=1,2,\ldots ,N $ are
covariance stationary with
$\mathrm{E}(\boldsymbol{w}_{it}%
\boldsymbol{w}_{it}^{\prime })=\boldsymbol{Q}_{i}$,
$\sup_{i,t=  \{ 1,2,\ldots ,T  \} }\mathrm{E} \llVert
\boldsymbol{w}_{it} \rrVert  ^{4}<C$,
$\sup_{i,T} \llVert  \boldsymbol{w}_{i,T+1} \rrVert  <C$, and
%
\begin{equation}
\sup_{i}\mathrm{\lambda }_{\max } (
\boldsymbol{Q}_{i } ) <C<\infty,\quad \text{and}\quad \sup
_{i}\mathrm{\lambda }_{\max } \bigl(
\boldsymbol{Q}%
_{i }^{-1} \bigr) <C<\infty
. \label{CQi}
\end{equation}%
(b) The sample covariance matrices
$\boldsymbol{Q}_{iT }=T^{-1}\boldsymbol{W%
}_{i}^{\prime }\boldsymbol{W}_{i}=T^{-1}\sum_{t=1}^{T}\boldsymbol{w}_{it}%
\boldsymbol{w}_{it}^{\prime }$, for $i=1,2,\ldots ,N$, satisfy the conditions
$\sup_{i}\mathrm{\lambda }_{\max }  ( \boldsymbol{Q}_{iT
}  ) <C<\infty $, and
$\sup_{i}\mathrm{\lambda }_{\max }  ( \boldsymbol{Q}_{iT }^{-1}
  ) <C<\infty $.
\end{assumption}

\begin{assumption}
\label{ass:weak_exogeneity_b}
There exists a fixed $T_{0}$ such that for all $T>T_{0}$,%
%
\begin{align}
\sup_{i}\mathrm{E} \bigl\llVert T^{-1/2}
\boldsymbol{W}_{i}^{\prime } \boldsymbol{%
\varepsilon
}_{i} \bigr\rrVert ^{4}&<C<\infty , \label{CWe}%
\\
\sup_{i}\mathrm{E} \bigl[ \lambda _{\max }^{4}
( \boldsymbol{Q}_{iT
} ) \bigr] &<C<\infty, \quad \text{and}\quad \sup
_{i}\mathrm{E} \bigl[ \lambda _{\max }^{4}
\bigl( \boldsymbol{Q}_{iT }^{-1} \bigr) \bigr] <C<\infty .
\label{CQiT}%
\end{align}
\end{assumption}

Under Assumption~\ref{ass:1}, the optimal forecast of $y_{i,T+1}$, in a
mean squared error sense, is given by
$\mathrm{E}  ( y_{i,T+1} \llvert  \boldsymbol{w}_{i,T+1},
\boldsymbol{W}_{i}    ) =\boldsymbol{\theta }_{i}^{
\prime }\boldsymbol{w}_{i.T+1}$. Note that $\boldsymbol{w}_{i.T+1}$ is
known at time $T$, and is bounded under Assumption~\ref{ass:3}. Assumption~\ref{ass:weak_exogeneity_a} allows the regressors to be weakly exogenous
with respect to $\boldsymbol{\varepsilon }_{i}$ and, therefore, permits
the inclusion of lagged dependent variables such as $y_{i,T}$ in
$%
\boldsymbol{w}_{i,T+1}$. Part (a) of Assumption~\ref{ass:3} is standard
in the forecasting literature and requires the regressors to be stationary.
Part (b) is an identification assumption that allows estimation of individual
slope coefficients, $\boldsymbol{\theta }_{i}$, by least squares. Assumption~\ref{ass:weak_exogeneity_b} is required when we compare average MSFEs based
on individual and pooled estimators. It provides sufficient conditions
under which (see Lemma~\ref{Lemma_2_VTEX1})
%
\begin{equation}
\mathrm{E} \bigl\llVert \sqrt{T} ( \hat{\boldsymbol{\theta }_{i}}-%
\boldsymbol{\theta }_{i} ) \bigr\rrVert ^{2}=\mathrm{E}
\bigl\llVert \boldsymbol{Q}_{iT }^{-1} \bigl(
T^{-1/2}\boldsymbol{W}_{i}^{
\prime }%
\boldsymbol{\varepsilon }_{i} \bigr) \bigr\rrVert ^{2}<C<
\infty , \label{Momentconthetai}
\end{equation}%
where
$\hat{\boldsymbol{\theta }_{i}}=  ( \boldsymbol{W}_{i}^{\prime }%
\boldsymbol{W}_{i}  ) ^{-1}\boldsymbol{W}_{i}^{\prime }
\boldsymbol{y}%
_{i} $ is the least squares estimator of $\boldsymbol{\theta }_{i}$. The
moment conditions in Assumption~\ref{ass:weak_exogeneity_b} can be relaxed
when $\boldsymbol{w}_{it}$ is strictly exogenous.

Under weakly exogenous regressors the least squares estimator has a small
$T$ bias, and
$\mathrm{E}  ( \hat{\boldsymbol{\theta }}_{i}-
\boldsymbol{%
\theta }_{i}  ) =O  ( T^{-1}  ) $. Under strictly exogenous
regressors, in contrast,
$\mathrm{E}  ( \hat{\boldsymbol{\theta }_{i}}-%
\boldsymbol{\theta }_{i}  ) =\mathbf{0}$. We also note that, under
Assumptions \ref{ass:3} and \ref{ass:weak_exogeneity_b},
$ \llVert  \boldsymbol{Q}_{iT }-\boldsymbol{Q}_{i} \rrVert  =O_{p}(T^{-1/2})$,
and
$%
 \llVert  \boldsymbol{Q}_{iT }^{-1}-\boldsymbol{Q}_{i}^{-1} \rrVert  =O_{p}(T^{-1/2})$. These results, which hold for each $i$, are used
in the implementation of our combination forecasts below. For proof of
consistency of the weights in the combined forecasts discussed in Section~\ref{sec:combinations}
below, we need the stronger conditions
%
\begin{equation}
\sup_{i} \llVert \boldsymbol{Q}_{iT }-
\boldsymbol{Q}_{i} \rrVert =O_{p} \biggl( \frac{\ln (N)}{
\sqrt{T}} \biggr) ,\quad  \text{and}\quad  \sup_{i} \bigl\llVert
\boldsymbol{Q}_{iT }^{-1}-\boldsymbol{Q}_{i}^{-1}
\bigr\rrVert =O_{p} \biggl( \frac{\ln (N)}{\sqrt{T}} \biggr) ,
\label{supQinv}
\end{equation}%
still allowing $N$ to rise much faster than $T$.\footnote{As noted by
 \citeauthor{Fanetal2015} (\citeyear{Fanetal2015}, Section~\ref{ForecastsIP}), this stronger condition is
typically satisfied for strictly stationary data that satisfy strong mixing
conditions.}

Finally, let
$\boldsymbol{g}_{it}=\boldsymbol{w}_{it}\varepsilon _{it}$, and note that
$T^{-1/2}\boldsymbol{W}_{i}^{\prime }\boldsymbol{\varepsilon }%
_{i}=T^{-1/2}\sum_{t=1}^{T}\boldsymbol{g}_{it}$. Also, under Assumption~\ref{ass:weak_exogeneity_a}
$\boldsymbol{g}_{it}$ is a martingale difference process with respect to
$\mathcal{I}_{it}=  ( \boldsymbol{w}_{it},%
\boldsymbol{w}_{i,t-1},\ldots   ) $, and we have
$\mathrm{E}  ( \boldsymbol{g}_{it}  ) =\mathbf{0}$,
\begin{equation*}
\operatorname{Var} \Biggl( T^{-1/2}\sum_{t=1}^{T}
\boldsymbol{g}_{it} \Biggr) =T^{-1} \sum
_{t=1}^{T}\mathrm{E} \bigl( \boldsymbol{g}_{it}
\boldsymbol{g}%
_{it}^{\prime } \bigr)
=T^{-1}\mathrm{E} \bigl( \boldsymbol{W}_{i}^{
\prime }%
\boldsymbol{\varepsilon }_{i}\boldsymbol{\varepsilon
}_{i}^{\prime }%
\boldsymbol{W}_{i}
\bigr) =T^{-1}\sum_{t=1}^{T}
\sigma _{i}^{2} \mathrm{E}%
 \bigl(
\boldsymbol{w}_{it}\boldsymbol{w}_{it}^{\prime }
\bigr) .
\end{equation*}%
Further, under Assumption~\ref{ass:3},
$\mathrm{E}  ( \boldsymbol{w}_{it}%
\boldsymbol{w}_{it}^{\prime }  ) =\boldsymbol{Q}_{i}$, and it follows
that
%
\begin{equation}
\mathrm{E} \bigl( T^{-1}\boldsymbol{W}_{i}^{\prime }
\boldsymbol{\varepsilon }%
_{i}\boldsymbol{\varepsilon
}_{i}^{\prime }\boldsymbol{W}_{i} \bigr) = \sigma
_{i}^{2}\boldsymbol{Q}_{i}.
\label{Gi}
\end{equation}

We next introduce assumptions that are required primarily for establishing
the properties of pooled and fixed effects predictors.

\begin{assumption}
\label{ass:4}
(a)
$\boldsymbol{\theta}_{i}=\boldsymbol{\theta}+\boldsymbol{%
\eta }_{i}$ with $ \llVert  \boldsymbol{\theta } \rrVert  <C$,
$\mathrm{E
} \llVert  \boldsymbol{\eta }_{i} \rrVert  <C$,
$\mathrm{E}  ( \boldsymbol{\eta }_{i}  ) =0$,
$\mathrm{E}  ( \boldsymbol{\eta }_{i}%
\boldsymbol{\eta }_{i}^{\prime }  ) =$
$\boldsymbol{\Omega }_{\eta }$, and
$ \llVert  \boldsymbol{\Omega }_{\eta } \rrVert  <C$. (b) Let
$%
\boldsymbol{q}_{it}=\boldsymbol{w}_{it}\boldsymbol{w}_{it}^{\prime }%
\boldsymbol{\eta }_{i}$, then
$\mathrm{E}  ( \boldsymbol{q}_{it}  ) =%
\boldsymbol{q}_{i}$ (fixed),
$\sup_{i} \llVert  \boldsymbol{q}%
_{i} \rrVert  <C$,
$\sup_{i,t}\mathrm{E} \llVert  \boldsymbol{q}%
_{it} \rrVert  ^{2}<C$, and
$\sup_{i}\mathrm{E } \llVert  \boldsymbol{w}%
_{i,T+1}^{\prime }\boldsymbol{\eta }_{i} \rrVert  ^{2}<C$.
\end{assumption}

\begin{assumption}
\label{ass:5}
$\boldsymbol{\eta }_{i}$ is distributed independently of
$%
\boldsymbol{\varepsilon }_{i}$, for all $i$.
\end{assumption}

\begin{assumption}
\label{ass:2.2}
$ \bar{\boldsymbol{\xi }}_{NT}=N^{-1}\sum_{i=1}^{N}%
\boldsymbol{\xi }_{iT}=O_{p}  ( N^{-1/2}T^{-1/2}  ) $, where
$%
\boldsymbol{\xi }_{iT }=T^{-1}\boldsymbol{W}_{i}^{\prime }
\boldsymbol{%
\varepsilon }_{i}=T^{-1}\sum_{t=1}^{T}\boldsymbol{w}_{it}
\varepsilon _{it}=O_{p}(T^{-1/2})$.
\end{assumption}

\begin{assumption}
\label{ass:3b}%
There exists a fixed $T_{0}$ such that for all $T>T_{0}$ and
$%
N=1,2,\ldots $, the pooled covariance matrices
$\boldsymbol{\bar{Q}}_{NT}$ and $\boldsymbol{\bar{Q}}_{N}$, defined in
terms of
$\boldsymbol{Q}_{iT
}=T^{-1}\boldsymbol{W}_{i}^{\prime }\boldsymbol{W}_{i}$ and
$\boldsymbol{Q}%
_{i}=\mathrm{E}  ( \boldsymbol{Q}_{iT }  ) $,
%
\begin{equation}
\boldsymbol{\bar{Q}}_{NT}=N^{-1}\sum
_{i=1}^{N}\boldsymbol{Q}_{iT}
, \quad \text{and}\quad \boldsymbol{\bar{Q}}_{N}=\mathrm{E} ( \boldsymbol{
\bar{Q}}_{NT} ) =N^{-1}\sum_{i=1}^{N}
\boldsymbol{Q}_{i}, \label{Qbar}
\end{equation}%
are positive definite,
$ \llVert  \boldsymbol{\bar{Q}}_{N}^{-1} \rrVert  <C$, and
\begin{equation*}
\sup_{N,T}\mathrm{E} \bigl[ \lambda _{\max }^{2}
( \boldsymbol{\bar{Q}}%
_{NT} ) \bigr] <C<\infty
,\quad  \text{and}\quad \sup_{N,T}\mathrm{E} \bigl[ \lambda
_{\max }^{2} \bigl( \boldsymbol{\bar{Q}}_{NT}^{-1}
\bigr) \bigr] <C<\infty .
\end{equation*}
\end{assumption}

\begin{assumption}
\label{ass:6}
$(\boldsymbol{\varepsilon }_{i},\boldsymbol{W}_{i},
\boldsymbol{%
\eta }_{i})$ are distributed independently over $i$.
\end{assumption}

For pooled estimation of $\boldsymbol{\theta }$, the conditions on
$%
\boldsymbol{Q}_{iT}$ can be relaxed and it is sufficient that
$\bar{%
\boldsymbol{Q}}_{NT}$ is positive definite, and
$\sup_{N,T}\mathrm{E}%
 \llVert  \boldsymbol{Q}_{NT}^{-1} \rrVert  ^{2}<C$. Assumptions
\ref{ass:4} and \ref{ass:5} identify the population mean of
$\boldsymbol{\theta }%
_{i}$ denoted by $\boldsymbol{\theta }$, but allow for correlated heterogeneity.\footnote{We simplify the notation and use $\boldsymbol{\theta }$, rather than
$%
\boldsymbol{\theta }_{0}$, to denote the population mean, which is technically
more appropriate.} The degree of parameter heterogeneity is measured by
the norm of $\boldsymbol{\Omega }_{\eta }$, and the extent to which heterogeneity
is correlated is measured by the norm of $\boldsymbol{q}%
_{i}$.\footnote{Under Assumption~\ref{ass:weak_exogeneity_a},
$\mathrm{E}  ( \boldsymbol{%
\xi }_{iT}  ) =T^{-1}\sum_{t=1}^{T}\mathrm{E}  (
\boldsymbol{w}%
_{it}\varepsilon _{it}  ) =\boldsymbol{0}$, and
$\mathrm{E}  ( \boldsymbol{\xi }_{NT}  ) =\boldsymbol{0}$. Note
that $\varepsilon _{it}$ and $\boldsymbol{w}_{it}$ are uncorrelated but
not independently distributed. Under Assumption~\ref{ass:3},
$ \llVert  \boldsymbol{\bar{Q}}%
_{NT} \rrVert  \leq \sup_{i} \llVert  \boldsymbol{Q}_{iT}
 \rrVert  <C$, and
$ \llVert  \boldsymbol{\bar{Q}}_{N} \rrVert  \leq \sup_{i}
 \llVert  \boldsymbol{Q}_{i} \rrVert  <C$.}

Assumptions \ref{ass:4}--\ref{ass:3b} are not required for forecasts based
on the individual estimates and the associated MSFE. Assumption~\ref{ass:6} of cross-sectional independence for $\varepsilon _{it}$ (or
$\boldsymbol{w}_{it}$) is not needed to establish results on the MSFE of individual
forecasts. However, we do require some degree of uncorrelatedness over
$i$ when the objective is to compute the MSFE averaged across all
$N$ units under consideration or over a subgroup of the units. In particular,
to ensure that the cross-sectional average MSFE tends to a nonrandom limit,
the units under consideration must satisfy the law of large numbers. To
this end, we need the units to be cross-sectionally weakly correlated,
possibly conditional on known (or estimated) common factors. The situation
is different when we consider pooled or Bayesian forecasts. Optimality
of these forecasts \textit{does} depend on the assumption of cross-sectional
independence, or at least some form of weak cross-sectional dependence.
A comprehensive analysis of the implications of cross-sectional dependence
for forecast combinations and comparisons of predictive accuracy are beyond
the scope of the present paper, however.\footnote{Cross-sectional dependence in forecast errors can be exploited by using
interactive time effects (latent factors) or spatial (network) effects;
see, for example, \cite{Chuetal2016}.}

We measure the degree of correlated heterogeneity for unit $i$ at time
$t$ by
$\boldsymbol{q}_{i}=\mathrm{E}  ( \boldsymbol{w}_{it}
\boldsymbol{w}%
_{it}^{\prime }\boldsymbol{\eta }_{i}  ) $ and, on average, by%
%
\begin{equation}
\boldsymbol{\bar{q}}_{NT}=N^{-1}T^{-1}\sum
_{i=1}^{N}\boldsymbol{W}%
_{i}^{\prime }
\boldsymbol{W}_{i}\boldsymbol{\eta }_{i}=N^{-1}T^{-1}%
\sum_{i=1}^{N}\sum
_{t=1}^{T}\boldsymbol{w}_{it}
\boldsymbol{w}_{it}^{
\prime }%
\boldsymbol{\eta
}_{i}. \label{qbarNT}
\end{equation}%
Taking expectations,%
%
\begin{equation}
\mathrm{E} ( \boldsymbol{\bar{q}}_{NT} ) = \boldsymbol{
\bar{q}}%
_{N}=N^{-1}\sum
_{i=1}^{N}\boldsymbol{q}_{i}.
\label{qbar}
\end{equation}%
Assumptions \ref{ass:4} and \ref{ass:5} accommodate correlated heterogeneity
and allow for nonzero values of
$\mathrm{E}  ( \boldsymbol{W}%
_{i}^{\prime }\boldsymbol{W}_{i}\boldsymbol{\eta }_{i}  ) $. In the
context of fixed effects models, the intercepts $\alpha _{i}$ in (\ref{eq:model1})
are allowed to have nonzero correlation with the regressors, but optimality
of forecasts based on pooled estimates of $\boldsymbol{\beta } $ requires
Assumption~\ref{ass:5} and the condition
$\lim_{n\rightarrow \infty }n^{-1}\*\sum_{i=1}^{n}\mathrm{E}  (
\boldsymbol{X}_{i}^{\prime }%
\boldsymbol{M}_{T}\boldsymbol{X}_{i}\boldsymbol{\eta }_{i\beta }
  ) =%
\boldsymbol{0}$, where
$\boldsymbol{\eta }_{i\beta }=\boldsymbol{\beta }_{i}-%
\boldsymbol{\beta }$,
$\boldsymbol{M}_{T}=\boldsymbol{I}_{T}-\boldsymbol{%
\tau }_{T}  ( \boldsymbol{\tau }_{T}^{\prime }
\boldsymbol{\tau }%
_{T}  ) ^{-1}\boldsymbol{\tau }_{T}^{\prime }$,
$\boldsymbol{\tau }_{T}$ is a $T\times 1$ vector of ones, and
$\boldsymbol{I}_{T}$ is a $T\times T$ identity matrix.\footnote{See \citet{Pesetal2024b}. Note that
$\mathrm{E}  ( \boldsymbol{X}%
_{i}^{\prime }\boldsymbol{M}_{T}\boldsymbol{X}_{i}\boldsymbol{\eta }_{i
\beta }  ) =\boldsymbol{0}$ is sufficient but not necessary for the
validity of fixed effects estimation. This condition is not met if
$\boldsymbol{x}%
_{it}$ includes lagged values of $y_{it}$, even if
$T\rightarrow \infty $.}

\section{Theoretical results on forecasting performance}

\label{sec:theory_msfe}

We next use the setup and assumptions from Section~\ref{SETUP} to establish
theoretical results on the forecasting performance of different modeling
approaches. Section~\ref{ForecastsIP} discusses forecasts based on individual
and pooled estimation, and building on this, Section~\ref{sec:FE_theory}
covers fixed effects forecasts.

Note that our theoretical framework can be equally applied to forecasts
across groups instead of individuals, when there are \textit{a priori} known groups
such as industries or states within a given country. Pooled regressions
can be applied to any given, \textit{a priori} known group, so long as the number
of units within the group is sufficiently large and the cross-sectional
dependence of units within the group is sufficiently weak. Failure of the
latter assumption implies that there are missing pervasive (strong) common
factors that must also be taken into account but such an extension lies beyond the scope
of the present paper.

\subsection{Forecasts based on individual and pooled estimation\label%
{ForecastsIP}}

We are interested in forecasting $y_{i,T+1}$ conditional on the information
known at time $T$, which we denote by $\boldsymbol{w}_{i,T+1}$ to clarify
the correspondence to $y_{i,T+1}$. Without loss of generality, given the
conditional nature of the forecasting exercise, we assume that $%
\sup_{i,T}\left\Vert \boldsymbol{w}_{i,T+1}\right\Vert <C$.\footnote{%
See part (a) of Assumption \ref{ass:3}.} Forecasts based on individual
estimators take the form 
\begin{equation}
\hat{y}_{i,T+1}=\hat{\boldsymbol{\theta }}_{i}^{\prime }\boldsymbol{w}%
_{i,T+1},\text{ }i=1,2,\ldots ,N,  \label{eq:indFore}
\end{equation}%
where $\hat{\boldsymbol{\theta }}_{i}=(\boldsymbol{W}_{i}^{\prime }%
\boldsymbol{W}_{i})^{-1}\boldsymbol{W}_{i}^{\prime }\boldsymbol{y}_{i},$ is
the least squares estimator of $\boldsymbol{\theta }_{i}$. Similarly,
forecasts based on the pooled estimator are given by%
\begin{equation}
\tilde{y}_{i,T+1}=\tilde{\boldsymbol{\theta }}^{\prime }\boldsymbol{w}%
_{i,T+1},\text{ }i=1,2,\ldots ,N,  \label{eq:poolFore}
\end{equation}%
where $\tilde{\boldsymbol{\theta }}=(\boldsymbol{W}^{\prime }\boldsymbol{W}%
)^{-1}\boldsymbol{W}^{\prime }\boldsymbol{y}$. Using (\ref{Qbar}), (\ref%
{qbarNT}) and the definition of $\boldsymbol{\bar{\xi}}_{NT}$ in Assumption %
\ref{ass:2.2}, 
\begin{equation}
\tilde{\boldsymbol{\theta }}-\boldsymbol{\theta }_{i}=-\boldsymbol{\eta }%
_{i}+\boldsymbol{\bar{Q}}_{NT}^{-1}\boldsymbol{\bar{q}}_{NT}+\boldsymbol{%
\bar{Q}}_{NT}^{-1}\boldsymbol{\bar{\xi}}_{NT}.  \label{bpooled}
\end{equation}%
Forecast errors from these schemes take the form%
\begin{eqnarray}
\hat{e}_{i,T+1} &=&y_{iT+1}-\hat{y}_{i,T+1}=\varepsilon _{i,T+1}-(\hat{%
\boldsymbol{\theta }_{i}}-\boldsymbol{\theta }_{i})^{\prime }\boldsymbol{w}%
_{i,T+1},  \label{ehat} \\
\tilde{e}_{i,T+1} &=&y_{iT+1}-\tilde{y}_{i,T+1}=\varepsilon _{i,T+1}-(\tilde{%
\boldsymbol{\theta }}-\boldsymbol{\theta }_{i})^{\prime }\boldsymbol{w}%
_{i,T+1}.  \label{etilda}
\end{eqnarray}

\subsubsection*{Forecasts based on individual estimation}

Noting that
$(\hat{\boldsymbol{\theta }}_{i}-\boldsymbol{\theta }%
_{i})^{\prime }\boldsymbol{w}_{i,T+1}=\boldsymbol{\varepsilon }_{i}^{
\prime }%
\boldsymbol{W}_{i}(\boldsymbol{W}_{i}^{\prime }\*\boldsymbol{W}_{i})^{-1}%
\boldsymbol{w}_{i,T+1}$, it is easily seen that the forecasts based on
the individual estimates generate the following average MSFE:%
%
\begin{equation}
N^{-1}\sum_{i=1}^{N}
\hat{e}_{i,T+1}^{2}=N^{-1}\sum
_{i=1}^{N} \varepsilon _{i,T+1}^{2}+T^{-1}S_{NT}-2R_{NT},
\label{MSFE_i}
\end{equation}%
where $S_{NT}=N^{-1}\sum_{i=1}^{N}s_{iT}$,
$R_{NT}=N^{-1}\sum_{i=1}^{N}r_{iT%
 }$, with elements%
%
\begin{equation}
r_{iT }= \bigl( \boldsymbol{\varepsilon }_{i}^{\prime }
\boldsymbol{W}_{i}\bigl(%
\boldsymbol{W}_{i}^{\prime }
\boldsymbol{W}_{i}\bigr)^{-1}\boldsymbol{w}%
_{i,T+1}
\bigr) \varepsilon _{i,T+1},  \label{r_iT}%
\end{equation}
and
\begin{equation}
s_{iT}=\boldsymbol{w}_{i,T+1}^{\prime }
\boldsymbol{Q}_{iT}^{-1} \bigl( T^{-1}%
\boldsymbol{W}_{i}^{\prime }\boldsymbol{\varepsilon
}_{i} \boldsymbol{%
\varepsilon }_{i}^{\prime }
\boldsymbol{W}_{i} \bigr) \boldsymbol{Q}_{iT}^{-1}%
\boldsymbol{w}_{i,T+1}. \label{s_iT}%
\end{equation}
Under Assumptions \ref{ass:1} and \ref{ass:3},
$\mathrm{E}  ( r_{iT
}  ) =0$ and
$\sup_{i,T}\mathrm{E} \llvert  r_{iT} \rrvert  <C$, and under cross-sectional
independence (Assumption~\ref{ass:6}) we have
$%
R_{NT}=O_{p}(N^{-1/2})$. Similarly,
$\sup_{i,T}\mathrm{E} \llvert  s_{iT} \rrvert  <C$,
\begin{equation*}
\mathrm{E} ( s_{iT } ) =\mathrm{E} \biggl[ \boldsymbol{w}%
_{i,T+1}^{\prime }
\boldsymbol{Q}_{iT}^{-1} \biggl( \frac{
\boldsymbol{W}%
_{i}^{\prime }\boldsymbol{\varepsilon
}_{i}\boldsymbol{\varepsilon }%
_{i}^{\prime }
\boldsymbol{W}_{i}}{T} \biggr) \boldsymbol{Q}_{iT}^{-1}%
\boldsymbol{w}_{i,T+1} \biggr] ,
\end{equation*}%
$S_{NT}=\mathrm{E}  ( S_{NT}  ) +O_{p}(N^{-1/2})$, and we obtain
the results summarized in the following proposition for the average MSFE
of the forecasts based on the individual estimates (for a detailed proof,
see Section~\ref{Proof_individual} of the Appendix):

\begin{proposition}
\label{prop:individual}
\begin{enumerate}
\item[(a)] Suppose that Assumptions \ref{ass:1}--\ref{ass:weak_exogeneity_b} and
\ref{ass:6} hold. Then, for a fixed $T_{0}$ such that $T>T_{0}$, the average
MSFE resulting from individual-specific estimation of the parameters, given
by (\ref{MSFE_i}), has the following representation:
%
\begin{equation}
N^{-1}\sum_{i=1}^{N}
\hat{e}_{i,T+1}^{2}=N^{-1}\sum
_{i=1}^{N} \varepsilon _{i,T+1}^{2}+T^{-1}h_{NT}+O_{p}
\bigl(N^{-1/2}\bigr), \label{MSFEI}
\end{equation}%
where
%
\begin{equation}
h_{NT}=N^{-1}\sum_{i=1}^{N}
\mathrm{E} \biggl[ \boldsymbol{w}_{i,T+1}^{
\prime }%
\boldsymbol{Q}_{iT}^{-1} \biggl( \frac{
\boldsymbol{W}_{i}^{\prime }\boldsymbol{%
\varepsilon
}_{i}\boldsymbol{\varepsilon }_{i}^{\prime }
\boldsymbol{W}_{i}}{T%
} \biggr) \boldsymbol{Q}_{iT}^{-1}
\boldsymbol{w}_{i,T+1} \biggr] , \label{h_NT}
\end{equation}%
$\boldsymbol{Q}_{iT}=T^{-1}\boldsymbol{W}_{i}^{\prime }\boldsymbol{W}_{i}$,
$%
h_{NT}>0$, and $h_{NT}=O(1)$.\vadjust{\goodbreak}

\item[(b)] If $\boldsymbol{W}_{i}$ is strictly exogenous, $h_{NT}$ simplifies
to
$ h_{NT}=N^{-1}\sum_{i=1}^{N}\sigma _{i}^{2}\mathrm{E}  (
\boldsymbol{w}%
_{i,T+1}^{\prime }\boldsymbol{Q}_{iT}^{-1}\*\boldsymbol{w}_{i,T+1}
  )  $.
\end{enumerate}
\end{proposition}

The $h_{NT}$ term captures the cost associated with the error in estimation
of $\hat{\boldsymbol{\theta }}_{i}$. For typical panel data sets,
$T$ is not large and parameter estimation uncertainty captured by the
$O  ( T^{-1}  ) $ term $T^{-1}h_{NT}$ in (\ref{MSFEI}) can therefore
be important. Parameter heterogeneity, in contrast, does not affect the
accuracy of the forecasts in (\ref{MSFEI}). The magnitude of $h_{NT}$ plays
an important role in the comparisons of forecasts based on individual and
pooled estimates and depends on how far the predictors are from their mean.
For example, when $\boldsymbol{w}_{it}=(1,x_{it})^{\prime }$ and
$x_{it}$ is strictly exogenous,
\begin{equation*}
h_{NT}=\bar{\sigma}_{N}^{2}+N^{-1}
\sum_{i=1}^{N}\sigma
_{i}^{2} \text{E} \biggl[ \frac{ (
x_{i,T+1}-\bar{x}_{iT} ) ^{2}}{s_{iT}^{2}}
\biggr] ,
\end{equation*}%
where $\bar{\sigma}_{N}^{2}=N^{-1}\sum_{i=1}^{N}\sigma _{i}^{2}$,
$%
s_{iT}^{2}=T^{-1}\sum_{t=1}^{T}(x_{it}-\bar{x}_{iT})^{2}$, and
$\bar{x}%
_{iT}=T^{-1}\sum_{t=1}^{T}x_{it}$. Hence, $h_{NT}$ is minimized when
$%
x_{i,T+1}=\bar{x}_{iT}$, for all $i$. When
$x_{i,T+1}\neq \bar{x}_{iT}$ for most $i$, $T$ must be sufficiently large
such that $\sup_{i} \mathrm{E}  [   ( x_{i,T+1}-\bar{x}_{iT}  ) ^{2}/s_{iT}^{2}
  ] <C$.

\subsubsection*{Forecasts based on pooled estimation}

While the forecast accuracy results for the individual regressions do not
depend on the degree of parameter heterogeneity, whether correlated or
not, the degree of correlated heterogeneity does matter for consistency
of the pooled estimator. Using (\ref{bpooled}) in (\ref{etilda}), we can
express the squared forecast error when pooled estimates are used as follows:
\begin{equation*}
\tilde{e}_{i,T+1}^{2}=\varepsilon _{i,T+1}^{2}+
\boldsymbol{w}%
_{i,T+1}^{\prime }\boldsymbol{d}_{i,NT}
\boldsymbol{d}_{i,NT}^{\prime }%
\boldsymbol{w}_{i,T+1}-2
\boldsymbol{d}_{i,NT}^{\prime } \boldsymbol{w}%
_{i,T+1}
\varepsilon _{i,T+1},
\end{equation*}%
where
$\boldsymbol{d}_{i,NT}=-\boldsymbol{\eta }_{i}+\boldsymbol{\bar{Q}}%
_{NT}^{-1}\boldsymbol{\bar{q}}_{NT}+\boldsymbol{\bar{Q}}_{NT}^{-1}%
\boldsymbol{\bar{\xi}}_{NT}$, $\boldsymbol{\bar{Q}}_{NT}$, and
$\boldsymbol{%
\bar{q}}_{NT}$ are defined by (\ref{Qbar}) and (\ref{qbarNT}), and
$%
\boldsymbol{\bar{\xi}}_{NT}$ is defined under Assumption~\ref{ass:2.2}. After some algebra, and averaging over $i$, we have
%
\begin{equation}
N^{-1}\sum_{i=1}^{N}
\tilde{e}_{i,T+1}^{2} = N^{-1}\sum
_{i=1}^{N} \varepsilon _{i,T+1}^{2}+N^{-1}
\sum_{i=1}^{N}\boldsymbol{w}_{i,T+1}^{
\prime }
\boldsymbol{%
\eta }_{i}\boldsymbol{\eta
}_{i}^{\prime }\boldsymbol{w}_{i,T+1}
+\tilde{S}_{N,T+1}+2\tilde{R}_{N,T+1},\label{MSFEPool}
\end{equation}%
where $\tilde{S}_{N,T+1}$, and $\tilde{R}_{N,T+1}$ are defined by equations
(%
\ref{stilda}) and (\ref{rtilda}) in Section~\ref{Proof_pooled} of the Appendix.
It can be shown that $\tilde{R}_{N,T+1}=O_{p}(N^{-1/2})$, and
$%
\tilde{S}_{N,T+1}=-\boldsymbol{\bar{q}}_{N}^{\prime }
\boldsymbol{\bar{Q}}%
_{N}^{-1}\boldsymbol{\bar{q}}_{N}+O_{p}  ( N^{-1/2}  ) $. 

The limiting properties of the average MSFE based on pooled estimates are summarized
in the following proposition.

\begin{proposition}
\label{prop:pooled}
\begin{enumerate}
\item[(a)] Under Assumptions \ref{ass:1}--\ref{ass:6}, the MSFE for the forecasts
based on pooled estimation of the parameters, given by (\ref{MSFEPool}),
is%
%
\begin{equation}
N^{-1}\sum_{i=1}^{N}
\tilde{e}_{i,T+1}^{2}=N^{-1}\sum
_{i=1}^{N} \varepsilon _{i,T+1}^{2}+
\Delta _{NT}+O_{p}\bigl(N^{-1/2}\bigr),
\label{MSFEP1}
\end{equation}%
where%
%
\begin{equation}
\Delta _{NT}=N^{-1}\sum_{i=1}^{N}
\mathrm{E} \bigl( \boldsymbol{w}%
_{i,T+1}^{\prime }
\boldsymbol{\eta }_{i}\boldsymbol{\eta }_{i}^{
\prime }%
\boldsymbol{w}_{i,T+1} \bigr) -\boldsymbol{\bar{q}}_{N}^{\prime }
\boldsymbol{%
\bar{Q}}_{N}^{-1}\boldsymbol{
\bar{q}}_{N}. \label{Delta_NT}
\end{equation}%
\item[(b)] Parameter heterogeneity (whether correlated or uncorrelated) increases
the MSFE of the forecasts based on the pooled estimator, namely
$\Delta _{NT}>0$.
\end{enumerate}\end{proposition}

Note that the impact on the MSFE from neglected heterogeneity,
$\Delta _{NT}$%
, does not vanish even if both $N$ and $T\rightarrow \infty $, which is
similar to the finding by \citet{PesSmi1995} for heterogeneous dynamic
panels since heterogeneity is always correlated in dynamic panels.\footnote{This latter property is illustrated by a simple panel AR(1) model with
heterogeneous AR coefficients in Section~\ref{PanelAR} of the Appendix.
 See also \citet{Pesetal2024a} where estimation of such models with
short $T$ panels is considered.}

\subsubsection*{A comparison of forecasts based on individual and pooled
estimates}
Next, we consider the difference in the average MSFE performance of the
forecasts based on the pooled versus individual parameter estimates. Proposition~\ref{prop:individual}
shows that the MSFE from the forecasts based on the individual estimates
will be affected by an estimation error term of the form
\begin{equation*}
h_{NT}=N^{-1}\sum_{i=1}^{N}
\mathrm{E} \biggl[ \boldsymbol{w}_{i,T+1}^{
\prime }%
\boldsymbol{Q}_{iT}^{-1} \biggl( \frac{
\boldsymbol{W}_{i}^{\prime }\boldsymbol{%
\varepsilon
}_{i}\boldsymbol{\varepsilon }_{i}^{\prime }
\boldsymbol{W}_{i}}{T%
} \biggr) \boldsymbol{Q}_{iT}^{-1}
\boldsymbol{w}_{i,T+1} \biggr] >0.
\end{equation*}%
While the forecasts from the pooled estimates are more robust to estimation
errors, they are in turn affected by correlated and uncorrelated heterogeneity
as captured by the term
\begin{equation*}
\Delta _{NT}=N^{-1}\sum_{i=1}^{N}
\mathrm{E} \bigl( \boldsymbol{w}%
_{i,T+1}^{\prime }
\boldsymbol{\eta }_{i}\boldsymbol{\eta }_{i}^{
\prime }%
\boldsymbol{w}_{i,T+1} \bigr) -\boldsymbol{\bar{q}}_{N}^{\prime }
\boldsymbol{%
\bar{Q}}_{N}^{-1}\boldsymbol{
\bar{q}}_{N}.
\end{equation*}%
We compare the difference in the average MSFE of the forecasts from the
pooled versus individual estimates as a ratio measured relative to the
MSFE of the forecasts from the individual estimates:
\begin{equation*}
\frac{N^{-1}\sum_{i=1}^{N}
\tilde{e}_{i,T+1}^{2}-N^{-1}\sum
_{i=1}^{N} \hat{e}_{i,T+1}^{2}}{N^{-1}
\sum_{i=1}^{N}\hat{e}_{i,T+1}^{2}}=
\frac{\Delta _{NT} -T^{-1}h_{NT}+O_{p}
\bigl(N^{-1/2}\bigr)}{N^{-1}\sum
_{i=1}^{N}\varepsilon _{i,T+1}^{2}
+T^{-1}h_{NT}+O_{p}\bigl(N^{-1/2}
\bigr)}.
\end{equation*}%
Hence, there exists a $T_{0}$ such that, for a fixed $T>T_{0}$, and as
$N\rightarrow \infty  $,
%
\begin{equation}
\frac{N^{-1}\sum_{i=1}^{N}
\tilde{e}_{i,T+1}^{2}-N^{-1}\sum
_{i=1}^{N}\hat{e}%
_{i,T+1}^{2}}{N^{-1}
\sum_{i=1}^{N}\hat{e}_{i,T+1}^{2}}
\overset{p}{%
\rightarrow }\frac{\Delta -T^{-1}h_{T}}{
\bar{\sigma}^{2}+T^{-1}h_{T}},
\label{finiteT}
\end{equation}%
where $h_{T}=\lim_{N\rightarrow \infty }h_{NT}\geq 0$,
$\Delta =\lim_{N\rightarrow \infty }\Delta _{N}\geq 0$, and
$\bar{\sigma}%
^{2}=\lim_{N\rightarrow \infty }N^{-1}\sum_{i=1}^{N}\sigma _{i}^{2}>0$%
\textbf{.} It follows that when $T$ is fixed and $N$ is large, the ranking
of the two forecasting schemes will depend on the sign and magnitude of
$\Delta -T^{-1}h_{T}$.\footnote{In comparing $\Delta _{T}$ with $T^{-1}h_{T}$, it is also important to
bear in mind that $h_{T}$ is well-defined if moments of
$\boldsymbol{\hat{\theta}}%
_{i}$ (at least up to second order) exist (see the moment condition (\ref{Momentconthetai})).
This in turn requires that $T>T_{0}$ for some finite $%
T_{0}$. The value of $T_{0}$ depends on the nature of the
$(\boldsymbol{w}%
_{it},\varepsilon _{it}$) process and its distributional properties.}

For large values of $T$, however, the individual forecasts generate the
lowest MSFE values. Specifically, for a fixed $N$ and as
$T\rightarrow \infty  $,
\begin{equation*}
\frac{N^{-1}\sum_{i=1}^{N}
\tilde{e}_{i,T+1}^{2}-N^{-1}\sum
_{i=1}^{N}\hat{e}%
_{i,T+1}^{2}}{N^{-1}
\sum_{i=1}^{N}\hat{e}_{i,T+1}^{2}}
\overset{p}{%
\rightarrow }\frac{\Delta _{N}}{\bar{
\sigma}^{2}}+O_{p}\bigl(N^{-1/2}\bigr).
\end{equation*}%
Similarly, when both $N$ and $T\rightarrow \infty $ (in any order)
\begin{equation*}
\frac{N^{-1}\sum_{i=1}^{N}
\tilde{e}_{i,T+1}^{2}-N^{-1}\sum
_{i=1}^{N}\hat{e}%
_{i,T+1}^{2}}{N^{-1}
\sum_{i=1}^{N}\hat{e}_{i,T+1}^{2}}
\overset{p}{%
\rightarrow }\Delta /\bar{\sigma}^{2}\geq 0,
\end{equation*}
where $\Delta =\lim_{T\rightarrow \infty }(\Delta _{T})$. Therefore, vanishing
estimation uncertainty implied by large $T$ means that, on average, individual
forecasts are at least as precise as pooled forecasts irrespective of
$N$.

\subsection{Forecasts based on fixed effects estimation\label{sec:FE_theory}}
The comparison of forecasts based on individual or pooled estimates can
be extended to intermediate cases where a subset of the parameters are
allowed to vary across units. A prominent example is the FE forecast%
%
\begin{equation}
\hat{y}_{i,T+1}^{\text{FE}}=\hat{\alpha}_{i,\text{FE}}+
\boldsymbol{\hat{\beta%
}}_{\text{FE}}^{\prime }
\boldsymbol{x}_{i,T+1}, \label{ForFE}
\end{equation}%
where
$\hat{\alpha}_{i,\text{FE}}=\boldsymbol{\tau }_{T}^{\prime }(%
\boldsymbol{y}_{i}-\boldsymbol{\hat{\beta}}_{\text{FE}}^{\prime }
\boldsymbol{%
X}_{i})/T$ and
$\boldsymbol{\hat{\beta}}_{\text{FE}}=  ( \sum_{i=1}^{N}%
\boldsymbol{X}_{i}^{\prime }\boldsymbol{M}_{T}\boldsymbol{X}_{i}
  ) ^{-1}\sum_{i=1}^{N}\boldsymbol{X}_{i}^{\prime }
\boldsymbol{M}_{T}\boldsymbol{%
y}_{i}$. The associated FE forecast error is given by
%
\begin{equation}
\hat{e}_{i,T+1}^{\text{FE}}=\bar{\bar{\varepsilon}}_{i,T+1}-(
\hat{%
\boldsymbol{\beta }}_{\text{FE}}-\boldsymbol{\beta
}_{i})^{\prime } \bar{\bar{%
\boldsymbol{x}}}_{i,T+1},
\label{eiFE}
\end{equation}%
where
$\bar{\bar{\varepsilon}}_{i,T+1}=\varepsilon _{i,T+1}-
\bar{\varepsilon}%
_{iT}$,
$\bar{\bar{\boldsymbol{x}}}_{i,T+1}=\boldsymbol{x}_{i,T+1}-%
\boldsymbol{\bar{x}}_{iT}$, $\bar{\varepsilon}_{iT}$
$=T^{-1}\sum_{t=1}^{T}%
\varepsilon _{it}$, and
$\bar{\boldsymbol{x}}_{iT}=T^{-1}\*\sum_{t=1}^{T}%
\boldsymbol{x}_{it}$. Section S.5 in the Online Supplement provides details
of the derivation of the MSFE under fixed effects estimation:
%
\begin{equation}
N^{-1}\sum_{i=1}^{N}
\bigl( \hat{e}_{i,T+1}^{\text{FE}} \bigr) ^{2}=N^{-1}
\sum_{i=1}^{N}\bar{\bar{
\varepsilon}}_{i,T+1}^{2}+\Delta _{NT}^{%
\text{FE}}-2c_{NT}^{\text{FE}}+O_{p}
\bigl(N^{-1/2}\bigr), \label{FEmsfe}
\end{equation}%
where%
%
\begin{equation}
\Delta _{NT}^{\text{FE}}=N^{-1}\sum
_{i=1}^{N}\mathrm{E}\bigl( \bar{\bar{%
\boldsymbol{x}}}_{i,T+1}^{\prime }\boldsymbol{\eta
}_{i,\beta } \boldsymbol{%
\eta }_{i,\beta }^{\prime }
\bar{\bar{\boldsymbol{x}}}_{i,T+1}\bigr)- \bar{%
\boldsymbol{q}}_{N,\beta }^{\prime }\bar{\boldsymbol{Q}}_{N,\beta }^{-1}
\bar{%
\boldsymbol{q}}_{N,\beta }, \label{eq:Delta_FE}
\end{equation}%
$\boldsymbol{\eta }_{i,\beta }=\boldsymbol{\beta }_{i}-
\boldsymbol{\beta }$,
$\bar{\boldsymbol{\xi }}_{NT,\beta }=N^{-1}\sum_{i=1}^{N}T^{-1}
\boldsymbol{X}%
_{i}^{\prime }\boldsymbol{M}_{T}\boldsymbol{\varepsilon }_{i}$, $\bar{%
\boldsymbol{Q}}_{NT,\beta }=N^{-1}\sum_{i=1}^{N}T^{-1}
\boldsymbol{X}%
_{i}^{\prime }\boldsymbol{M}_{T}\boldsymbol{X}_{i}$,
$\bar{\boldsymbol{q}}%
_{NT,\beta }=N^{-1}\sum_{i=1}^{N}  ( T^{-1}\boldsymbol{X}_{i}^{
\prime }%
\boldsymbol{M}_{T}\boldsymbol{X}_{i}  ) \boldsymbol{\eta }_{i,
\beta }$ and
%
\begin{equation}
c_{NT}^{\text{FE}}=-N^{-1}\sum
_{i=1}^{N}\mathrm{E} \bigl( \boldsymbol{\eta
}%
_{i,\beta }^{\prime }\bar{\bar{\boldsymbol{x}}}_{i,T+1}
\bar{\varepsilon}%
_{iT} \bigr) +\bar{\boldsymbol{q}}_{N,\beta }^{\prime }
\bar{\boldsymbol{Q}}%
_{N,\beta }^{-1} \Biggl[
N^{-1}\sum_{i=1}^{N}
\mathrm{E} ( \bar{\boldsymbol{%
x}}_{iT}\bar{
\varepsilon}_{iT} ) \Biggr] . \label{cNT-FE}
\end{equation}%
$c_{NT}^{\text{FE}}$ tends to zero for $T$ sufficiently large or if
$%
\boldsymbol{x}_{it}$ is strictly exogenous.

Similar to the case of the individual and pooled forecasts, for $T$ finite
and $N$ large, the ranking of the individual and FE forecasts will depend
on the relative magnitudes of estimation error and parameter heterogeneity.
Precise expressions can be found in the Supplemental Appendix. For
$T\rightarrow \infty $ the individual forecasts will be more precise than
the FE forecasts.

\section{Forecast combinations}\label{sec:combinations}

We next consider approaches that combine the forecasts from Section~\ref{sec:theory_msfe}
to minimize the MSFE.

\subsection{Combinations of individual and pooled forecasts}\label%
{com_Ind_pool}

Given the MSFE trade-off associated with the forecasts in (\ref{eq:indFore})
and (\ref{eq:poolFore}), combining the forecasts based on the individual
and pooled estimates, $\hat{y}_{i,T+1}$ and $\tilde{y}_{i,T+1}$, may be
desirable. As noted in the literature (e.g., \citet{Tim2006}), forecast
combinations tend to perform particularly well, relative to the underlying
forecasts, if the forecast errors are weakly correlated and have MSFE values
of a similar magnitude. Correlations between forecast errors based on the
individual and pooled estimation schemes tend to be lower for (i) greater
differences in the estimates of $\boldsymbol{\theta }_{i}$ resulting from
larger estimation errors (small $T$); (ii) greater heterogeneity (large
$%
 \llVert  \boldsymbol{\Omega }_{\eta } \rrVert  $), and (iii) greater
bias of the pooled estimator due to correlated heterogeneity.

If the level of parameter heterogeneity is either very large or very small,
one of the individual or pooled estimation approaches will be dominant,
reducing potential gains from forecast combination. Similarly, if
$T$ is very small but $N$ is large and there is little parameter heterogeneity,
we would expect pooled estimation to dominate individual estimation by
a sufficiently large margin that forecast combination offers small, if
any, gains. Conversely, if $T$ is very large, forecasts using individual
estimates will dominate forecasts based on pooled estimates by a sufficient
margin that renders forecast combination less attractive. Building on these
observations, we combine the two forecasts $\hat{y}_{i,T+1}$ and
$\tilde{y}%
_{i,T+1}$ using common weights, $\omega $, to obtain\footnote{We focus here on a simple constant-coefficient linear combination scheme.
\cite{Lahetal2017} discuss a broader range of combination methods and
\citet{Ell2017} provides an analysis of the effect on the combination weights
and forecasting performance from having a large common component in the
forecast errors.}
%
\begin{equation}
y_{i,T+1}^{\ast }(\omega )=\omega \hat{y}_{i,T+1}+(1-
\omega ) \tilde{y}%
_{i,T+1}, \label{combination}
\end{equation}%
with associated forecast error
$e_{i,T+1}^{\ast }(\omega )=\omega \hat{e}%
_{i,T+1}+(1-\omega )\tilde{e}_{i,T+1}$. The average MSFE of the combined
forecast is given by
\begin{eqnarray*}
N^{-1}\sum_{i=1}^{N}e_{i,T+1}^{\ast 2}(
\omega ) &=&\omega ^{2} \Biggl( N^{-1}\sum
_{i=1}^{N}\hat{e}_{i,T+1}^{2}
\Biggr) +(1-\omega )^{2} \Biggl( N^{-1}\sum
_{i=1}^{N}\tilde{e}_{i,T+1}^{2}
\Biggr)
\\
&&{}+2\omega (1-\omega ) \Biggl( N^{-1}\sum
_{i=1}^{N}\hat{e}_{i,T+1}
\tilde{e}%
_{i,T+1} \Biggr) .
\end{eqnarray*}%
The value of $\omega $ that minimizes the average MSFE is therefore given
by%
%
\begin{equation}
\omega _{NT}^{\ast }= \frac{N^{-1}\sum
_{i=1}^{N}\tilde{e}_{i,T+1}^{2}-
\Biggl( N^{-1}\sum_{i=1}^{N}
\hat{e}_{i,T+1}\tilde{e}_{i,T+1} \Biggr) }{ \Biggl(
N^{-1}\sum_{i=1}^{N}
\hat{e}_{i,T+1}^{2} \Biggr) + \Biggl( N^{-1}\sum
_{i=1}^{N}%
\tilde{e}_{i,T+1}^{2} \Biggr) -2 \Biggl( N^{-1}
\sum_{i=1}^{N}\hat{e}_{i,T+1}%
\tilde{e}_{i,T+1} \Biggr) }. \label{wstar}
\end{equation}%
Expressions for $N^{-1}\sum_{i=1}^{N}\hat{e}_{i,T+1}^{2}$ and
$%
N^{-1}\sum_{i=1}^{N}\tilde{e}_{i,T+1}^{2}$ are given by (\ref{MSFEI})
and (%
\ref{MSFEP1}), respectively. We obtain a similar expression for
$%
N^{-1}\sum_{i=1}^{N}\hat{e}_{i,T+1}\tilde{e}_{i,T+1}$, with
$%
N^{-1}\sum_{i=1}^{N}\varepsilon _{i,T+1}^{2}$ canceling out from
$\omega _{NT}^{\ast }$. The result is summarized in the following proposition
(proven in Appendix Section~\ref{ProofCombined}).

\begin{proposition}
\label{prop:combined}
\begin{enumerate}
\item[(a)] Under Assumptions \ref{ass:1}--\ref{ass:6}, and for a given value of
$\boldsymbol{w}_{i,T+1}$, the optimal combination weight that minimizes
the MSFE of the forecast combination in (\ref{combination}) is given by
%
\begin{equation}
\omega _{NT}^{\ast }= \frac{\Delta _{NT}-T^{-1}{
\psi }_{NT}}{\Delta _{NT}+T^{-1}h_{NT}-2T^{-1}{
\psi }_{NT}}+O_{p}\bigl(N^{-1/2}\bigr),
\label{w*NT}
\end{equation}%
where%
%
\begin{align}
h_{NT}&=N^{-1}\sum_{i=1}^{N}
\mathrm{E} \biggl[ \boldsymbol{w}_{i,T+1}^{
\prime }%
\boldsymbol{Q}_{iT}^{-1} \biggl( \frac{
\boldsymbol{W}_{i}^{\prime }\boldsymbol{%
\varepsilon
}_{i}\boldsymbol{\varepsilon }_{i}^{\prime }
\boldsymbol{W}_{i}}{T%
} \biggr) \boldsymbol{Q}_{iT}^{-1}
\boldsymbol{w}_{i,T+1} \biggr] >0, \label{hNT}%
\\
\Delta _{NT}&=N^{-1}\sum_{i=1}^{N}
\mathrm{E} \bigl( \boldsymbol{w}%
_{i,T+1}^{\prime }
\boldsymbol{\eta }_{i}\boldsymbol{\eta }_{i}^{
\prime }%
\boldsymbol{w}_{i,T+1} \bigr) -\boldsymbol{\bar{q}}_{N}^{\prime }
\boldsymbol{%
\bar{Q}}_{N}^{-1}\boldsymbol{
\bar{q}}_{N}>0, \label{DeltaNT}%
\end{align}
and
%
\begin{align}
\psi _{NT} ={}&TN^{-1}\sum_{i=1}^{N}
\mathrm{E} \bigl[ \boldsymbol{\varepsilon }%
_{i}^{\prime }
\boldsymbol{W}_{i}\bigl(\boldsymbol{W}_{i}^{\prime }
\boldsymbol{W}%
_{i}\bigr)^{-1}
\boldsymbol{w}_{i,T+1}\boldsymbol{w}_{i,T+1}^{\prime }
\bigr] \boldsymbol{\bar{Q}}_{N}^{-1}\boldsymbol{
\bar{q}}_{N} \notag
\\
& {}-TN^{-1}\sum_{i=1}^{N}
\mathrm{E} \bigl[ \boldsymbol{\varepsilon }%
_{i}^{\prime }
\boldsymbol{W}_{i}\bigl(\boldsymbol{W}_{i}^{\prime }
\boldsymbol{W}%
_{i}\bigr)^{-1}
\boldsymbol{w}_{i,T+1}\boldsymbol{w}_{i,T+1}^{\prime }
\boldsymbol{%
\eta }_{i} \bigr] .\label{epsi_NT}
\end{align}%

\item[(b)] Under strict exogeneity, irrespective of whether heterogeneity is correlated,
we have $\psi _{NT}=0$,
$h_{NT}=N^{-1}%
\sum_{i=1}^{N}\sigma ^{2}_{i}\mathrm{E}  ( \boldsymbol{w}_{i,T+1}^{
\prime }\boldsymbol{Q}%
_{iT}^{-1}\boldsymbol{w}_{i,T+1}  ) $, and
\[
\Delta _{NT}=N^{-1}\sum_{i=1}^{N}\mathrm{E}  \bigl( \boldsymbol{w}_{i,T+1}^{
\prime }
\boldsymbol{\Omega }_{\eta }\boldsymbol{w}_{i,T+1}  \bigr) .
\]
\end{enumerate}\end{proposition}

For small to moderate values of $T$ and large $N $, we expect
$\omega _{NT}^{\ast }<1$, with a nonzero weight placed on forecasts
based on the pooled estimate.

\subsubsection*{Forecast combinations with individual weights}

\cite{Pesetal2022} show that, under strict exogeneity of the regressors
and uncorrelated heterogeneity, optimal weights can be obtained that are
specific to the individual unit. The combination of individual and pooled
forecast is then
\begin{equation*}
y_{i,T+1}^{\ast }=\omega _{i}
\hat{y}_{i,T+1}+(1-\omega _{i}) \tilde{y}%
_{i,T+1},
\end{equation*}%
where the optimal value of $\omega _{i}$ is given by
%
\begin{equation}
\omega _{i}^{\ast }= \frac{\boldsymbol{w}_{i,T+1}^{\prime }
\boldsymbol{\Omega }_{\eta }\boldsymbol{w}_{i,T+1}}{
\boldsymbol{w}_{i,T+1}^{\prime }\bigl(T^{-1}\sigma
_{i}^{2}\boldsymbol{Q}_{iT}^{-1}+
\boldsymbol{\Omega }_{\eta }\bigr)%
\boldsymbol{w}_{i,T+1}}.
\label{eq:weights_per_i}
\end{equation}%
The weights again depend on the variances and covariances of the underlying
forecast errors. Related to this, \cite{Giaetal2023} develop a random
effects approach for linear panels that similarly combines univariate and
pooled forecasts in a way that minimizes minimax-regret and MSFE.

\subsection{Combining individual and fixed effect forecasts}
\label{com_Ind_FE}

Combination weights can also be determined for the case where the pooled
forecast is replaced with the FE forecasts. In this case, the combined
forecast is given by%
%
\begin{equation}
y_{i,T+1}^{\ast }(\omega _{\text{FE}})=\omega
_{\text{FE}}\hat{y}%
_{i,T+1}+(1-\omega _{\text{FE}})
\hat{y}_{i,T+1,\text{FE}}, \label{eq:combinationFE}
\end{equation}%
yielding the optimal pooled weight
%
\begin{equation}
\omega _{\text{FE},NT}^{\ast }= \frac{N^{-1}\sum
_{i=1}^{N} \bigl( \hat{e}%
_{i,T+1}^{\text{FE}}
\bigr) ^{2}- \Biggl( N^{-1}\sum
_{i=1}^{N}\hat{e}_{i,T+1}^{%
\text{FE}}
\hat{e}_{i,T+1} \Biggr) }{ \Biggl( N^{-1}\sum
_{i=1}^{N}\hat{e}%
_{i,T+1}^{2}
\Biggr) +N^{-1}\sum_{i=1}^{N}
\bigl( \hat{e}_{i,T+1}^{\text{FE}%
} \bigr) ^{2}-2
\Biggl( N^{-1}\sum_{i=1}^{N}
\hat{e}_{i,T+1}^{\text{FE}}\hat{e}%
_{i,T+1}
\Biggr) }. \label{wNT-FE}
\end{equation}%
The expressions for
$N^{-1}\sum_{i=1}^{N}  ( \hat{e}_{i,T+1}^{\text{FE}%
}  ) ^{2}$ and $N^{-1}\sum_{i=1}^{N}\hat{e}_{i,T+1}^{2}$ are given
by (%
\ref{FEmsfe}) and (\ref{MSFEI}), respectively, and the expression for
$%
N^{-1}\sum_{i=1}^{N}\hat{e}_{i,T+1}^{\text{FE}}\hat{e}_{i,T+1}$ can be
similarly obtained. In this case, the shared term
$\sum_{i=1}^{N}(%
\varepsilon _{i,T+1}-\bar{\varepsilon}_{iT})^{2}/N$ cancels out and we
have the result summarized in the following proposition with proofs provided
in Section~\ref{ProofCombinedFE} of the Appendix.

\begin{proposition}
\label{prop:combinedFE}
\begin{enumerate}
\item[(a)] Under Assumptions \ref{ass:1}--\ref{ass:6}, the optimal combination
weight that minimizes the MSFE of the forecast combination in (\ref{eq:combinationFE})
is given by%
%
\begin{equation}
\omega _{\text{FE},NT}^{\ast }= \frac{\Delta _{NT}^{\text{FE}}-T^{-1}
\psi _{NT}^{\text{FE}}- \bigl( c_{NT}^{\text{FE}}-c_{NT,\beta }
\bigr) }{\Delta _{NT}^{\text{FE}}+T^{-1}h_{NT,\beta }-2T^{-1}
\psi _{NT}^{\text{FE}}}%
+O_{p}
\bigl(N^{-1/2}\bigr), \label{wFE}
\end{equation}%
where $\Delta _{NT}^{\text{FE}}$ and $c_{NT}^{\text{FE}}$ are defined in (%
\ref{eq:Delta_FE}) and (\ref{cNT-FE}), respectively,
\begin{eqnarray*}
h_{NT,\beta }&=&N^{-1}\sum_{i=1}^{N}
\mathrm{E} \biggl[ \bar{\bar{\boldsymbol{x}}}%
_{i,T+1}^{\prime }
\boldsymbol{Q}_{iT,\beta }^{-1} \biggl( \frac{
\boldsymbol{X}%
_{i}^{\prime }\boldsymbol{M}_{T}
\boldsymbol{\varepsilon }_{i}\boldsymbol{%
\varepsilon
}_{i}^{\prime }\boldsymbol{M}_{T}
\boldsymbol{X}_{i}}{T} \biggr) \boldsymbol{Q}_{iT,\beta }^{-1}
\bar{\bar{\boldsymbol{x}}}_{i,T+1} \biggr] , \label{eq:hNT_FEapp2}
\\
\psi _{NT}^{\text{FE}} &=&TN^{-1}\sum
_{i=1}^{N}\mathrm{E} \bigl[ (\hat{%
\boldsymbol{\beta}}_{\text{FE}}-\boldsymbol{\beta}_{i}
)^{
\prime }\bar{\bar{%
\boldsymbol{x}}}_{i,T+1}\bar{\bar{
\boldsymbol{x}}}_{i,T+1}^{\prime } \bigr] \bar{
\boldsymbol{Q}}_{N,\beta }^{-1}\bar{\boldsymbol{q}}_{N,
\beta }
\notag
\\
&&{}-TN^{-1}\sum_{i=1}^{N}
\mathrm{E} \bigl[ ( \hat{\boldsymbol{\beta}}_{%
\text{FE}}-\boldsymbol{
\beta}_{i} )^{\prime } \bar{\bar{\boldsymbol{x}}}%
_{i,T+1}
\bar{\bar{\boldsymbol{x}}}_{i,T+1}^{\prime } \boldsymbol{\eta
}%
_{i,\beta } \bigr], \label{FEepsi}
\end{eqnarray*}
and
\begin{equation*}
c_{NT,\beta }=N^{-1}\sum_{i=1}^{N}
\mathrm{E} \bigl[ \bar{\bar{\boldsymbol{x}}}%
_{i,T+1}^{\prime }
\bigl(\boldsymbol{X}_{i}^{\prime }\boldsymbol{M}_{T}
\boldsymbol{%
X}_{i}\bigr)^{-1}
\boldsymbol{X}_{i}^{\prime }\boldsymbol{M}_{T}
\boldsymbol{%
\varepsilon }_{i}\bar{\varepsilon}_{iT}
\bigr] . \label{cNTbeta}
\end{equation*}

\item[(b)] Under uncorrelated heterogeneity, $\psi _{NT}^{\text{FE}}=0$, and
$%
\Delta _{NT}^{\text{FE}}$ and $h_{NT,\beta}$ will be affected accordingly.
\end{enumerate}
\end{proposition}

\subsection{Estimation of combination weights}
\label{comW}

Estimates of the weights for the forecast combination in Proposition~\ref{prop:combined}
require estimates of $\Delta _{NT}$, $h_{NT}$, and $\psi _{NT}$. Under
Assumption~\ref{ass:6}, these terms can be estimated by their sample means
with unknown parameters replaced by their estimates. We summarize the estimators
here with details in Appendix~\ref{app:estimationOfWeights}%
. Using (\ref{hNT}) and (\ref{DeltaNT}), the estimators of
$\Delta _{NT}$ and $h_{NT}$ are given by
%
\begin{equation}
\hat{\Delta}_{NT}=N^{-1}\sum
_{i=1}^{N}\boldsymbol{w}_{i,T+1}^{
\prime }
\tilde{%
\boldsymbol{\eta }}_{i}\tilde{\boldsymbol{\eta
}}_{i}^{\prime } \boldsymbol{w}%
_{i,T+1},
\label{eq:Deltahat}
\end{equation}%
where
$\tilde{\boldsymbol{\eta }}_{i}=\tilde{\boldsymbol{\theta }}-%
\boldsymbol{\hat{\theta}}_{i}$, and
%
\begin{equation}
\hat{h}_{NT}=N^{-1}\sum_{i=1}^{N}
\boldsymbol{w}_{i,T+1}^{\prime } \boldsymbol{%
Q}_{iT}^{-1}
\hat{\boldsymbol{H}}_{iT}\boldsymbol{Q}_{iT}^{-1}
\boldsymbol{w}%
_{i,T+1}, \label{eq:hhat}
\end{equation}%
where
$\hat{\boldsymbol{H}}_{iT}=\hat{\sigma}_{i}^{2}T^{-1}\sum_{t=1}^{T}%
\boldsymbol{w}_{it}\boldsymbol{w}_{it}^{\prime }$,
$\hat{\sigma}_{i}^{2}=$
$%
\sum_{t=1}^{T}\hat{\varepsilon}_{it}^{2}/(T-K)$, and
$\hat{\varepsilon}%
_{it}=y_{it}-\hat{\boldsymbol{\theta }}_{i}^{\prime }\boldsymbol{w}_{it}$.
We show in Appendix~\ref{app:estimationOfWeights} that
\begin{eqnarray*}
\hat{\Delta}_{NT}-\Delta _{NT} &=&O_{p}
\bigl( N^{-1/2} \bigr) +O_{p}\bigl(T^{-1}\bigr),
\\
\hat{h}_{NT}-h_{NT} &=&O_{p}
\bigl(N^{-1/2}\bigr)+O_{p} \biggl( \frac{\ln (N)}{
\sqrt{T}}%
 \biggr) .
\end{eqnarray*}%
In the case of strictly exogenous regressors, $\hat{\Delta}_{NT}$ is a
consistent estimator of $\Delta _{NT}$ for fixed $T$ as
$N\rightarrow \infty $.

Consider now $\psi _{NT}$, given by (\ref{epsi_NT}), and recall that
$\psi _{NT}=0$ under uncorrelated heterogeneity. To estimate
$\psi _{NT}$ under correlated heterogeneity, we first note that the approach
of replacing expectations by sample moments and then estimating
$\boldsymbol{\varepsilon}_{i}$ from the OLS residuals,
$\boldsymbol{\hat{\varepsilon}}_{i}=  (\mathbf{y}_{i}-\mathbf{W}_{i}
\hat{\boldsymbol{\theta }}_{i}  ) $ will not work in the case of
$\psi _{NT}$, since
$\mathbf{W}_{i}^{\prime} \hat{\boldsymbol{\varepsilon}}_{i} = \mathbf{W}_{i}^{
\prime }  (\mathbf{y}_{i} -\mathbf{W}_{i}
\hat{\boldsymbol{\theta }}_{i}  ) =\mathbf{0}$ for all $i$. If used
in (\ref{epsi_NT}), this results in $\hat{\psi}_{NT}=0$, which is not a
consistent estimator of $\psi _{NT}$ under correlated heterogeneity. To
overcome this problem, we replace
$\boldsymbol{\varepsilon }_{i}^{\prime} \boldsymbol{W}_{i}(
\boldsymbol{W}_{i}^{\prime }\boldsymbol{W}_{i})^{-1}$ by
$  ( \hat{\boldsymbol{\theta }}_{i}-\boldsymbol{\theta }_{i}
  )^{\prime}$ and note that $\psi _{NT}$ can be written equivalently
as (noting that $\boldsymbol{w}_{i,T+1} $ for $i=1,2,\ldots ,N$ are given)
%
\begin{align}
\psi _{NT} ={}&N^{-1}\sum_{i=1}^{N}
\mathrm{E} \bigl[ T ( \hat{\boldsymbol{%
\theta }}_{i}-
\boldsymbol{\theta }_{i} ) ^{\prime } \bigr]
\mathrm{E}%
 \bigl( \boldsymbol{w}_{i,T+1}\boldsymbol{w}_{i,T+1}^{\prime }
\bigr) \boldsymbol{\bar{Q}}_{N}^{-1}\boldsymbol{
\bar{q}}_{N} \notag
\\
&{} -N^{-1}\sum_{i=1}^{N}
\mathrm{E} \bigl[ T ( \hat{\boldsymbol{\theta }}%
_{i}-
\boldsymbol{\theta }_{i} ) ^{\prime } \bigr] \mathrm{E} \bigl(
\boldsymbol{w}_{i,T+1}\boldsymbol{w}_{i,T+1}^{\prime }
\boldsymbol{\eta }%
_{i} \bigr) .\label{epsi_NT2}
\end{align}
We now employ a half-jackknife estimator of
$\boldsymbol{\theta }_{i}$ (\citet{DhaJoc2015}, \cite{Chuetal2018})
to estimate
$T\mathrm{E}%
  ( \hat{\boldsymbol{\theta }}_{i}-\boldsymbol{\theta }_{i}
  ) $, which is the small sample bias of
$\hat{\boldsymbol{\theta }}_{i}$ in the case of weakly exogenous regressors.
The half-jackknife estimator of $%
\boldsymbol{\theta }_{i}$ is defined by
$\hat{\boldsymbol{\theta }}_{i,\mathit{JK}}=2%
\hat{\boldsymbol{\theta }}_{i}-\frac{1}{2}  (
\hat{\boldsymbol{\theta }}%
_{ia}+\hat{\boldsymbol{\theta }}_{ib}  ) $, where
$\hat{\boldsymbol{%
\theta }}_{ia}$ and $\hat{\boldsymbol{\theta }}_{ib}$ are least squares
estimators of $\boldsymbol{\theta }_{i}$ based on two equal halves of the
sample of size $T_{h}=T/2$ (omitting an observation in the case of uneven
$T$%
), namely
$\hat{\boldsymbol{\theta }}_{ia}=  ( \sum_{t=1}^{T_{h}}%
\boldsymbol{w}_{it}\boldsymbol{w}_{it}^{\prime }  ) ^{-1}\sum_{t=1}^{T_{h}}
\boldsymbol{w}_{it}y_{it}$ and
$\hat{\boldsymbol{%
\theta }}_{ib}=  ( \sum_{t=T_{h}+1}^{T}\boldsymbol{w}_{it}
\boldsymbol{w}%
_{it}^{\prime }  ) ^{-1}\sum_{t=T_{h}+1}^{T}\boldsymbol{w}_{it}y_{it}$.
Then
$\mathrm{E}  [ T  ( \hat{\boldsymbol{\theta }}_{i}-
\boldsymbol{%
\theta }_{i}  )   ] $ can be estimated by
$T  [ \frac{1}{2}  ( \hat{\boldsymbol{\theta }}_{ia}+
\hat{\boldsymbol{\theta }}_{ib}  ) -\hat{%
\boldsymbol{\theta }}_{i}  ] $ and $\psi _{NT}$ by%
%
\begin{eqnarray}
{\hat{\psi}_{NT}} &=& \Biggl[ TN^{-1}\sum
_{i=1}^{N} \biggl[ \frac{1}{2} ( \hat{
\boldsymbol{\theta }}_{ia}+ \hat{\boldsymbol{\theta
}}_{ib} ) -\hat{%
\boldsymbol{\theta }}_{i}
\biggr] ^{\prime }\boldsymbol{w}_{i,T+1} \boldsymbol{%
w}_{i,T+1}^{\prime }
\Biggr] \boldsymbol{\bar{Q}}_{NT}^{-1} \boldsymbol{
\bar{q}%
}_{NT} ( \hat{\boldsymbol{\eta }} )\notag
\\
&&{}-TN^{-1}\sum_{i=1}^{N}
\biggl[ \frac{1}{2} ( \hat{\boldsymbol{\theta }}%
_{ia}+
\hat{\boldsymbol{\theta }}_{ib} ) - \hat{\boldsymbol{\theta
}}_{i}%
 \biggr] ^{\prime }\boldsymbol{w}_{i,T+1}
\boldsymbol{w}_{i,T+1}^{
\prime }\hat{%
\boldsymbol{\eta
}}_{i},\label{eq:psihat}
\end{eqnarray}%
where
$\hat{\boldsymbol{\eta }}_{i}=\hat{\boldsymbol{\theta }}%
_{i}-N^{-1}\sum_{i=1}^{N}\hat{\boldsymbol{\theta }}_{i}$, and
$\boldsymbol{%
\bar{q}}_{NT}  ( \hat{\boldsymbol{\eta }}  ) =(NT)^{-1}
\sum_{i=1}^{N}\boldsymbol{W}_{i}^{\prime }\boldsymbol{W}_{i}
\hat{%
\boldsymbol{\eta }}_{i}$. Consistency of ${\hat{\psi}_{NT}}$ as an estimator
of $\psi _{NT}$ is established as $N,T\rightarrow \infty $, since
$T\mathrm{E%
}  ( \hat{\boldsymbol{\theta }}_{i,\mathit{JK}}-
\boldsymbol{\theta }_{i}  ) =O  ( T^{-1}  ) $.\footnote{Note that ${\hat{\psi}_{NT}-\psi _{NT}}$ depends on
$\mathrm{E}  [   ( \hat{\boldsymbol{\theta }}_{i}-
\boldsymbol{\theta }%
_{i}  ) -  ( \frac{1}{2}  ( \hat{\boldsymbol{\theta }}_{ia}+
\hat{%
\boldsymbol{\theta }}_{ib}  ) -\hat{\boldsymbol{\theta }}_{i}
  )   ] =\mathrm{E}  ( \hat{\boldsymbol{\theta }}_{i,
\mathit{JK}}-\boldsymbol{%
\theta }_{i}  )$.} Thus, to use the half-jackknife method for models
with weakly exogenous regressors we need $T$ large, although it is not
required that $\sqrt{T}/N$ tends to zero, as it is in the case of large
$N$ and $T$ asymptotics.

The components of the weights in Proposition~\ref{prop:combinedFE} that
combine individual and fixed effects forecasts can be estimated in a similar
fashion, with details provided in Appendix~\ref{app:estimationOfWeights}.

\subsection{Empirical Bayes forecasts}

Bayesian panel forecasts are becoming increasingly common in empirical
applications and constitute an alternative approach to the frequentist
forecasts discussed so far. Due to their resemblance to our forecast combination
schemes and their recent popularity (e.g., \cite{Armetal2022} and
\citet{Efr2016}), we focus on empirical Bayes (EB) methods. The EB forecast
uses the estimator of \cite{Hsietal1999} and takes the form
$ \hat{y}_{i,T+1}^{\mathit{EB}}=\hat{\boldsymbol{\theta }}_{i,
\mathit{EB}%
}^{\prime }\boldsymbol{w}_{i,T+1}$, where
%
\begin{equation}
\hat{\boldsymbol{\theta }}_{i,\mathit{EB}}=\bigl(\hat{\sigma}_{i}^{-2}
\boldsymbol{%
W}_{i}^{\prime }\boldsymbol{W}_{i}+
\boldsymbol{\hat{\Omega}}_{\eta }^{-1}\bigr)^{-1}
\bigl( \hat{\sigma}_{i}^{-2}\boldsymbol{W}_{i}^{\prime }
\boldsymbol{y}%
_{i}+\boldsymbol{\hat{\Omega}}_{\eta }^{-1}
\bar{\hat{\boldsymbol{%
\theta }}}\bigr), \label{EB}
\end{equation}%
$%
\bar{\hat{\boldsymbol{\theta }}}=N^{-1}\sum_{i=1}^{N}
\hat{\boldsymbol{\theta }}_{i}$, $  \hat{\sigma}_{i}^{2}=(T-K)^{-1}
\hat{\boldsymbol{\varepsilon }}%
_{i}^{\prime }\hat{\boldsymbol{\varepsilon }}_{i}$,  and
$\boldsymbol{\hat{\Omega}}_{\eta }=\frac{1}{N}\sum_{i=1}^{N}(
\hat{%
\boldsymbol{\theta }}_{i}-\bar{\hat{\boldsymbol{\theta }}})(
\hat{\boldsymbol{%
\theta }}_{i}-\bar{\hat{\boldsymbol{\theta }}})^{\prime }$, where
$\hat{%
\boldsymbol{\varepsilon }}_{i}=\boldsymbol{y}_{i}-\boldsymbol{W}_{i}
\hat{%
\boldsymbol{\theta }}_{i}$, and
$\hat{\boldsymbol{\theta }}_{i}=  ( \boldsymbol{W}_{i}^{\prime }
\boldsymbol{W}_{i}  ) ^{-1}\boldsymbol{W}%
_{i}^{\prime }\boldsymbol{y}_{i}$.\footnote{It is necessary that $N>T$ for $\boldsymbol{\hat{\Omega}}_{\eta }$ to be
positive definite.} $\hat{\boldsymbol{\theta }}_{i,\mathit{EB}}$ can also
be written as a weighted average of
$\hat{\boldsymbol{\theta }}_{i}$, which allows for full heterogeneity,
and the mean group estimator, $\bar{\hat{%
\boldsymbol{\theta }}}$, namely
$\hat{\boldsymbol{\theta }}_{i,\mathit{EB}}=%
\boldsymbol{\mathcal{W}}_{iT}\hat{\boldsymbol{\theta }}_{i}+  (
\boldsymbol{I}_{k}-\boldsymbol{\mathcal{W}}_{iT}  )
\bar{\hat{%
\boldsymbol{\theta }}}$, with the weight matrix
$\boldsymbol{\mathcal{W}}%
_{iT}$ given by%
%
\begin{equation}
\boldsymbol{\mathcal{W}}_{iT}= \bigl( \boldsymbol{I}_{k}+T^{-1}
\hat{\sigma}%
_{i}^{2}\boldsymbol{Q}_{iT }^{-1}
\boldsymbol{\hat{\Omega}}_{\eta }^{-1} \bigr)
^{-1}, \label{eq:empBayweights}
\end{equation}
recalling that
$\boldsymbol{Q}_{iT }=T^{-1}\boldsymbol{W}_{i}^{\prime }%
\boldsymbol{W}_{i}$ is invertible under Assumption~\ref{ass:3}. The weights
on the heterogeneous estimates are larger, the greater the degree of heterogeneity,
as measured by the norm of $\boldsymbol{\hat{\Omega}}_{\eta }$%
, with
$\hat{\boldsymbol{\theta }}_{i,\mathit{EB}}\rightarrow
\hat{%
\boldsymbol{\theta }}_{i}$ as
$ \llVert  \boldsymbol{\hat{\Omega}}_{\eta } \rrVert
\rightarrow \infty $. Also, since
$\hat{\sigma}_{i}^{2}%
\boldsymbol{Q}_{iT }^{-1}\boldsymbol{\hat{\Omega}}_{\eta }^{-1}$ is bounded
in $T$, $\hat{\boldsymbol{\theta }}_{i,\mathit{EB}}$ converges
\textit{%
numerically} to $\hat{\boldsymbol{\theta }}_{i}$, as
$T\rightarrow \infty $. Hence, one would expect the EB estimator to perform
well even when $T$ is relatively small and the degree of heterogeneity
is not too large. For large $T$, EB and individual forecasts coincide and
both methods will work well.

The EB weights do \textit{not} depend on $\boldsymbol{w}_{i,T+1}$ and are
derived assuming uncorrelated heterogeneity and strictly exogenous regressors.
They have the desirable feature of placing more weights on individual estimates
if they are precisely estimated relative to the degree of parameter heterogeneity
measured by $\boldsymbol{\hat{\Omega}}_{\eta }$. The individual optimum
weights in (\ref{eq:weights_per_i}) fall somewhere between the common optimal
weights and the EB weights.\footnote{While the EB estimator in (\ref{EB}) is fully parametric, other studies
pursue a nonparametric approach to the distribution of
$\hat{\boldsymbol{%
\theta }}_{i}$; see, for example, \citet{BroGre2009} and \citet{GuKoe2017}, and more recently, \citet{Liu2023}
 and \cite{Liuetal2023}.} Like
the EB weights, consistent estimation of individual weights require strict
exogeneity and uncorrelated heterogeneity.

The EB weights are comparable to the unit-specific weights given by (\ref{eq:weights_per_i})
and the two sets of weights coincide only when
$\mathbf{w%
}_{i,T+1}$ is an scalar. To see this, note that the estimates of the unit
specific weights can be written as
%
\begin{equation}
\hat{\omega}_{iT}^{\ast }= \biggl[ 1+\hat{
\sigma}_{i}^{2}T^{-1} \frac{
\mathbf{w}%
_{i,T+1}^{\prime }\mathbf{Q}_{iT}^{-1}
\mathbf{w}_{i,T+1}}{\mathbf{w}%
_{i,T+1}^{\prime }
\hat{\boldsymbol{\Omega}}_{\eta }\mathbf{w}_{i,T+1}}
\biggr] ^{-1}, \label{eq:indOptimalWeights}
\end{equation}%
and reduces to the EB weights only when $K=1$, and
$\hat{\omega}_{iT}^{\ast }$ no longer depend on $\mathbf{w}_{i,T+1}$. But
in general the estimates of the unit-specific weights differ from the EB
weights.

An alternative to the EB forecast is a hierarchical Bayesian approach as
proposed by \citet{LinSmi1972} and further explored by \citet{Geletal1996}. 
The full Bayesian treatment would
require choices of the priors of each component, including the parameter
covariance matrix. In the Supplemental Appendix, we provide Monte Carlo
results that shows that the resulting forecast performance is highly sensitive
to the choice of priors.

\section{Monte Carlo experiments}

\label{sec:MC}

We examine the finite-sample performance of the panel forecasting schemes in
the context of a dynamic heterogeneous panel data model using Monte Carlo
experiments.\footnote{%
Further analytical results for a simple panel AR(1) model are provided in
Section~\ref{PanelAR} of the Appendix.} We allow for dynamics, parameter
heterogeneity, and correlations between the regressors and coefficients. The
forecasting methods are: (1) individual estimation which serves as the
benchmark against which other methods are compared, (2) pooled estimation,
(3) random effects, (4) fixed effects, (5) combination of individual and
pooled forecasts using the weights in~(\ref{w*NT}), (6) combination of
individual and FE forecasts using the weights in~(\ref{wFE}), (7) individual
forecast combination weights, and (8) EB forecasts.\footnote{%
Additional results for equal weighted combinations and oracle weights are in
Section~S.4 of the Supplemental Appendix.}

Results do not vary greatly along the $N$ dimension, so we focus on the case
with $N=100$ and provide results for $N=1000$ in the Supplemental Appendix. The $%
T$ dimension of the panel is more important, so we consider three different
values, $T=\{20,50,100\}$. The values of the parameters used in the
simulations are reported in Table~S.1 in Appendix~S.2.

\subsection{Data generating process}
\label{sec:MCDesignARX}

Our DGP augments a panel AR(1) model with an additional regressor,
%
\begin{equation}
y_{it}=\alpha _{i}+\beta _{i}y_{i,t-1}+
\gamma _{i}x_{it}+\varepsilon _{it},
\label{eq:MC1}
\end{equation}%
where $\varepsilon _{it}=\sigma _{i}(z_{it}^{2}-1)/\sqrt{2}$  with
$%
z_{it}\sim \mathit{iid}\mathrm{N}(0,1)$,
$\sigma _{i}^{2}\sim \mathrm{iid}%
   ( 1+\chi _{1}^{2}  ) /2$, and $x_{it}$ is generated as
%
\begin{equation}
x_{it}=\mu _{xi}+\xi _{it},
\label{xit}
\end{equation}%
where
$%
\xi _{it}=\rho _{xi}\xi _{i,t-1}+\sigma _{xi}  ( 1-\rho _{xi}^{2}
  ) ^{1/2}\nu _{it}$, $  \nu _{it}\sim \mathrm{iidN}  ( 0,1
  ) $,  $\mu _{xi}=(z_{i}^{2}-1)/\sqrt{2}$,
$z_{i}\sim \mathrm{iidN}  ( 0,1  ) $, and
$\sigma _{xi}^{2}\sim \mathrm{iid}   ( 1+\chi _{1}^{2}  ) /2$,
for individual units $i=1,2,\ldots ,N$, and observation periods
$t=1,2,\ldots ,T$. The autocorrelation coefficient of $x_{it}$ is
$%
\rho _{xi}\sim \mathrm{iid} \text{Uniform}(0,0.95)$, allowing for a high
degree of dynamic heterogeneity in the regressors.

The coefficients of the lagged dependent variables, $y_{i,t-1}$, are generated
as $ \beta _{i}=\beta _{0}+\eta _{i\beta }$, with
$ \eta _{i\beta }\sim \mathrm{iid}\operatorname{Uniform} (-a_{\beta }/2,a_{\beta }/2)$
and $0\leq a_{\beta }<2(1- \llvert  \beta _{0} \rrvert  )$.

To allow for correlated heterogeneity, we set
%
\begin{equation}
\alpha _{i}=\alpha _{0i}+\phi \mu _{xi}+
\sigma _{\eta }\eta _{i},\quad  \text{and}\quad %
\gamma
_{i}=\gamma _{0i}+\pi \mu _{xi}+\sigma
_{\zeta }\zeta _{i}, \label{aci}
\end{equation}%
where $\eta _{i},\zeta _{i}\sim \mathrm{iidN}(0,1)$ and
$\alpha _{0}=\mathrm{%
E}  ( \alpha _{i}  ) =\alpha _{0i}+\phi \mathrm{E}  (
\mu _{xi}  ) =\alpha _{0i}$. We examine three settings:
\begin{itemize}
\item $\alpha _{0i}=2/3$ if $i\leq N/2$, $\alpha _{0i}=4/3$ if
$i>N/2$, $%
\sigma _{\alpha }^{2}=0.5$, $\gamma _{0i}=0.1$, and
$\sigma _{\gamma }^{2}=a_{\beta }=0$
\item $\alpha _{0i}=2/3$ if $i\leq N/2$, $\alpha _{0i}=4/3$ if
$i>N/2$, $%
\sigma _{\alpha }^{2}=0.5$, $\gamma _{0i}=0.2/3$ if $i\leq N/2$,
$\gamma _{0i}=0.4/3$ if $i>N/2$, $\sigma _{\gamma }^{2}=0.1$, and
$a_{\beta }=0.5$
\item $\alpha _{0i}=2/3$ if $i\leq N/2$, $\alpha _{0i}=4/3$ if
$i>N/2$, $%
\sigma _{\alpha }^{2}=1$, $\gamma _{0i}=0.2/3$ if $i\leq N/2$,
$\gamma _{0i}=0.4/3$ if $i>N/2$, $\sigma _{\gamma }^{2}=0.2$, and
$a_{\beta }=1$
\end{itemize}

\noindent Note that nonzero correlations need not bias the pooled estimates. What
matters for pooled estimates is the correlation between
$y_{i,t-1}^{2}, x_{it}^{2}$ and the individual coefficients.

Using (\ref{xit}) and (\ref{aci}), we have
\begin{align*}
\mathrm{E} \bigl[ x_{it} ( \gamma _{i}-\gamma
_{0} ) \bigr] &=%
\mathrm{E} \bigl[ (\mu _{xi}+
\xi _{it}) ( \pi \mu _{xi}+\sigma _{
\zeta }\zeta
_{i} ) \bigr] =\pi \mathrm{E} \bigl( \mu _{xi}^{2}
\bigr) \neq 0,%
\\
\mathrm{E} \bigl[ x_{it}^{2} ( \gamma _{i}-
\gamma _{0} ) \bigr] &=%
\mathrm{E} \bigl[ (\mu
_{xi}+\xi _{it})^{2} ( \pi \mu
_{xi}+ \sigma _{\zeta }\zeta _{i} ) \bigr] =
\pi \mathrm{E} \bigl( \mu _{xi}^{3} \bigr) .%
\end{align*}Therefore,
$\mathrm{E}  [ x_{i,t-1}^{2}  ( \gamma _{i}-\gamma _{0}
  )   ] =0$ if $\mu _{xi}$ are draws from a symmetric distribution
around $0$. To rule out this possibility, we draw $\mu _{xi}$ from a chi-square
distribution. To control the degree of correlated heterogeneity, we first
note that (taking expectations with respect to both $%
i$ and $t$)
\begin{align*}
\mathrm{E} ( \gamma _{i} ) & =\gamma _{0},\qquad
\operatorname{Var}(\gamma _{i})=\pi ^{2}+\sigma
_{\zeta }^{2},
\\
\mathrm{E} ( x_{it} ) & =\mathrm{E} ( \mu _{xi}+\xi
_{it} ) =0, \qquad  \operatorname{Var} ( x_{it} ) =\mathrm{E} (
x_{it}-\mu _{xi} ) ^{2}=\sigma
_{xi}^{2},
\end{align*}%
and
$\mathrm{E}  [ \operatorname{Var}  ( x_{it}  )   ] =
\mathrm{E}%
  ( 1+\chi _{1}^{2}  ) /2=1$. Also, since $\nu _{it}$ is distributed
independently of $\eta _{j}$ and $\zeta _{j}$ for all $t$, $i$,  and
$j$, $%
\operatorname{Cov}  ( \gamma _{i},x_{it}  ) =\pi $ and
$\operatorname{Corr}  ( \gamma _{i},x_{it}  ) =\pi   ( \sigma _{
\zeta }^{2}+\pi ^{2}  ) ^{-1/2}$. While heterogeneity is generally
correlated in AR panel models (\citet{PesSmi1995}), this setup allows
us to study further the role of correlated heterogeneity by varying the
correlation between the coefficient $%
\gamma _{i}$ and $x_{it}$ as measured by $\rho _{\gamma x}$. To achieve
a given level of $\operatorname{Corr}   ( \gamma _{i},x_{it}  ) =\rho _{
\gamma x}$, we set
%
\begin{equation}
\pi = \frac{\rho _{\gamma x}\sigma _{\zeta }}{ \bigl( 1-\rho
_{\gamma x}^{2} \bigr) ^{1/2}}. \label{eq:pi}
\end{equation}%
Similarly, to achieve $\operatorname{Corr}   ( \alpha _{i},x_{i,t-1}  ) =
\rho _{\alpha x}$, we set
%
\begin{equation}
\phi = \frac{\rho _{\alpha x}\sigma _{\eta }}{ \bigl( 1-\rho
_{\alpha x}^{2} \bigr) ^{1/2}}. \label{eq:phi}
\end{equation}%
Defining
$\sigma _{\gamma }^{2}=\operatorname{Var}(\gamma _{i})=\pi ^{2}+\sigma _{
\zeta }^{2}$, we can use (\ref{eq:pi}) to see that
$\pi =\rho _{\gamma x}\sigma _{\gamma }$. An equivalent result emerges
for $\phi $ where, for
$%
\sigma _{\alpha }^{2}=\operatorname{Var}(\alpha _{i})$, we have
$\phi =\rho _{\alpha x}\sigma _{\alpha }$. We thus use the parameters
$\sigma _{\alpha }^{2}$, $\sigma _{\gamma }^{2}$, and $a_{\beta }$ to vary
the degree of parameter heterogeneity in $\alpha _{i}$,
$\gamma _{i}$, and $\beta _{i}$, respectively.

We set $\xi _{i0}=0$ and initialize $y_{i0}$ as
$y_{i0}\sim \mathrm{iidN}%
  ( \mu _{iy0},\sigma _{iy0}^{2}  ) $ with
$%
\mu _{iy0}=\frac{\alpha _{i}+\gamma _{i}\mu _{xi}}{1-\beta _{i}^{2}}$, $
 \sigma _{iy0}^{2}=
\frac{\gamma _{i}^{2}\sigma _{xi}^{2}+\sigma _{i}^{2}}{%
1-\beta _{i}^{2}}$,  We also experimented with initialization schemes that
started the DGP on values away from the long run equilibrium, which did
not change the results qualitatively.

Since the forecast combinations use
$\boldsymbol{w}%
_{i,T+1}=(1,y_{iT},x_{i,T})^{\prime }$ as an input, in the simulations
we set $\boldsymbol{w}_{i,T+1}$ as
$\boldsymbol{w}_{i,T+1}=  ( 1,E  ( y_{it}  ) +\kappa _{i}
\sqrt{\operatorname{Var}(y_{it})},\mu _{xi}+\kappa _{i}\sigma _{xi}  ) ^{
\prime }$, where $E  ( y_{it}  ) $ and
$%
\operatorname{Var}(y_{it})$ are derived by assuming $y_{it}$ is stationary and
conditional on the model's parameters.\footnote{It is easily established that
$\mathrm{E}  ( y_{it}  ) =
\frac{\alpha _{i}+\beta _{i}\mu _{xi}}{1-\beta _{i}}$ and
$\operatorname{Var}(y_{it})=\frac{%
\sigma _{i}^{2}}{1-\beta _{i}^{2}}+  (
\frac{\gamma _{i}^{2}\sigma _{xi}^{2}}{1-\beta _{i}^{2}}  )
  ( 1+\frac{2\beta _{i}\rho _{xi}}{%
1-\beta _{i}\rho _{xi}}  ) $.}

The panel forecasts are evaluated using the ratio of the average MSFE of
method $j$ (pooled, fixed effects, random effects, empirical Bayes, and
the forecast combinations) measured relative to that of the reference individual
forecasts
\begin{equation*}
\text{rMSFE}_{j}= \frac{\frac{1}{NR}\sum
_{i=1}^{N}\sum_{r=1}^{R}(y_{i,T+1,r}-%
\hat{y}_{i,T+1,j,r})^{2}}{\frac{1}{NR}\sum
_{i=1}^{N}%
\sum
_{r=1}^{R}(y_{i,T+1,r}-
\hat{y}_{i,T+1,b,r})^{2}},
\end{equation*}
where $b$ denotes the benchmark forecast, which is the individual forecast.
Replications are denoted by $r=1,2,\ldots ,R$, where $R=10{,}000$.

\subsection{Simulation results}

Monte Carlo simulation results are reported in Table~\ref{tbl:MC_ARX_N100}.
Our theoretical analysis shows that the term $h_{NT}$ that adversely affects
forecasts from the individual estimates depends on the value of
$ \llVert  \boldsymbol{w}_{i,T+1}-\mathrm{E}[\boldsymbol{w}_ {i,T+1}]
 \rrVert  $, with small values of these deviations leading to better
forecasting performance for the individual estimates. To examine this effect,
we present two sets of conditional forecasting performance results, namely
for $\kappa _{i}=0$, that is, when $\boldsymbol{w}_{i,T+1}$ is set to its
mean
$\mathrm{E}%
(\boldsymbol{w}_{it})=  ( 1,\mathrm{E}  ( y_ {it}  ) ,
\mu _{xi}+\kappa _{i}\sigma _{xi}  ) $ in the top panel and when
$\boldsymbol{%
w}_{i,T+1}$ deviates from its mean by generating forecasts conditional
on
$%
\boldsymbol{w}_{i,T+1}=  ( 1,\mathrm{E}  ( y_ {it}  ) +
\kappa _{i}%
\sqrt{\operatorname{Var}(y_{it})},\mu _{xi}+\kappa _ {i}\sigma _{xi}  )
^{\prime }$ in the bottom panel. We set $\kappa _ {i}=1$ for
$i\leq N/2$, and $\kappa _{i}=-1$, for $i>N/2$.

We vary the parameter that controls the degree of correlated heterogeneity ($
\rho _{\gamma x}$) across three blocks of results and examine different
combinations of the two hyperparameters that determine the degree of heterogeneity,
$%
a_{\beta }$ and $\sigma _{\alpha }^{2}$. Finally, we vary the time-series
dimension ($T$) along the columns.

\input{./tables/MC_N100_mainBody_April2025.tex}

With little heterogeneity and a small time-series dimension, $T=20$, consistent
with Propositions \ref{prop:individual} and \ref{prop:pooled}, pooling
yields an MSFE up to 25\% lower than the individual forecast with the gain
being largest when the predictor is far from its mean ($\kappa _{i}=
\pm 1$). However, the advantage of the pooled forecasts over the individual
forecasts vanishes quickly for the two larger values of $T$ and turns to
distinctly worse performance under larger parameter heterogeneity---particularly
when the predictors are away from their means.

The RE estimator produces the most accurate forecasts when parameter heterogeneity
is limited to the intercept ($a_{\beta }=0$, $\sigma _{\alpha }^{2}=0.5$)
and the predictor is far from its mean. When slope coefficients are heterogeneous,
this method yields quite poor forecasting performance that deteriorates
with $T$. Similar findings hold for the forecasts based on the FE method.
Forecast accuracy for both RE and FE methods tend to worsen (relative to the benchmark
forecasts) under correlated heterogeneity.

Regardless of the level of heterogeneity in parameters (whether correlated
or not), the empirical Bayes forecasts perform very well particularly for
the smallest sample size ($T=20$). Unlike forecasts based on the pooled,
RE or FE estimators, the empirical Bayes forecasts have the attractive
feature that they never perform worse, on average, than the benchmark.
These forecasts perform particularly well when the predictor is away from
its mean value.

Among the three forecast combinations, the cross-sectional averaging scheme
that combines the pooled and individual forecasts generally performs better
than the fixed effect combination scheme and also, in some cases, improves
on the EB forecasts. When $\boldsymbol{w}_{i,T+1}$ is far away from its
mean, $T$ is small, and parameter heterogeneity is high, the combination
scheme with individual weights performs particularly well, including relative
to the EB forecast.

\section{Empirical applications}

\label{sec:applications}

We next apply our set of panel forecasting methods to two empirical applications
on house price inflation in U.S. metropolitan areas and inflation in CPI
subindices. These applications represent quite different levels of in-sample
fit: For the CPI data, the pooled $\mathrm{R}^{2}$ ($%
\mathrm{PR}^{2}$) of our models is around 0.2 while for house prices it
exceeds 0.8.

\subsection{Measures of forecasting performance}

Our empirical applications compute the out-of-sample MSFE as
$\mathrm{MSFE}%
_{ij}=(T-T_{1})^{-1}\*\sum_{t=T_{1}}^{T-1}(y_{i,t+1}-\hat{y}_{i,j,t+1})^{2}$,
where $\hat{y}_{i,j,t+1}$ is the forecast of $y_{i,t+1}$ using method
$j$ and information known at time $t$. Each forecast in the test sample,
$\hat{y}%
_{i,j,t+1}$, is generated using a rolling estimation window of observations
$%
t-w+1,t-w,\ldots ,t$, where $w$ is the length of the rolling window, which
we set to $w=60$ in both applications. As in the simulations, we report
the ratio of the average MSFE of method $j$ relative to the average MSFE
for the benchmark forecasts ($b$) from the individual-specific model
$\mathrm{rMSFE}%
_{j}=  ( N^{-1}\sum_{i=1}^{N}\mathrm{MSFE}_{ij}  ) /  ( N^{-1}
\sum_{i=1}^{N}\mathrm{MSFE}_{{ib}}  ) $. We also report the proportion
of units in the cross-section for which each method produces a smaller
MSFE than the benchmark along with the proportion of units in the cross-section
for which each method has the smallest or largest MSFE value.

Similar to the simulation study, we distinguish between forecasts where
the regressors are close to their means and when they are one standard deviation away from their means. Unlike in the Monte Carlo experiments, the parameters are unknown
in the two applications, and we therefore select forecasts based on
$d_{i,T+1}=%
\hat{\boldsymbol{\theta }}_{i}^{\prime }\boldsymbol{w}_{i,T+1}$. Regressors
are said to be in the neighborhood of the mean of $d_{it}$ when
$|d_{i,T+1}-%
\bar{d}_{i}-\kappa _{i}s_{d}|<c\sigma _{d}$, where $\bar{d}_{i}$ is the
mean and $s_{d}$ the standard deviation of $d_{it}$ in the estimation sample,
$%
t=1,2,\ldots ,T$, and $c=0.1$. $\kappa _{i}=0$ then gives the results where
the predictors are close to their mean and $\kappa _{i}=\pm 1$ shows the
results when the predictors are one standard deviation away from the mean.
Additionally, we report results for all forecasts.

We examine the significance of any differences in forecast accuracy using
the \citet{DieMar1995} (DM) test of predictive accuracy both for
the panel as a whole and for the individual series. First, we use the panel
version of the DM test proposed by \cite{Pesetal2013}, which tests the
null that the MSFE generated by the individual forecasts, averaged both
across time and units, is equal in expectation to the equivalent MSFE generated
by the panel models.\footnote{The panel DM test first computes the difference between the cross-sectional
average squared forecast error at a given point in time for the benchmark
versus competing model. It then uses the time series of these average squared
forecast errors to compute Newey--West HAC standard errors that account
for serial dependencies.} Second, we apply the DM test to the $N$ forecasts
for individual units in the sample and report the number of significant
values in either direction and the number of insignificant test statistics.
The tests are set up so that negative values indicate that the panel forecasts
are more accurate than the individual forecasts, while positive values
of the DM tests indicate that the individual forecasts are more accurate.
For simplicity, we report results for all forecasts.

\subsection{U.S.\ house prices}

Our first application uses quarterly data on real house price inflation
in 377 U.S. Metropolitan Statistical Areas (MSAs) from the first quarter
of 1975 to the first quarter of 2023, which we obtain from the Freddie
Mac website.\footnote{For each MSA, house prices are calculated by deflating the Freddie Mac
house price index by the CPI.} Our forecasts target the one-quarter-ahead
MSA-level rate of house price log changes. After accounting for the necessary
presample and the estimation window, the first forecast is for 1991Q2 and
the last for 2023Q1, a total of 128 forecasts per MSA.

Our prediction model for the house price inflation rate in quarter
$t$ for MSA $i$, $y_{it}$, takes the form
%
\begin{equation}
y_{it}=\alpha _{i}+\beta _{i}y_{i,t-1}+
\beta _{i}^{\ast }y_{i,t-1}^{
\ast }+\gamma
_{Ri}\bar{y}_{i,t-1}^{(R)}+\gamma
_{Ci}\bar{y}_{t-1}^{(C)}+%
\varepsilon _{it}, \label{eq:spatial_model}
\end{equation}%
where $i=1,2,\ldots ,N$ denotes individual MSAs and
$t=1,2,\ldots ,T$  refers to the time period,
$y_{it}^{\ast }=\sum_{k=1,k\neq i}^{N}\omega _{ik}^{s}y_{kt}$ is the spatial
effect for a set of spatial weights $\omega _{ik}^{s}$,
$\bar{y}_{it}^{(R)}$ is the average house price inflation in the region
of unit $i$, and $\bar{y}_{t}^{(C)}$ is the countrywide average house price
inflation. The weights, $\omega _{ik}$, measure the spatial effect of house
prices in MSA $k$ on house prices in MSA $i$ and are based on geographic
distance, that is, $\omega _{ik}^{s}=v_{ik}/\sum_{k=1}^{N}v_{ik} $ and
$v_{ik}=1$ if MSAs $  ( i,k  ) $ are at most 100 miles apart and
is zero otherwise. We obtain the weights from the data set of \citet{Yan2021}
and exclude MSAs without neighbors within 100 miles, which leaves 362 MSAs
in our sample.

The top panel in Table~\ref{tbl:applications_msfe} reports the results.
The column labeled ``all'' shows results
averaged across the full test sample, while columns labeled
$\kappa _{i}=0$ and $\kappa _{i}=\pm 1$ show results for subsamples in
which the predictor vector is close to the mean and one standard deviation
away from the mean, respectively. In the first three columns, the first
row shows the cross-sectional average MSFE value for the forecasts based
on individual estimates. Subsequent rows report ratios of the mean of the
individual MSFE for the respective methods relative to the benchmark forecasts.
Values below unity show that the ratio of average MSFE performance (across
MSAs) is better for the method listed in the row than for the benchmark
while values above unity indicate the opposite. The next three columns
headed ``freq. beating benchmark'' report the proportion of MSAs for which the respective
methods have a smaller MSFE than the benchmark, while the columns headed
``freq. smallest MSFE'' and ``freq. largest MSFE'' show the proportion
of MSAs for which the respective methods have the smallest or largest MSFE
among all forecasting methods.

\input{./tables/applications_April_2025}

Across the full sample, the average MSFE ratio below one for the pooled,
RE, and FE forecasts. However, these methods do notably worse than the
forecasts based on individual estimates when the predictors are close to
their mean ($%
\kappa _{i}=0$). Empirical Bayes forecast produce the best overall MSFE
performance, reducing the MSFE of the benchmark by 10\%, followed by reductions
of 6--8\% among the three forecast combination schemes. The EB forecasts
perform particularly well when the predictors are far away from their mean.

While the proportional reductions in MSFE ratios may not seem very large,
they translate into very high frequencies of beating the benchmark. The
EB forecasts produce lower MSFE values than the benchmark for 94\% of the
housing price series followed by 92--94\% for the forecast combinations
but only 59--62\% for the pooled, RE, and FE forecasts.

Turning to evidence of individual forecasts being ``best'' or ``worst,'' for the full test sample
the benchmark forecasts only produce the smallest MSFE for 1\% of the variables
versus 19\% for the EB and 16\% for pooled forecast combination schemes.
Using this metric, again the benchmark forecasts perform much better when
the predictors are close to their sample mean for which they are most accurate
for 27\% of the MSAs versus 16\% and 8\% for the EB and pooled combinations,
respectively. Conversely, forecasts based on individual estimates are worst
overall for 57\% of the variables versus 1\% or less for the EB and forecast
combination schemes.

These results show that the EB and combination approaches offer the attractive
feature of not only improving on the MSFE values of the baseline
``on average'' but, equally importantly,
rarely producing markedly worse forecasts than the baseline and often generating
substantially better results. Interestingly, the risk of producing the
highest MSFE value is notably lower for the pooled combination and individual
weighted combination than for the EB forecasts when the predictors are
close to their mean.\footnote{Equal-weighted combinations also performs quite well in both of the empirical
applications, which is a known feature in the forecast combination literature.}

Figure~\ref{fig:densities} summarizes our findings visually through density
plots fitted to the cross-sectional distribution of MSFE ratios for our
forecasting methods.\footnote{To reduce the number of lines, we do not plot the densities for the FE
and RE approaches, which are very similar to those from pooling.} MSFE
ratios have a widely dispersed, right-skewed distribution for the pooled
forecasts compared to the Bayesian and combination approaches whose distributions
are far more peaked and centered just below unity. This feature is highly
undesirable as it raises the likelihood of very poor forecasts for an individual
housing price series compared with that of the Bayesian and combination
approaches.\footnote{The impressive performance of the EB approach for the tail groups is consistent
with \citet{Efr2011}.}

\begin{figure}[tbp]
\caption{Distributions of ratios of MSFEs}
\label{fig:densities}\centering
\hspace*{1.5cm}{\footnotesize {House price forecasts \hspace{3cm} CPI
subindices forecasts}\newline
\hspace*{.3cm} 
\includegraphics[scale=.4,trim={2.1cm 7.5cm 2.1cm
7.5cm},clip]{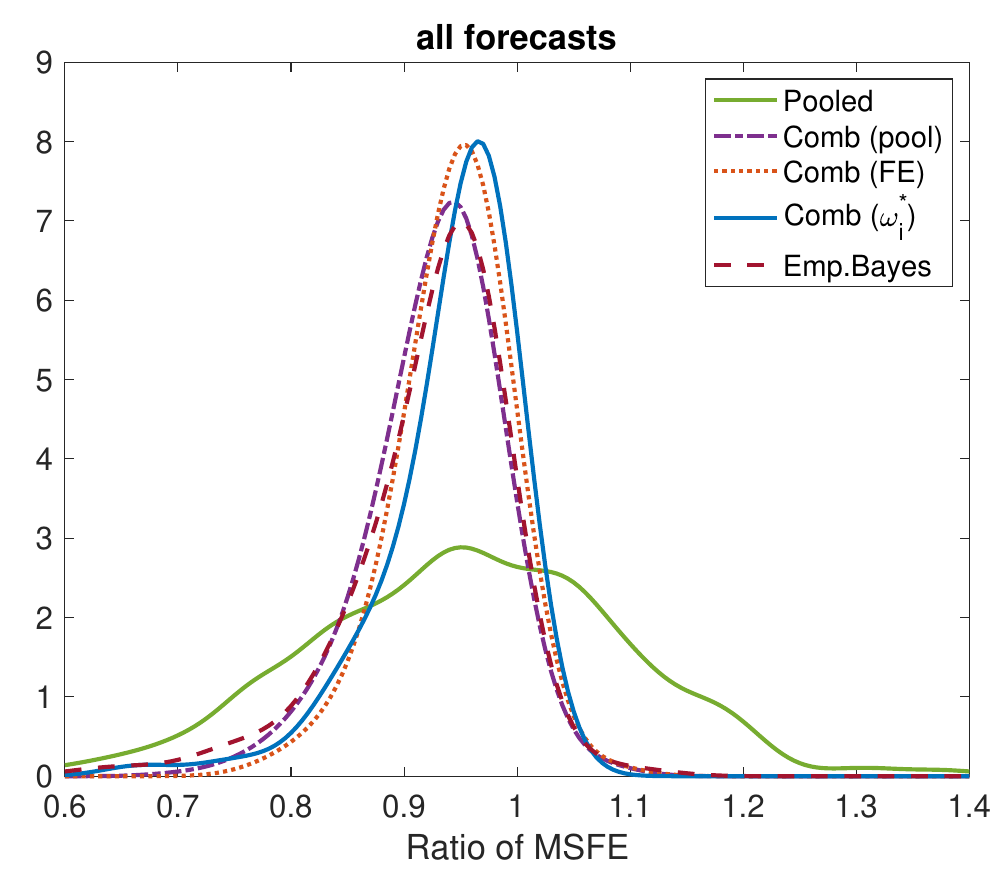}%
\includegraphics[scale=.4,trim={2.1cm 7.5cm 2.1cm
7.5cm},clip]{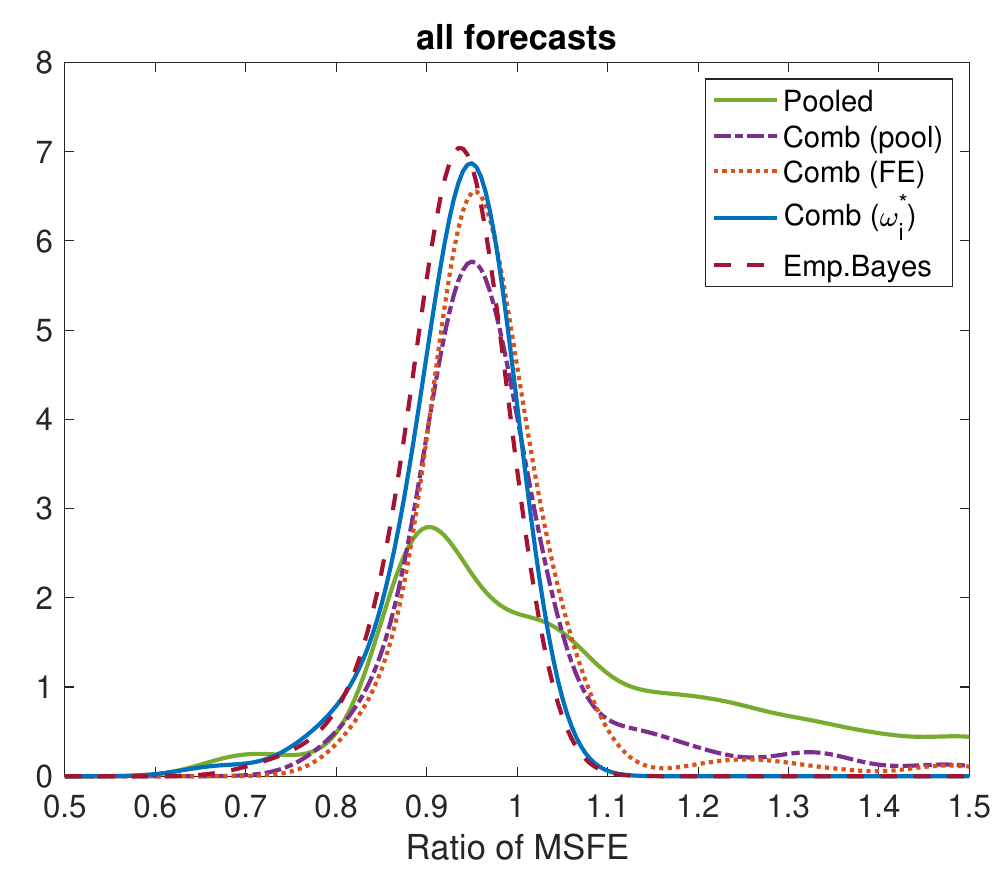}\newline
\hspace*{.3cm} 
\includegraphics[scale=.4,trim={2.1cm
7.5cm 2.1cm 7.5cm},clip]{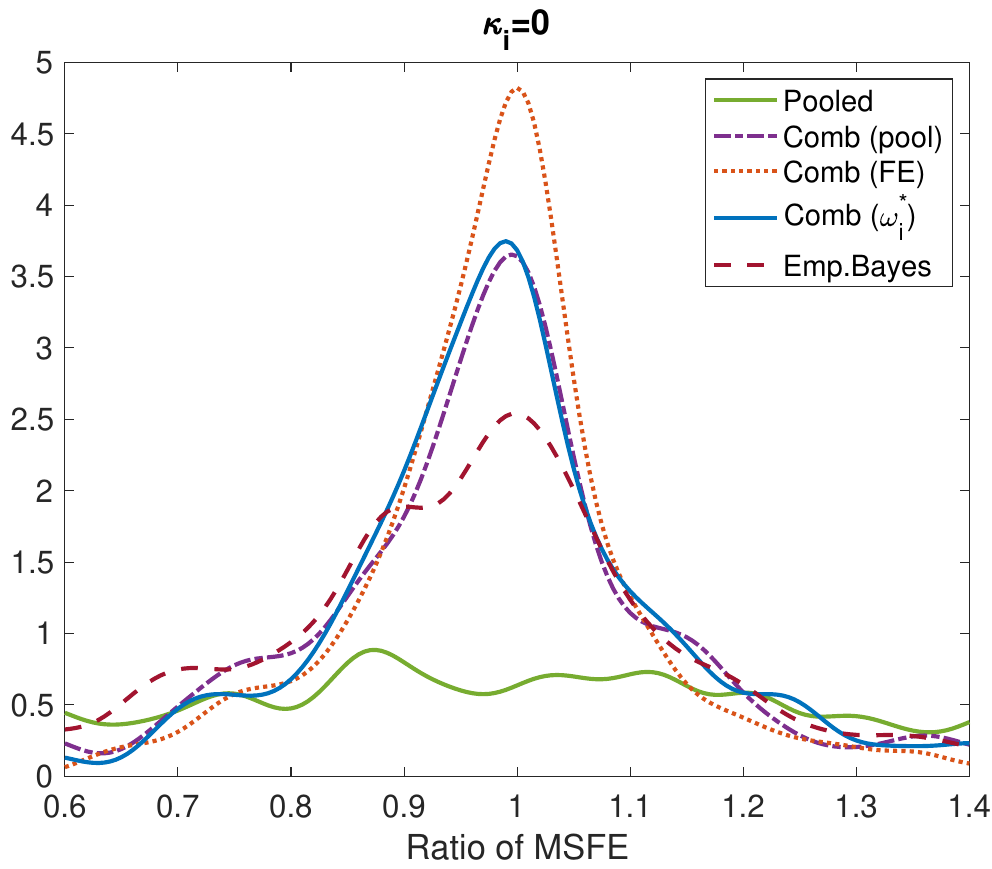}%
\includegraphics[scale=.4,trim={2.1cm 7.5cm 2.1cm
7.5cm},clip]{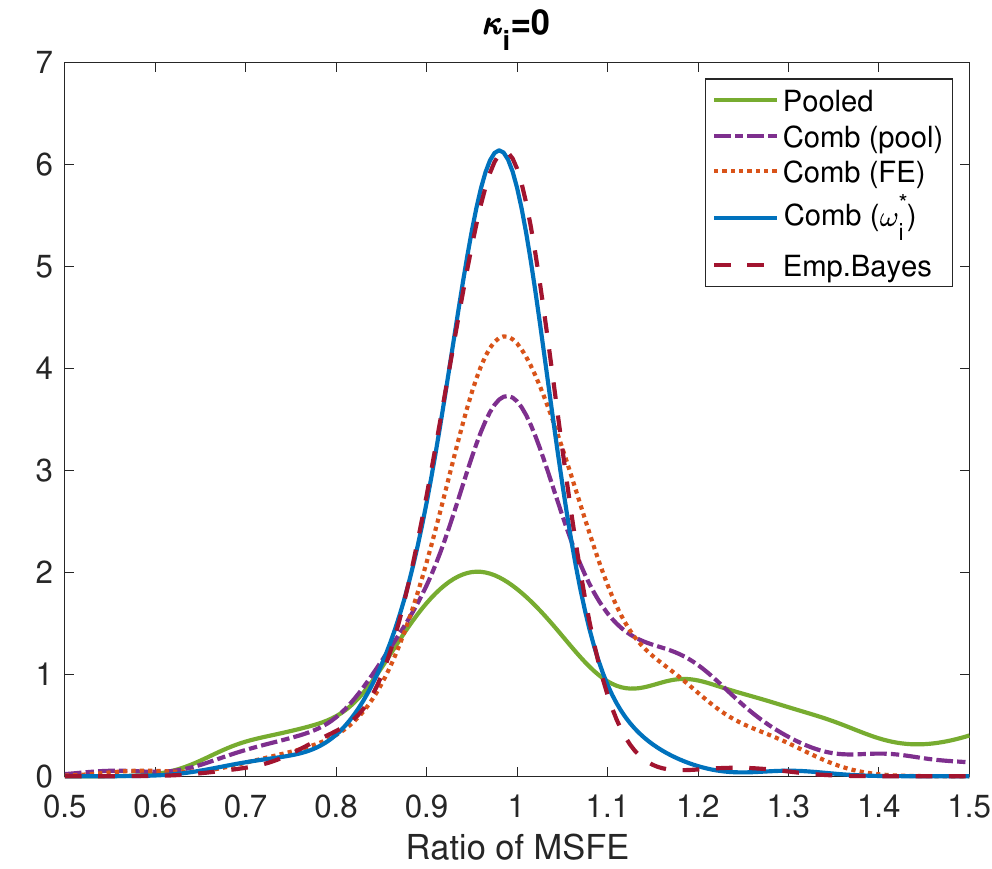}\newline
\hspace*{.3cm} 
\includegraphics[scale=.4,trim={2.1cm 7.2cm 2.1cm
7.5cm},clip]{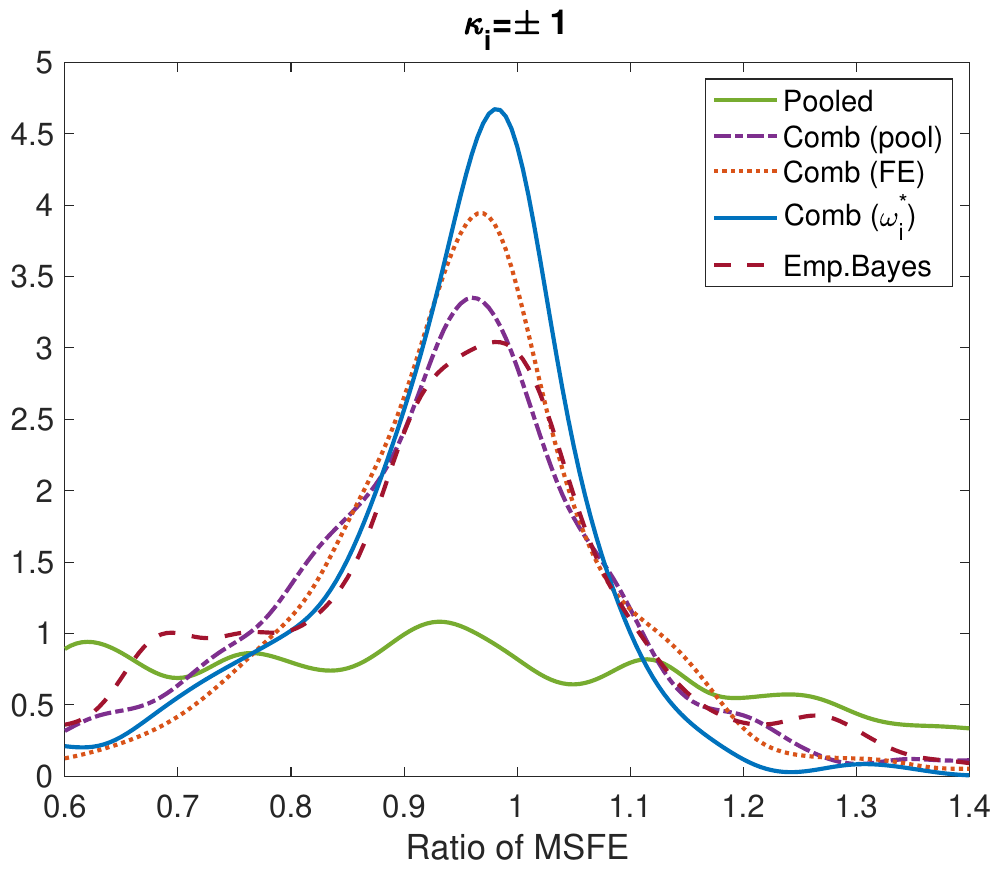}%
\includegraphics[scale=.4,trim={2.1cm 7.2cm 2.1cm
7.5cm},clip]{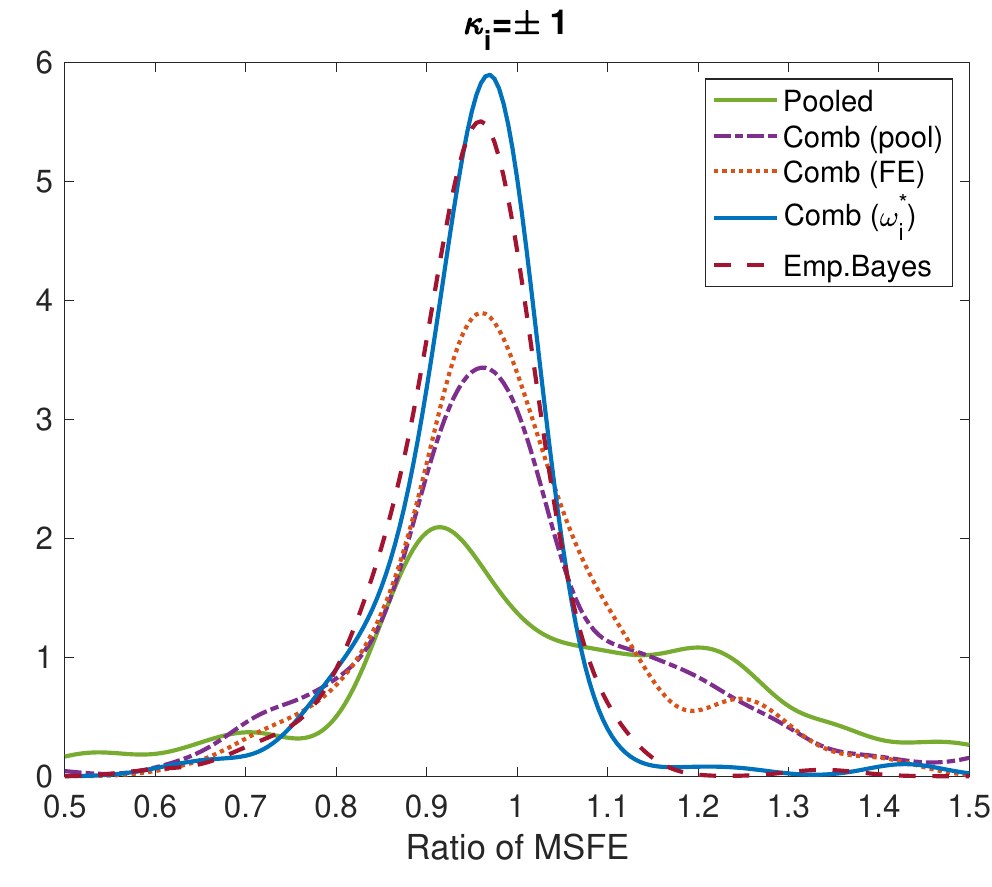}\newline
} 
{\footnotesize \vspace{1em} 
\parbox{13cm}{\footnotesize{Notes:
The graphs show density plots of the ratios of MSFEs for the house price
application in the left column and those for the CPI subindices application
in the right column. In the first row are the density plots for the MSFEs
from all forecasts, in the second row for the forecasts for which 
$d_{i,T+1} = \hat{\boldsymbol{\theta}}_{i}^{\prime}\boldsymbol{w}_{i,T+1}$ 
is close to its mean in the estimation sample, and in
the third row for the forecast for which $d_{i,T+1}$ is close to plus or minus one
standard deviation from its mean in the estimation sample. The
density estimates use a normal kernel with a bandwidth 0.04. The forecasting
methods are listed in the footnote of Tables~\ref{tbl:MC_ARX_N100}.
}} }
\end{figure}

The first and second rows of Table~\ref{tbl:DM} reports panel DM test statistics
and the number of cross-sectional units with a DM test below
$%
-1.96$ (panel forecasts are significantly more accurate) or above 1.96
(individual-specific forecasts are significantly more accurate), respectively,
for each application.

\input{./tables/DMtestStats_20April25.tex}

The panel DM tests show that the EB and combination forecasts are significantly
more accurate than the individual forecasts ``on average''
as well as for a large portion of the individual series (between 169 and
240 MSAs), while the opposite only happens for two individual MSAs in the
case of the EB forecasts. Pooled, RE and FE panel forecasts are also significantly
more accurate than the individual forecasts on average as well as for between
57 and 62 of the individual MSAs and significantly less accurate for very
few MSAs.


\subsection{CPI inflation of sub-indices}

Our second application covers inflation rates for up to 187 subindices
of the U.S. consumer price index (CPI) obtained from the FRED database.
The data is measured at the monthly frequency and spans the period from
January 1967 to December 2022. Again, we use rolling estimation windows
with 60 observations and require each estimation sample to be balanced,
excluding individual series without a complete set of observations in a
given window. After accounting for the necessary pre-samples, we generate
up to 599 forecasts for each series, with the first forecast computed for
February 1973.

We consider an autoregressive forecasting specification with lags 1, 2,
and 12 augmented with lagged values of the first principal component of
the data, the default yield and term spread.

The bottom panel of Table~\ref{tbl:applications_msfe} shows that, for the
full test sample, all forecasting methods produce lower MSFE values than
the benchmark. The pooled, RE, FE, and EB forecasts reduce the average
MSFE of the benchmark by around 12\%, while the forecast combination methods
reduce it by 7--10\%. Interestingly, when the predictors are close to their
sample mean, the lowest MSFE ratios are produced by the three forecast
combination methods, while conversely the EB scheme performs best when
the predictors are further removed from their mean.

The EB forecasting scheme performs particularly well overall, beating the
benchmark model's accuracy for 98\% of the variables followed by 97\% for
the individual weights, 73--79\% for the pooled and FE combinations and
around 50\% for the RE and FE schemes. As in the first application, these
percentages are notably lower for predictors close to their mean and higher
further away.

The EB forecasts also produce by far the highest frequency with the smallest
MSFE values overall (39\%) followed by 20\% for the pooled and individual
forecast combination scheme. This is matched by very low probabilities
of producing the worst forecast, which never occurs in our sample for the
EB method or any of the three forecast combination schemes but is far more
likely to occur for the benchmark (43.9\%) and pooled forecasts (39.6\%).

Our evidence is summarized by the probability density plots for the MSFE ratios
in the right panels of Figure~\ref{fig:densities}. The figure clearly highlights
the pronounced dispersion and thick right tails of the MSFE-ratio distribution
for the pooled forecasts. The distributions of MSFE ratios of the EB and
combination approaches are far more concentrated and less asymmetrical.
For values of the predictors farther away from the mean, the tails of the
densities are somewhat thicker, with the EB approach standing out as having
the thinnest right tail, and hence, the lowest probability of generating
forecasts less accurate than those from the individual-specific benchmark.

Turning to the DM test results for the CPI inflation data in Table~\ref{tbl:DM}, all panel models generate significantly negative DM panel
test statistics and so their associated forecasts are significantly more
accurate, on average, than the individual forecasts. The pooled, RE, and
FE models perform somewhat worse in this application, as the number of
individual CPI series for which their forecasts are significantly more
accurate than the individual-specific forecasts is smaller than those for
which the opposite holds. Conversely, the EB and combination forecasts
continue to be significantly more accurate than the benchmark forecasts
for between 50 and 137 of the individual CPI series and are only significantly
less accurate for between zero and 23 series. The EB and individual combination
approaches perform particularly well in this application.

\section{Conclusion}

We provide a comprehensive examination of the out-of-sample predictive
accuracy of a large set of novel and existing panel forecasting methods,
including individual estimation, pooled estimation, random effects, fixed
effects, empirical Bayes, and forecast combinations.

Our main findings can be summarized in three points. First, we find that
many panel forecasting approaches perform systematically better than forecasts
based on individual estimates. For panels with a small or medium-sized
time-series dimension $T$---a setting relevant to many empirical applications
in economics---our Monte Carlo simulations and empirical applications demonstrate
sizeable gains both on average and for the majority of individual units
from exploiting panel information.

Second, our analytical results and Monte Carlo simulations show that one
should not expect a single forecasting approach to be uniformly dominant
across applications that differ in terms of the cross-sectional and time-series
dimensions, strength of predictive power, and degree of heterogeneity in
intercept and slope coefficients along with how correlated this heterogeneity
is.

Forecasts based on pooled estimates are most accurate only in situations
with little or no parameter heterogeneity and a small $T$ dimension, while
forecasts based on FE and RE estimates perform relatively well mainly when
heterogeneity is confined to model intercepts and $T$ is small. Neither
of these approaches perform well in settings with high levels of heterogeneity
where individual-specific forecasts tend to perform better, particularly
if $%
T$ is relatively large. By overweighting forecasts that perform well and
underweighting forecasts that perform poorly, forecast combination and
empirical Bayes methods manage to produce the most accurate forecasts across
a broad range of settings.

Third, the panel forecasting methods differ in terms of their ability to
reduce the probability of generating very poor forecasts for individual
units in a cross-section. While the individual, pooled, random and fixed
effect estimation methods perform poorly in some of the simulations and
empirical applications, the forecast combination and empirical Bayes methods
rarely generate the least accurate forecasts for individual units and retain
some probability of being the best forecasting method. These panel forecasting
approaches therefore come out on top of our analysis.

In a nutshell, our simulations and empirical applications suggest that
forecast combinations and Bayesian panel methods offer insurance against
poor performance. Compared to the alternative forecasting methods we consider,
this better ``risk-return'' trade-off makes the combination and Bayes methods attractive in forecast
applications with panel data.

%
\newpage

\appendix

\section*{Mathematical appendix\label{A}}

\bigskip

\hrule

\bigskip

\section{Lemmas}

\begin{lemma}
\label{Lemma_1_VTEX1}%
Suppose that Assumptions \ref{ass:3b} and \ref{ass:6} hold, then for a
fixed $T>T_{0}$ we have
%
\begin{equation}
\boldsymbol{\bar{Q}}_{NT}-\mathrm{E} ( \boldsymbol{
\bar{Q}}_{NT} ) =O_{p}\bigl(N^{-1/2}
\bigr),\quad  \text{and}\quad \boldsymbol{\bar{q}}_{NT}- \mathrm{E} (
\boldsymbol{\bar{q}}_{NT} ) =O_{p} \bigl(
N^{-1/2} \bigr) , \label{Qqnorms}
\end{equation}%
and%
%
\begin{equation}
\boldsymbol{\bar{Q}}_{NT}^{-1}-\mathrm{E} ( \boldsymbol{
\bar{Q}}%
_{NT} ) ^{-1}=O_{p}
\bigl(N^{-1/2}\bigr), \label{Qbarinv}
\end{equation}%
where
$\boldsymbol{\bar{Q}}_{NT}=N^{-1}\sum_{i=1}^{N}\boldsymbol{Q}_{iT}$,
$%
\boldsymbol{\bar{q}}_{NT}=N^{-1}\sum_{i=1}^{N}\boldsymbol{q}_{iT}$,
$%
\boldsymbol{Q}_{iT}=T^{-1}\sum_{t=1}^{T}\boldsymbol{w}_{it}
\boldsymbol{w}%
_{it}^{\prime }$, and
$\boldsymbol{q}_{iT}=T^{-1}\sum_{t=1}^{T}\boldsymbol{w}%
_{it}\boldsymbol{w}_{it}^{\prime }\boldsymbol{\eta }_{i}$. Further, under
Assumptions \ref{ass:3} and \ref{ass:4},%
%
\begin{equation}
\mathrm{E} ( \boldsymbol{\bar{Q}}_{NT} ) = \boldsymbol{
\bar{Q}}_{N}%
, \quad \text{and}\quad \mathrm{E} ( \boldsymbol{
\bar{q}}_{NT} ) = \boldsymbol{%
\bar{q}}_{N},
\label{Qqstat}
\end{equation}%
where
$\boldsymbol{\bar{Q}}_{N}=N^{-1}\sum_{i=1}^{N}\boldsymbol{Q}_{i}$,
$%
\boldsymbol{\bar{q}}_{N}=N^{-1}\sum_{i=1}^{N}\boldsymbol{q}_{i}$,
$%
\boldsymbol{Q}_{i}$ $\boldsymbol{=}$ E$  ( \boldsymbol{w}_{it}
\boldsymbol{%
w}_{it}^{\prime }  ) $, and $\boldsymbol{q}_{i}=\mathrm{E}  (
\boldsymbol{w}%
_{it}\boldsymbol{w}_{it}^{\prime }\boldsymbol{\eta }_{i}  ) $.
\end{lemma}
\begin{proof} Note that
\begin{align*}
\boldsymbol{\bar{Q}}_{NT}-\mathrm{E} ( \boldsymbol{
\bar{Q}}_{NT} ) &=N^{-1}\sum_{i=1}^{N}
\bigl[ \boldsymbol{Q}_{iT}-E ( \boldsymbol{Q}%
_{iT}
) \bigr] ,\quad \text{and}
\\
 \boldsymbol{\bar{q}}_{NT}- \mathrm{E} (
\boldsymbol{\bar{q}}_{NT} ) &=N^{-1}\sum
_{i=1}^{N} \bigl[ \boldsymbol{q}%
_{iT}-E
( \boldsymbol{q}_{iT} ) \bigr] .
\end{align*}%
Under Assumptions \ref{ass:3} and \ref{ass:6}, the elements of
$\boldsymbol{Q%
}_{iT}-\mathrm{E}  ( \boldsymbol{Q}_{iT}  ) $ and
$\boldsymbol{q}%
_{iT}-\mathrm{E}  ( \boldsymbol{q}_{iT}  ) $ are independently
distributed with mean zero and finite variances. Therefore, (\ref{Qqnorms})
follows. Also,
\begin{eqnarray*}
\bigl\llVert \boldsymbol{\bar{Q}}_{NT}^{-1}-\mathrm{E} (
\boldsymbol{\bar{Q%
}}_{NT} ) ^{-1} \bigr\rrVert
&=& \bigl\llVert \boldsymbol{\bar{Q}}_{NT}^{-1}%
 \bigl[ \boldsymbol{\bar{Q}}_{NT}-\mathrm{E} ( \boldsymbol{
\bar{Q}}%
_{NT} ) \bigr] \mathrm{E} ( \boldsymbol{
\bar{Q}}_{NT} ) ^{-1} \bigr\rrVert
\\
&\leq & \bigl\llVert \boldsymbol{\bar{Q}}_{NT}^{-1}
\bigr\rrVert \bigl\llVert \boldsymbol{\bar{Q}}_{NT}-\mathrm{E} (
\boldsymbol{\bar{Q}}_{NT} ) \bigr\rrVert \bigl\llVert \mathrm{E} (
\boldsymbol{\bar{Q}}_{NT} ) ^{-1} \bigr\rrVert ,
\end{eqnarray*}%
and, by Assumption~\ref{ass:3b},
$ \llVert  \boldsymbol{\bar{Q}}%
_{NT}^{-1} \rrVert  =\lambda _{\max }  (
\boldsymbol{\bar{Q}}%
_{NT}^{-1}  ) <C$, and
$ \llVert  \mathrm{E}  ( \boldsymbol{\bar{Q}}%
_{NT}  ) ^{-1} \rrVert  = \llVert  \boldsymbol{\bar{Q}}%
_{N}^{-1} \rrVert  =O(1)$. Hence,
$ \llVert  \boldsymbol{\bar{Q}}%
_{NT}^{-1}-\mathrm{E}  ( \boldsymbol{\bar{Q}}_{NT}  ) ^{-1}
 \rrVert  $ has the same order as
$ \llVert  \boldsymbol{\bar{Q}}%
_{NT}^{-1}-\mathrm{E}  ( \boldsymbol{\bar{Q}}_{NT}  ) ^{-1}
 \rrVert  =O_{p}(N^{-1/2})$, as required. Result (\ref{Qqstat}) follows
from the stationarity properties, $\boldsymbol{Q}_{i}$
$\boldsymbol{=%
}$ E$  ( \boldsymbol{w}_{it}\boldsymbol{w}_{it}^{\prime }  ) $
and $%
\boldsymbol{q}_{i}=\mathrm{E}  ( \boldsymbol{w}_{it}\boldsymbol{w}_{it}^{
\prime }%
\boldsymbol{\eta }_{i}  ) $.
\end{proof}
\begin{lemma}  %
\label{Lemma_2_VTEX1}%
Under Assumptions \ref{ass:1}--\ref{ass:6},
%
\begin{equation}
\sup_{i,T}\mathrm{E} \bigl\llVert \sqrt{T} ( \hat{
\boldsymbol{\theta}}_{i}-%
\boldsymbol{\theta
}_{i} ) \bigr\rrVert ^{s}<C, \quad s=1,2,
\label{normi}
\end{equation}%
where
$\hat{\boldsymbol{\theta}}_{i}-\boldsymbol{\theta }_{i}=(
\boldsymbol{W}%
_{i}^{\prime }\boldsymbol{W}_{i})^{-1}\boldsymbol{W}_{i}^{\prime }%
\boldsymbol{\varepsilon }_{i}$, and
%
\begin{equation}
\sup_{i,T}\mathrm{E} \llVert \tilde{\boldsymbol{\theta }}-
\boldsymbol{%
\theta }_{i} \rrVert <C, \label{normNT}
\end{equation}%
$\tilde{\boldsymbol{\theta }}-\boldsymbol{\theta }_{i}=$
$-\boldsymbol{\eta }%
_{i}+\boldsymbol{\bar{Q}}_{NT}^{-1}\boldsymbol{\bar{q}}_{NT}+
\boldsymbol{%
\bar{Q}}_{NT}^{-1}\boldsymbol{\bar{\xi}}_{NT}$.
\end{lemma}
\begin{proof} Since
$ \llVert  \sqrt{T}  ( \hat{\boldsymbol{\theta}}_{i}-%
\boldsymbol{\theta }_{i}  )  \rrVert  \leq  \llVert
\boldsymbol{Q}%
_{iT}^{-1} \rrVert   \llVert  T^{-1/2}\boldsymbol{W}_{i}^{
\prime }%
\boldsymbol{\varepsilon }_{i} \rrVert  $, then
\begin{equation*}
\bigl\llVert \sqrt{T} ( \hat{\boldsymbol{\theta}}_{i}- \boldsymbol{
\theta }%
_{i} ) \bigr\rrVert ^{2}\leq \bigl
\llVert \boldsymbol{Q}%
_{iT}^{-1} \bigr\rrVert
^{2} \bigl\llVert T^{-1/2}\boldsymbol{W}_{i}^{
\prime }%
\boldsymbol{\varepsilon }_{i} \bigr\rrVert ^{2},
\end{equation*}%
and by the Cauchy--Schwarz inequality
\begin{eqnarray*}
\sup_{i,T}\mathrm{E} \bigl\llVert \sqrt{T} ( \hat{
\boldsymbol{\theta}}_{i}-%
\boldsymbol{\theta
}_{i} ) \bigr\rrVert ^{2} &\leq & \Bigl( \sup
_{i,T}%
\mathrm{E} \bigl\llVert \boldsymbol{Q}_{iT}^{-1}
\bigr\rrVert ^{4} \Bigr) ^{1/2} \Bigl( \sup
_{i,T}\mathrm{E} \bigl\llVert T^{-1/2}
\boldsymbol{W}%
_{i}^{\prime }\boldsymbol{\varepsilon
}_{i} \bigr\rrVert ^{4} \Bigr) ^{1/2}
\\
&=& \Bigl\{ \sup_{i,T}\mathrm{E} \bigl[ \lambda
_{\max }^{4} \bigl( \boldsymbol{%
Q}_{iT}^{-1}
\bigr) \bigr] \Bigr\} ^{1/2} \Bigl( \sup_{i,T}
\mathrm{E}%
 \bigl\llVert T^{-1/2}\boldsymbol{W}_{i}^{\prime }
\boldsymbol{\varepsilon }%
_{i} \bigr\rrVert
^{4} \Bigr) ^{1/2}.
\end{eqnarray*}%
Both of the terms on the right-hand side of the above are bounded under
Assumption~\ref{ass:weak_exogeneity_b}, and we have $\sup_{i,T}\mathrm{E}
 \llVert  \sqrt{T}  ( \hat{\boldsymbol{\theta}}_{i}-
\boldsymbol{\theta }_{i}  )  \rrVert  ^{2}<C$. This result in turn
implies $\sup_{i,T}\mathrm{E} \llVert  \sqrt{T}  (
\hat{\boldsymbol{\theta}}_{i}-\boldsymbol{\theta }_{i}  )
 \rrVert  <C$, and result (\ref{normi}) follows. Regarding
$\tilde{%
\boldsymbol{\theta }}-\boldsymbol{\theta }_{i}$, we first note that
\begin{equation*}
\llVert \tilde{\boldsymbol{\theta }}-\boldsymbol{\theta }_{i}
\rrVert \leq \llVert \boldsymbol{\eta }_{i} \rrVert + \bigl\llVert
\boldsymbol{%
\bar{Q}}_{NT}^{-1} \bigr\rrVert
\llVert \boldsymbol{\bar{q}}%
_{NT} \rrVert + \bigl\llVert
\boldsymbol{\bar{Q}}_{NT}^{-1} \bigr\rrVert \llVert
\boldsymbol{\bar{\xi}}_{NT} \rrVert ,
\end{equation*}%
and
%
\begin{equation}
\mathrm{E} \llVert \tilde{\boldsymbol{\theta }}- \boldsymbol{\theta
}%
_{i} \rrVert \leq \mathrm{E} \llVert \boldsymbol{\eta
}_{i} \rrVert + \bigl( \mathrm{E} \bigl\llVert \boldsymbol{
\bar{Q}}_{NT}^{-1} \bigr\rrVert ^{2} \bigr)
^{1/2} \bigl( \mathrm{E} \llVert \boldsymbol{\bar{q}}%
_{NT}
\rrVert ^{2} \bigr) ^{1/2}+ \bigl( \mathrm{E} \bigl\llVert
\boldsymbol{%
\bar{Q}}_{NT}^{-1} \bigr\rrVert
^{2} \bigr) ^{1/2} \bigl( \mathrm{E} \llVert \boldsymbol{
\bar{\xi}}_{NT} \rrVert ^{2} \bigr) ^{1/2}.
\label{normNT1}
\end{equation}%
Under Assumption~\ref{ass:4}, E$ \llVert  \boldsymbol{\eta }_{i}
 \rrVert  <C$ and $\sup_{i,t}\mathrm{E} \llVert  \boldsymbol{w}_{it}
\boldsymbol{w}%
_{it}^{\prime }\boldsymbol{\eta }_{i} \rrVert  ^{2}<C$. Also, by the
Cauchy--Schwarz inequality,
\begin{equation*}
\mathrm{E} \llVert \boldsymbol{w}_{it}\varepsilon _{it}
\rrVert ^{2}\leq \bigl( \mathrm{E} \llVert \boldsymbol{w}_{it}
\rrVert ^{4} \bigr) ^{1/2} \bigl( \llvert \varepsilon
_{it} \rrvert ^{4} \bigr) ^{1/2},
\end{equation*}%
and under Assumptions \ref{ass:1} and \ref{ass:3}, we have
$\sup_{i,t}\mathrm{E}%
 \llVert  \boldsymbol{w}_{it}\varepsilon _{it} \rrVert  ^{2}<C$. Then,
applying Minkows\-ki's inequality to
$\boldsymbol{\bar{\xi}}%
_{NT}=N^{-1}T^{-1}\sum_{i=1}^{N}\sum_{t=1}^{T}\boldsymbol{w}_{it}
\varepsilon _{it}$,
\begin{equation*}
\mathrm{E} \llVert \boldsymbol{\bar{\xi}}_{NT} \rrVert
_{2}= \bigl( \mathrm{E} \llVert \boldsymbol{\bar{
\xi}}_{NT} \rrVert ^{2} \bigr) ^{1/2}\leq
N^{-1}T^{-1}\sum_{i=1}^{N}
\sum_{t=1}^{T} \mathrm{E} \llVert
\boldsymbol{w}_{it}\varepsilon _{it} \rrVert
_{2}\leq \sup_{i,t} \bigl( \mathrm{E} \llVert
\boldsymbol{w}_{it}\varepsilon _{it} \rrVert
^{2} \bigr) ^{1/2},
\end{equation*}%
and it follows that E$ \llVert  \boldsymbol{\bar{\xi}}_{NT} \rrVert  ^{2}<C$. Similarly, since
$\boldsymbol{\bar{q}}_{NT}=N^{-1}T^{-1}%
\sum_{i=1}^{N}\sum_{t=1}^{T}\boldsymbol{w}_{it}\boldsymbol{w}_{it}^{
\prime }%
\boldsymbol{\eta }_{i}$ and
$\sup_{i,t}\mathrm{E} \llVert  \boldsymbol{w}%
_{it}\boldsymbol{w}_{it}^{\prime }\boldsymbol{\eta }_{i} \rrVert  ^{2}<C$,
then
$\mathrm{E} \llVert  \boldsymbol{\bar{q}}_{NT} \rrVert  ^{2}<C$. Also,
by Assumption~\ref{ass:3b},
$ \llVert  \boldsymbol{\bar{Q}}%
_{NT}^{-1} \rrVert  ^{2}=\lambda _{\max }  (
\boldsymbol{\bar{Q}}%
_{NT}^{-2}  ) <C$. Using these results in (\ref{normNT1}) now yields
(%
\ref{normNT}).
\end{proof}

\section{Proofs of the propositions}\label{Proofs}

\subsection{Proof of Proposition~\protect\ref{prop:individual}}\label%
{Proof_individual}

Let
$\boldsymbol{P}_{i}=\boldsymbol{W}_{i}(\boldsymbol{W}_{i}^{\prime }%
\boldsymbol{W}_{i})^{-1}$. Using (\ref{r_iT}), note that
\begin{equation*}
\mathrm{E}\lleft( r_{iT}\vert \boldsymbol{\varepsilon
}_{i}, \boldsymbol{W%
}_{i},
\boldsymbol{w}_{i,T+1}\rright. ) = \bigl( \boldsymbol{\varepsilon
}%
_{i}^{\prime }\boldsymbol{P}_{i}
\boldsymbol{w}_{i,T+1} \bigr) \mathrm{E}%
\lleft( \varepsilon
_{i,T+1}\vert \boldsymbol{\varepsilon }_{i},%
\boldsymbol{W}_{i},\boldsymbol{w}_{i,T+1}\rright. ) ,
\end{equation*}%
and, under Assumptions \ref{ass:1} and \ref{ass:weak_exogeneity_a},
$\mathrm{%
E}  ( \varepsilon _{i,T+1} \llvert  \boldsymbol{\varepsilon }_{i},%
\boldsymbol{W}_{i},\boldsymbol{w}_{i,T+1}    ) =0$, for all
$i$. Hence, unconditionally $\mathrm{E}  ( r_{iT }  ) =0$. Furthermore,
$%
 \llvert  r_{iT} \rrvert  \leq  \llVert
\boldsymbol{\varepsilon }%
_{i}^{\prime }\boldsymbol{P}_{i} \rrVert   \llVert
\boldsymbol{w}%
_{i,T+1} \rrVert   \llvert  \varepsilon _{i,T+1} \rrvert  $ and
$%
 \llvert  \varepsilon _{i,T+1} \rrvert  $ is distributed independently
of $\boldsymbol{w}_{i,T+1}$ and
$T^{-1}\boldsymbol{\varepsilon }_{i}^{\prime }\boldsymbol{P}_{i}$. Hence,
by the Cauchy--Schwarz inequality
\begin{equation*}
\mathrm{E} \llvert r_{iT} \rrvert \leq \bigl[ \mathrm{E} \bigl
\llVert \boldsymbol{\varepsilon }_{i}^{\prime }
\boldsymbol{P}_{i} \bigr\rrVert ^{2}%
 \bigr]
^{1/2} \bigl( \mathrm{E} \llVert \boldsymbol{w}_{i,T+1}
\rrVert ^{2} \bigr) ^{1/2}\mathrm{E} \llvert \varepsilon
_{i,T+1} \rrvert .
\end{equation*}%
Again, under Assumption~\ref{ass:1},
$\sup_{i,T}\mathrm{E} \llvert  \varepsilon _{i,T+1} \rrvert  <C$
and
$\sup_{i,T}\mathrm{E} \llVert  \boldsymbol{w}_{i,T+1} \rrVert  ^{2}<C$.
Also, since $\boldsymbol{Q}%
_{iT}^{-1}$ is symmetric,
$ \llVert  \boldsymbol{Q}_{iT}^{-1} \rrVert  ^{2}=\lambda _{
\max }^{2}  ( \boldsymbol{Q}_{iT}^{-1}  )$ and we have
%
\begin{eqnarray}
\bigl\llVert \boldsymbol{\varepsilon }_{i}^{\prime }
\boldsymbol{P}%
_{i} \bigr\rrVert ^{2} &=& \bigl
\llVert T^{-1} \boldsymbol{\varepsilon }%
_{i}^{\prime }
\boldsymbol{W}_{i}\bigl(T^{-1}\boldsymbol{W}_{i}^{\prime }%
\boldsymbol{W}_{i}\bigr)^{-1} \bigr\rrVert
^{2}\nonumber
\\
&\leq& T^{-1} \bigl\llVert \boldsymbol{Q}%
_{iT}^{-1}
\bigr\rrVert ^{2} \bigl\llVert T^{-1/2}
\boldsymbol{W}_{i}^{
\prime }%
\boldsymbol{\varepsilon
}_{i} \bigr\rrVert ^{2}\notag
\\
&\leq &\lambda _{\max }^{2} \bigl( \boldsymbol{Q}_{iT}^{-1}
\bigr) \bigl\llVert T^{-1}\boldsymbol{W}_{i}^{\prime }
\boldsymbol{\varepsilon }_{i} \bigr\rrVert ^{2}. \label{Pi}
\end{eqnarray}%
By the Cauchy--Schwarz inequality and under Assumption~\ref{ass:weak_exogeneity_b},
\begin{equation*}
\sup_{i,T}\mathrm{E} \bigl\llVert \boldsymbol{\varepsilon
}_{i}^{
\prime }%
\boldsymbol{P}_{i}
\bigr\rrVert ^{2}\leq \Bigl\{ \sup_{i,T}
\mathrm{E} \bigl[ \lambda _{\max }^{2} \bigl(
\boldsymbol{Q}_{iT}^{-1} \bigr) \bigr] \Bigr\}
^{1/2} \Bigl[ \sup_{i,T} \bigl\llVert
T^{-1} \boldsymbol{W}_{i}^{\prime }%
\boldsymbol{\varepsilon }_{i} \bigr\rrVert ^{4} \Bigr]
^{1/2}<C,
\end{equation*}%
and $\sup_{i,T}\mathrm{E} \llvert  r_{iT} \rrvert  <C$. Finally,
under Assumption~\ref{ass:6}, $r_{iT}\boldsymbol{ }$ are independently
distributed over $i$. Then, by the law of large numbers for independently
distributed processes with zero means, we have
%
\begin{equation}
R_{NT}=O_{p}\bigl(N^{-1/2}\bigr).
\label{RNT}
\end{equation}%
Consider now $S_{NT}$ and note that
\begin{equation*}
S_{NT}=N^{-1}\sum_{i=1}^{N}
\mathrm{E}(s_{iT})+N^{-1}\sum
_{i=1}^{N} \bigl[ s_{iT}-
\mathrm{E}(s_{iT}) \bigr] ,
\end{equation*}%
where $s_{iT}$ is given by (\ref{s_iT}). Under Assumption~\ref{ass:6},
$%
s_{iT}$ is distributed independently across $i$, and the second term of
$%
S_{NT}$ will be $O_{p}(N^{-1/2})$ if
$\sup_{i,T}\mathrm{E} \llvert  s_{iT} \rrvert  <C$. Also,
\begin{equation*}
\llvert s_{iT} \rrvert \leq \llVert \boldsymbol{w}%
_{i,T+1}
\rrVert ^{2} \bigl\llVert \boldsymbol{Q}_{iT}^{-1}
\bigr\rrVert ^{2} \bigl\llVert T^{-1/2}
\boldsymbol{W}_{i}^{\prime } \boldsymbol{\varepsilon
}%
_{i} \bigr\rrVert ^{2},
\end{equation*}%
and $\sup_{i,T} \llVert  \boldsymbol{w}_{i,T+1} \rrVert  <C$. Hence,
$%
\sup_{i,T}\mathrm{E} \llvert  s_{iT} \rrvert  <C$ follows if
\begin{equation*}
\sup_{i,T}\mathrm{E} \bigl[ \bigl\llVert
\boldsymbol{Q}_{iT}^{-1} \bigr\rrVert ^{2}
\bigl\llVert T^{-1/2}\boldsymbol{W}_{i}^{\prime }
\boldsymbol{\varepsilon }%
_{i} \bigr\rrVert
^{2} \bigr] <C.
\end{equation*}%
This condition is satisfied by Assumptions \ref{ass:1} and \ref{ass:weak_exogeneity_b},
noting that by the Cauchy--Schwarz inequality
\begin{equation*}
\mathrm{E} \bigl[ \bigl\llVert \boldsymbol{Q}_{iT}^{-1}
\bigr\rrVert ^{2} \bigl\llVert T^{-1/2}\boldsymbol{W}_{i}^{\prime }
\boldsymbol{\varepsilon }%
_{i} \bigr\rrVert
^{2} \bigr] \leq \bigl[ \mathrm{E} \bigl\llVert \boldsymbol{Q}%
_{iT}^{-1}
\bigr\rrVert ^{4} \bigr] ^{1/2} \bigl[ \mathrm{E} \bigl
\llVert T^{-1/2}%
\boldsymbol{W}_{i}^{\prime }
\boldsymbol{\varepsilon }_{i} \bigr\rrVert ^{4}%
 \bigr] ^{1/2},
\end{equation*}%
and
$ \llVert  \boldsymbol{Q}_{iT}^{-1} \rrVert  ^{4}=\lambda _{
\max }^{4}  ( \boldsymbol{Q}_{iT }^{-1}  ) $. Therefore,
$S_{NT}=\mathrm{E%
}  ( S_{NT}  ) +O_{p}(N^{-1/2})$, where
$\mathrm{E}  ( S_{NT}  ) =N^{-1}\sum_{i=1}^{N}\mathrm{E}(s_{iT})=h_{NT}$,
and the result in equation (\ref{MSFEI}) follows, with $h_{NT}$ given by
(\ref{h_NT}).

\subsection{Proof of Proposition~\protect\ref{prop:pooled}}\label%
{Proof_pooled}

The average MSFE of forecasts based on pooled estimates is given by (\ref{MSFEPool}),
which we reproduce here for convenience:
%
\begin{equation}
N^{-1}\sum_{i=1}^{N}
\tilde{e}_{i,T+1}^{2}=N^{-1}\sum
_{i=1}^{N} \varepsilon _{i,T+1}^{2}+N^{-1}
\sum_{i=1}^{N}\boldsymbol{w}_{i,T+1}^{
\prime }
\boldsymbol{%
\eta }_{i}\boldsymbol{\eta
}_{i}^{\prime }\boldsymbol{w}_{i,T+1}+
\tilde{S}%
_{N,T+1}+2\tilde{R}_{N,T+1},
\label{MSFEp}
\end{equation}%
where
%
\begin{eqnarray}
\tilde{S}_{N,T+1} &=&\boldsymbol{\bar{q}}_{NT}^{\prime }
\boldsymbol{\bar{Q}}%
_{NT}^{-1}\boldsymbol{
\bar{Q}}_{N,T+1}\boldsymbol{\bar{Q}}_{NT}^{-1}%
\boldsymbol{\bar{q}}_{NT}+\boldsymbol{\bar{\xi}}_{NT}^{\prime }
\boldsymbol{%
\bar{Q}}_{NT}^{-1}\boldsymbol{
\bar{Q}}_{N,T+1}\boldsymbol{\bar{Q}}_{NT}^{-1}%
\boldsymbol{\bar{\xi}}_{NT}\nonumber
\\
&&{}-2\boldsymbol{\bar{q}}_{NT}^{\prime }\boldsymbol{
\bar{Q}}_{NT}^{-1}  %
\boldsymbol{\bar{q}}_{N,T+1}-2\boldsymbol{\bar{\xi}}_{NT}^{\prime }%
\boldsymbol{\bar{Q}}_{NT}^{-1}\boldsymbol{
\bar{q}}_{N,T+1}\nonumber
\\
&&{}+2 \boldsymbol{\bar{%
\xi}}_{NT}^{\prime }
\boldsymbol{\bar{Q}}_{NT}^{-1} \boldsymbol{
\bar{Q}}%
_{N,T+1}\boldsymbol{\bar{Q}}_{NT}^{-1}
\boldsymbol{\bar{q}}_{NT}, \label{stilda}
\end{eqnarray}%
%
\begin{equation}
\tilde{R}_{N,T+1}=N^{-1}\sum_{i=1}^{N}
\boldsymbol{\eta }_{i}^{
\prime }%
\boldsymbol{w}_{i,T+1}\varepsilon _{i,T+1}- \bigl( \boldsymbol{
\bar{q}}%
_{NT}^{\prime }\boldsymbol{
\bar{Q}}_{NT}^{-1}+ \boldsymbol{\bar{\xi}}%
_{NT}^{\prime }
\boldsymbol{\bar{Q}}_{NT}^{-1} \bigr) \Biggl(
N^{-1} \sum_{i=1}^{N}
\boldsymbol{w}_{i,T+1}\varepsilon _{i,T+1} \Biggr) ,
\label{rtilda}
\end{equation}%
and%
%
\begin{equation}
\boldsymbol{\bar{Q}}_{N,T+1}=N^{-1}\sum
_{i=1}^{N}\boldsymbol{w}_{i,T+1}%
\boldsymbol{w}_{i,T+1}^{\prime },\quad \text{and}\quad \boldsymbol{
\bar{q}}%
_{N,T+1}=N^{-1}\sum
_{i=1}^{N}\boldsymbol{w}_{i,T+1}
\boldsymbol{w}%
_{i,T+1}^{\prime }\boldsymbol{\eta
}_{i}. \label{QqT+1}
\end{equation}%
Under Assumption~\ref{ass:2.2},
$\bar{\boldsymbol{\xi }}_{NT}=O_{p}  ( N^{-1/2}  ) $. Using Lemma~\ref{Lemma_1_VTEX1}, we have
$\boldsymbol{\bar{Q}}%
_{NT}^{-1}=\boldsymbol{\bar{Q}}_{N}^{-1}+O_{p}  ( N^{-1/2}  )
=O_{p}(1)$, and similarly $\boldsymbol{\bar{q}}_{NT}=O_{p}(1)$ and
$%
\boldsymbol{\bar{q}}_{N,T+1}=O_{p}(1)$. Using these results in (\ref{stilda}%
), we now have%
%
\begin{equation}
\tilde{S}_{N,T+1}=\boldsymbol{\bar{q}}_{NT}^{\prime }
\boldsymbol{\bar{Q}}%
_{NT}^{-1}\boldsymbol{
\bar{Q}}_{N,T+1}\boldsymbol{\bar{Q}}_{NT}^{-1}%
\boldsymbol{\bar{q}}_{NT}-2\boldsymbol{\bar{q}}_{NT}^{\prime }
\boldsymbol{%
\bar{Q}}_{NT}^{-1}\boldsymbol{
\bar{q}}_{N,T+1}+O_{p} \bigl( N^{-1/2} \bigr) .
\label{stilda1}
\end{equation}%
Note that under stationarity (see Assumptions \ref{ass:3} and
\ref{ass:4}),
$%
\mathrm{E}  ( \boldsymbol{w}_{i,T+1}\boldsymbol{w}_{i,T+1}^{
\prime }%
\boldsymbol{\eta }_{i}  ) =\boldsymbol{q}_{i}$,\break
$\mathrm{E}  ( \boldsymbol{w}_{i,T+1}\boldsymbol{w}_{i,T+1}^{
\prime }  ) =\boldsymbol{Q}%
_{i}$, and consider
\begin{align*}
\boldsymbol{\bar{q}}_{NT}^{\prime }\boldsymbol{
\bar{Q}}_{NT}^{-1} \boldsymbol{%
\bar{Q}}_{N,T+1}\boldsymbol{\bar{Q}}_{NT}^{-1}
\boldsymbol{\bar{q}}_{NT}={}& ( \boldsymbol{\Delta }_{q,NT}+
\boldsymbol{\bar{q}}_{N} ) ^{
\prime } \bigl( \boldsymbol{\Delta
}_{Q,NT}+\boldsymbol{\bar{Q}}_{N}^{-1} \bigr)
( \boldsymbol{\bar{Q}}_{N,T+1}-\boldsymbol{\bar{Q}}_{N}+
\boldsymbol{%
\bar{Q}}_{N} )
\\
& {}\times \bigl( \boldsymbol{\Delta }_{Q,NT}+\boldsymbol{
\bar{Q}}%
_{N}^{-1} \bigr) ( \boldsymbol{
\bar{q}}_{N,T+1}- \boldsymbol{\bar{q}}%
_{N}+
\boldsymbol{\bar{q}}_{N} ) ,
\end{align*}%
where (by Lemma~\ref{Lemma_1_VTEX1})
\begin{equation*}
\boldsymbol{\Delta }_{Q,NT}=\boldsymbol{\bar{Q}}_{NT}^{-1}-
\boldsymbol{\bar{Q%
}}_{N}^{-1}=O_{p}
\bigl(N^{-1/2}\bigr),\qquad \boldsymbol{\Delta
}_{q,NT}= \boldsymbol{%
\bar{q}}_{NT}-
\boldsymbol{\bar{q}}_{N}=O_{p}\bigl(N^{-1/2}
\bigr),
\end{equation*}%
and $\boldsymbol{\bar{Q}}_{N}$ and $\boldsymbol{\bar{q}}_{N}$ are defined
by (\ref{Qbar}) and (\ref{qbar}), respectively. Also, note that
\begin{eqnarray*}
\boldsymbol{\bar{Q}}_{N,T+1} &=&N^{-1}\sum
_{i=1}^{N}\mathrm{E} \bigl( \boldsymbol{w}_{i,T+1}
\boldsymbol{w}_{i,T+1}^{\prime } \bigr) +O_{p}
\bigl(N^{-1/2}\bigr)= \boldsymbol{\bar{Q}}_{N}+O_{p}
\bigl(N^{-1/2}\bigr),
\\
\boldsymbol{\bar{q}}_{N,T+1} &=&N^{-1}\sum
_{i=1}^{N}\mathrm{E} \bigl( \boldsymbol{w}_{i,T+1}
\boldsymbol{w}_{i,T+1}^{\prime } \boldsymbol{\eta
}%
_{i} \bigr) +O_{p}\bigl(N^{-1/2}
\bigr)=\boldsymbol{\bar{q}}_{N}+O_{p}
\bigl(N^{-1/2}\bigr).
\end{eqnarray*}%
Hence, it readily follows that
%
\begin{equation}
\boldsymbol{\bar{q}}_{NT}^{\prime }\boldsymbol{
\bar{Q}}_{NT}^{-1} \boldsymbol{%
\bar{Q}}_{N,T+1}\boldsymbol{\bar{Q}}_{NT}^{-1}
\boldsymbol{\bar{q}}_{NT}=%
\boldsymbol{\bar{q}}_{N}^{\prime }
\boldsymbol{\bar{Q}}_{N}^{-1} \boldsymbol{%
\bar{q}}_{N}+O_{p} \bigl( N^{-1/2} \bigr) .
\label{stilda2}
\end{equation}%
Similarly,
$\boldsymbol{\bar{q}}_{NT}^{\prime }\boldsymbol{\bar{Q}}_{NT}^{-1}%
\boldsymbol{\bar{q}}_{N,T+1}=\boldsymbol{\bar{q}}_{N}^{\prime }
\boldsymbol{%
\bar{Q}}_{N}^{-1}\boldsymbol{\bar{q}}_{N}+O_{p}  ( N^{-1/2}
  ) $, and as a result
\begin{equation*}
\tilde{S}_{N,T+1}=-\boldsymbol{\bar{q}}_{N}^{\prime }
\boldsymbol{\bar{Q}}%
_{N}^{-1}\boldsymbol{
\bar{q}}_{N}+O_{p} \bigl( N^{-1/2} \bigr) .
\end{equation*}%
Finally, since $\varepsilon _{i,T+1}$ (which has zero mean) is distributed
independently of $\boldsymbol{w}_{i,T+1}\boldsymbol{ }$ and
$\boldsymbol{%
\eta }_{i}$, under Assumption~\ref{ass:6},
\begin{equation*}
N^{-1}\sum_{i=1}^{N}
\boldsymbol{\eta }_{i}^{\prime }\boldsymbol{w}%
_{i,T+1}
\varepsilon _{i,T+1}=O_{p}\bigl(N^{-1/2}
\bigr),\quad \text{and}\quad N^{-1}\sum_{i=1}^{N}%
\boldsymbol{w}_{i,T+1}\varepsilon _{i,T+1}=O_{p}
\bigl(N^{-1/2}\bigr),
\end{equation*}%
and $\tilde{R}_{N,T+1}=O_{p}(N^{-1/2})$, noting that
$  ( \boldsymbol{%
\bar{q}}_{NT}^{\prime }\boldsymbol{\bar{Q}}_{NT}^{-1}+
\boldsymbol{\bar{\xi}}%
_{NT}^{\prime }\boldsymbol{\bar{Q}}_{NT}^{-1}  ) =O_{p}(1)$. Using
this result and (\ref{stilda2}) in (\ref{MSFEPool}) now yields
%
\begin{eqnarray}
N^{-1}\sum_{i=1}^{N}
\tilde{e}_{i,T+1}^{2}&=&N^{-1}\sum
_{i=1}^{N} \varepsilon _{i,T+1}^{2}+N^{-1}
\sum_{i=1}^{N}\boldsymbol{w}_{i,T+1}^{
\prime }
\boldsymbol{%
\eta }_{i}\boldsymbol{\eta
}_{i}^{\prime }\boldsymbol{w}_{i,T+1}\nonumber
\\
&&{}-
\boldsymbol{%
\bar{q}}_{N}^{\prime }\boldsymbol{
\bar{Q}}_{N}^{-1} \boldsymbol{\bar{q}}%
_{N}+O_{p}
\bigl( N^{-1/2} \bigr) . \label{MSFE_pooled2}
\end{eqnarray}%
Also, under Assumption~\ref{ass:6},
$\boldsymbol{w}_{i,T+1}^{\prime }%
\boldsymbol{\eta }_{i}\boldsymbol{\eta }_{i}^{\prime }\boldsymbol{w}_{i,T+1}$
is independently distributed over $i$ and we
have, noting that under Assumption~\ref{ass:4},
\begin{align*}
\sup_{i,T}\mathrm{E} \llvert  \boldsymbol{w}%
_{i,T+1}^{\prime }\boldsymbol{\eta }_{i}\boldsymbol{\eta }_{i}^{
\prime }%
\boldsymbol{w}_{i,T+1} \rrvert  &=\sup_{i,T}\mathrm{E} \llVert
\boldsymbol{w}%
_{i,T+1}^{\prime }\boldsymbol{\eta }_{i} \rrVert  ^{2}<C,
\\
N^{-1}\sum_{i=1}^{N}
\boldsymbol{w}_{i,T+1}^{\prime } \boldsymbol{\eta
}_{i}%
\boldsymbol{\eta }_{i}^{\prime }
\boldsymbol{w}_{i,T+1}&=N^{-1}\sum
_{i=1}^{N}%
\mathrm{E} \bigl(
\boldsymbol{w}_{i,T+1}^{\prime }\boldsymbol{\eta
}_{i}%
\boldsymbol{\eta }_{i}^{\prime }
\boldsymbol{w}_{i,T+1} \bigr) +O_{p} \bigl(
N^{-1/2} \bigr) .
\end{align*}%
Using this result in (\ref{MSFE_pooled2}) now yields equation (\ref{MSFEP1}%
).

To establish part (b) of Proposition~\ref{prop:pooled}, note that the first
term of $\Delta _{NT}$,
\[
N^{-1}\sum_{i=1}^{N}\mathrm{E}  ( \boldsymbol{w}%
_{i,T+1}^{\prime }\boldsymbol{\eta }_{i}\boldsymbol{\eta }_{i}^{
\prime }%
\boldsymbol{w}_{i,T+1}  ) =N^{-1}\sum_{i=1}^{N}\mathrm{E}  (
\boldsymbol{w}_{i,T+1}^{\prime }\boldsymbol{\eta }_{i}  ) ^{2}
\geq 0,
\]
arises irrespective of whether heterogeneity is correlated or
not. The second term,
$\boldsymbol{\bar{q}}_{N}\boldsymbol{\bar{Q}}_{N}^{-1}%
\boldsymbol{\bar{q}}_{N}^{\prime }$, enters only if heterogeneity is correlated.
The balance of the two terms ($\Delta _{NT}$) can be signed under stationarity
where
$\mathrm{E}  ( \boldsymbol{w}_{i,T+1}^{\prime }%
\boldsymbol{\eta }_{i}\boldsymbol{\eta }_{i}^{\prime }
\boldsymbol{w}%
_{i,T+1}  ) =\mathrm{E}  ( \boldsymbol{w}_{it}^{\prime }
\boldsymbol{%
\eta }_{i}\boldsymbol{\eta }_{i}^{\prime }\boldsymbol{w}_{it}  ) $.
In this case, we have
%
\begin{equation}
\Delta _{N}=N^{-1}\sum_{i=1}^{N}
\mathrm{E} \bigl( \boldsymbol{w}_{it}^{
\prime }\boldsymbol{\eta
}_{i}\boldsymbol{\eta }_{i}^{\prime }
\boldsymbol{w}%
_{it} \bigr) -\boldsymbol{
\bar{q}}_{N}^{\prime }\boldsymbol{\bar{Q}}_{N}^{-1}%
\boldsymbol{\bar{q}}_{N}. \label{DeltaSt_N}
\end{equation}%
To establish that the net effect of the two terms in $\Delta _{NT}$ is
nonnegative, we first show that the sample estimate of
$\Delta _{NT}$ can be obtained as the sum of squares of the residuals from
the pooled panel regression of
$\boldsymbol{\eta }_{i}^{\prime }\boldsymbol{w}_{it}$ on
$%
\boldsymbol{w}_{it}$. Consider the panel regression
$\boldsymbol{\eta }%
_{i}^{\prime }\boldsymbol{w}_{it}=\boldsymbol{\gamma }^{\prime }
\boldsymbol{w%
}_{it}+\nu _{it}$, and note that the pooled estimator of
$\boldsymbol{\gamma }$ is given by
\begin{equation*}
\boldsymbol{\hat{\gamma}}_{NT}\boldsymbol{=} \Biggl(
N^{-1}T^{-1}\sum_{i=1}^{N}
\sum_{t=1}^{T}\boldsymbol{w}_{it}
\boldsymbol{w}%
_{it}^{\prime } \Biggr)
^{-1}N^{-1}T^{-1}\sum
_{i=1}^{N}\sum_{t=1}^{T}%
\boldsymbol{w}_{it}\boldsymbol{w}_{it}^{\prime }
\boldsymbol{\eta }_{i}=%
\boldsymbol{\bar{Q}}_{NT}^{-1}
\boldsymbol{\bar{q}}_{NT},
\end{equation*}%
which yields the residual sum of squares
\begin{equation*}
N^{-1}T^{-1}\sum_{i=1}^{N}
\sum_{t=1}^{T}\hat{\nu}_{it}^{2}=N^{-1}T^{-1}%
\sum_{i=1}^{N}\sum
_{t=1}^{T} \bigl( \boldsymbol{\eta
}_{i}^{\prime }%
\boldsymbol{w}_{it}-
\boldsymbol{\hat{\gamma}}_{NT}^{\prime } \boldsymbol{w}%
_{it}
\bigr) ^{2}=\mathring{\Delta}_{NT}.
\end{equation*}%
By construction, $\mathring{\Delta}_{NT}$ is nonnegative and is given by
\begin{equation*}
\mathring{\Delta}_{NT}=T^{-1}N^{-1}\sum
_{t=1}^{T}\sum
_{i=1}^{N} \boldsymbol{w}%
_{it}^{\prime }
\boldsymbol{\eta }_{i}\boldsymbol{\eta }_{i}^{\prime }%
\boldsymbol{w}_{it}-\boldsymbol{\bar{q}}_{NT}^{\prime }
\boldsymbol{\bar{Q}}%
_{NT}^{-1}\boldsymbol{
\bar{q}}_{NT}\geq 0.
\end{equation*}%
This result also holds for a fixed $T$ and as $N\rightarrow \infty $ (applying
Slutsky's theorem to the second term):
\begin{equation*}
\lim_{N\rightarrow \infty }\mathring{\Delta}_{NT}=\plim
_{N
\rightarrow \infty }N^{-1}T^{-1}\sum
_{i=1}^{N}\sum_{t=1}^{T}
\hat{\nu}_{it}^{2}\geq 0.
\end{equation*}

\subsection{Proof of Proposition~\protect\ref{prop:combined}}\label%
{ProofCombined}

Using (\ref{ehat}) and (\ref{etilda}),
%
\begin{eqnarray}
N^{-1}\sum_{i=1}^{N}
\hat{e}_{i,T+1}\tilde{e}_{i,T+1} &=&N^{-1}\sum
_{i=1}^{N} \varepsilon
_{i,T+1}^{2}+N^{-1}\sum
_{i=1}^{N}( \hat{%
\boldsymbol{\theta
}_{i}}-\boldsymbol{\theta }_{i})^{\prime }
\boldsymbol{w}%
_{i,T+1}\boldsymbol{w}_{i,T+1}^{\prime }(
\tilde{\boldsymbol{\theta }}-%
\boldsymbol{\theta }_{i})
\nonumber
\\
&&{}-N^{-1}\sum_{i=1}^{N}(
\hat{\boldsymbol{\theta}}_{i}- \boldsymbol{\theta
}%
_{i})^{\prime }\boldsymbol{w}_{i,T+1}
\varepsilon _{i,T+1}\nonumber
\\
&&{}-N^{-1}\sum_{i=1}^{N}(
\tilde{\boldsymbol{\theta }}-\boldsymbol{%
\theta }_{i})^{\prime }
\boldsymbol{w}_{i,T+1}\varepsilon _{i,T+1},\label{ehattilda}
\end{eqnarray}%
where $\tilde{\boldsymbol{\theta }}-\boldsymbol{\theta }_{i}=$
$-\boldsymbol{%
\eta }_{i}+\boldsymbol{\bar{Q}}_{NT}^{-1}\boldsymbol{\bar{q}}_{NT}+%
\boldsymbol{\bar{Q}}_{NT}^{-1}\boldsymbol{\bar{\xi}}_{NT}$, and
$\hat{%
\boldsymbol{\theta}}_{i}-\boldsymbol{\theta }_{i}=(\boldsymbol{W}%
_{i}^{\prime }\boldsymbol{W}_{i})^{-1}\boldsymbol{W}_{i}^{\prime }%
\boldsymbol{\varepsilon }_{i}$. Noting that the third term in the above,
apart from the minus sign, is the same as $R_{NT}$ defined below (\ref{MSFE_i}),
by (\ref{RNT}) it follows that
%
\begin{equation}
N^{-1}\sum_{i=1}^{N}(\hat{
\boldsymbol{\theta}}_{i}- \boldsymbol{\theta }%
_{i})^{\prime }
\boldsymbol{w}_{i,T+1}\varepsilon _{i,T+1}=N^{-1}
\sum_{i=1}^{N}r_{iT}=R_{NT}=O_{p}
\bigl(N^{-1/2}\bigr) . \label{RNT1}
\end{equation}%
Further,
\begin{eqnarray*}
&&N^{-1}\sum_{i=1}^{N}(
\tilde{\boldsymbol{\theta }}- \boldsymbol{\theta }%
_{i})^{\prime }
\boldsymbol{w}_{i,T+1}\varepsilon _{i,T+1}
\\
 &&\quad =-N^{-1}
\sum_{i=1}^{N}\boldsymbol{\eta
}_{i}^{\prime }\boldsymbol{w}%
_{i,T+1}
\varepsilon _{i,T+1}+\boldsymbol{\bar{Q}}_{NT}^{-1}
\boldsymbol{\bar{q%
}}_{NT} \Biggl( N^{-1}\sum
_{i=1}^{N}\boldsymbol{w}_{i,T+1}
\varepsilon _{i,T+1} \Biggr)
\\
&&\qquad {}+\boldsymbol{\bar{Q}}_{NT}^{-1}\boldsymbol{\bar{
\xi}}_{NT} \Biggl( N^{-1} \sum
_{i=1}^{N}\boldsymbol{w}_{i,T+1}
\varepsilon _{i,T+1} \Biggr) .
\end{eqnarray*}%
By Lemma~\ref{Lemma_1_VTEX1}, $\boldsymbol{\bar{Q}}_{NT}^{-1}=O_{p}(1)$ and
$%
\boldsymbol{\bar{q}}_{NT}=O_{p}(1)$, and by Assumption~\ref{ass:2.2}
$,%
\boldsymbol{\bar{\xi}}_{NT}=O_{p}(N^{-1/2})$. Also, under Assumptions
\ref{ass:1} and \ref{ass:5},
$\boldsymbol{\eta }_{i}^{\prime }\boldsymbol{w}%
_{i,T+1}\varepsilon _{i,T+1}$ and
$\boldsymbol{w}_{i,T+1}\varepsilon _{i,T+1} $ have mean zero and first-order
moments. Hence, given Assumption~\ref{ass:6} we have%
%
\begin{equation}
N^{-1}\sum_{i=1}^{N}(
\tilde{\boldsymbol{\theta }}- \boldsymbol{\theta }%
_{i})^{\prime }
\boldsymbol{w}_{i,T+1}\varepsilon _{i,T+1}=O_{p}
\bigl( N^{-1/2} \bigr) . \label{SNT2}
\end{equation}%
Consider now the second term of (\ref{ehattilda}):
\begin{eqnarray*}
&&N^{-1}\sum_{i=1}^{N}(
\hat{\boldsymbol{\theta}}_{i}- \boldsymbol{\theta
}%
_{i})^{\prime }\boldsymbol{w}_{i,T+1}
\boldsymbol{w}_{i,T+1}^{\prime }( \tilde{%
\boldsymbol{
\theta }}-\boldsymbol{\theta }_{i})
\\
&&\quad =N^{-1}\sum_{i=1}^{N}
\bigl( -\boldsymbol{\eta }_{i}+ \boldsymbol{\bar{Q}}%
_{NT}^{-1}
\boldsymbol{\bar{q}}_{NT}+\boldsymbol{\bar{Q}}_{NT}^{-1}%
\boldsymbol{\bar{\xi}}_{NT} \bigr) ^{\prime }
\boldsymbol{w}_{i,T+1}%
\boldsymbol{w}_{i,T+1}^{\prime }
\bigl(\boldsymbol{W}_{i}^{\prime } \boldsymbol{W}%
_{i}
\bigr)^{-1}\boldsymbol{W}_{i}^{\prime }\boldsymbol{
\varepsilon }_{i}
\\
&&\quad =N^{-1}\sum_{i=1}^{N}
\bigl( -\boldsymbol{\eta }_{i}^{\prime }+ \boldsymbol{%
\bar{q}}_{NT}^{\prime }\boldsymbol{\bar{Q}}_{NT}^{-1}
\bigr) \boldsymbol{w}%
_{i,T+1}\boldsymbol{w}_{i,T+1}^{\prime }
\bigl(\boldsymbol{W}_{i}^{\prime }%
\boldsymbol{W}_{i}\bigr)^{-1}\boldsymbol{W}_{i}^{\prime }
\boldsymbol{\varepsilon }%
_{i}
\\
&&\qquad {}+\boldsymbol{\bar{\xi}}_{NT}^{\prime }\boldsymbol{
\bar{Q}}_{NT}^{-1} \Biggl[ N^{-1}\sum
_{i=1}^{N}\boldsymbol{w}_{i,T+1}
\boldsymbol{w}_{i,T+1}^{
\prime }\bigl(%
\boldsymbol{W}_{i}^{\prime }\boldsymbol{W}_{i}
\bigr)^{-1}\boldsymbol{W}%
_{i}^{\prime }
\boldsymbol{\varepsilon }_{i} \Biggr] ,
\end{eqnarray*}%
where, as noted above,
$\boldsymbol{\bar{\xi}}_{NT}^{\prime }\boldsymbol{%
\bar{Q}}_{NT}^{-1}=O_{p}  ( N^{-1/2}  ) $. Also, under stationarity
(Assumption~\ref{ass:3}) and using (\ref{Qqnorms}) and (\ref{Qbarinv})
(See Lemma~\ref{Lemma_1_VTEX1}),
$\boldsymbol{\bar{q}}_{NT}=\boldsymbol{\bar{q}}%
_{N}+O_{p}  ( N^{-1/2}  ) $ and
$\boldsymbol{\bar{Q}}_{NT}^{-1}=%
\boldsymbol{\bar{Q}}_{N}^{-1}+O_{p}  ( N^{-1/2}  ) $, and we have
\begin{equation*}
N^{-1}\sum_{i=1}^{N}(\hat{
\boldsymbol{\theta}}_{i}- \boldsymbol{\theta }%
_{i})^{\prime }
\boldsymbol{w}_{i,T+1}\boldsymbol{w}_{i,T+1}^{\prime }(
\tilde{%
\boldsymbol{\theta }}-\boldsymbol{\theta }_{i})= (
g_{1,nT}-g_{2,nT} ) +O_{p}\bigl(T^{-1/2}N^{-1/2}
\bigr),
\end{equation*}%
where
\begin{eqnarray*}
g_{1,NT} &=& \Biggl[ N^{-1}\sum
_{i=1}^{N}\boldsymbol{\varepsilon
}_{i}^{
\prime }\boldsymbol{W}_{i}\bigl(
\boldsymbol{W}_{i}^{\prime }\boldsymbol{W}_{i}
\bigr)^{-1}%
\boldsymbol{w}_{i,T+1}
\boldsymbol{w}_{i,T+1}^{\prime } \Biggr] \boldsymbol{%
\bar{Q}}_{N}^{-1}\boldsymbol{\bar{q}}_{N},
\\
g_{2,NT} &=&N^{-1}\sum_{i=1}^{N}
\boldsymbol{\varepsilon }_{i}^{
\prime }%
\boldsymbol{W}_{i}\bigl(\boldsymbol{W}_{i}^{\prime }
\boldsymbol{W}_{i}\bigr)^{-1}%
\boldsymbol{w}_{i,T+1}\boldsymbol{w}_{i,T+1}^{\prime }
\boldsymbol{\eta }_{i}.
\end{eqnarray*}%
We also note that under Assumptions \ref{ass:3},
\ref{ass:weak_exogeneity_b}%
, \ref{ass:3b}, and \ref{ass:6},
\begin{equation*}
g_{1,NT}=\mathrm{E} ( g_{1,nT} ) +O_{p} \bigl(
N^{-1/2} \bigr) ,%
\quad \text{and}\quad g_{2,NT}=
\mathrm{E} ( g_{2,NT} ) +O_{p} \bigl( N^{-1/2}
\bigr) .
\end{equation*}%
Hence,
\begin{align*}
&N^{-1}\sum_{i=1}^{N}(\hat{
\boldsymbol{\theta}}_{i}- \boldsymbol{\theta }%
_{i})^{\prime }
\boldsymbol{w}_{i,T+1}\boldsymbol{w}_{i,T+1}^{\prime }(
\tilde{%
\boldsymbol{\theta }}-\boldsymbol{\theta }_{i})
\\
&\quad =
\mathrm{E} ( g_{1,nT} ) -\text{E} ( g_{2,NT} )
+O_{p} \bigl( N^{-1/2} \bigr) +O_{p} \bigl(
T^{-1/2}N^{-1/2} \bigr) .
\end{align*}%
Substituting this result together with (\ref{RNT1}) and (\ref{SNT2}) in
(\ref{ehattilda}), we obtain%
%
\begin{equation}
N^{-1}\sum_{i=1}^{N}
\hat{e}_{i,T+1}\tilde{e}_{i,T+1}=N^{-1}\sum
_{i=1}^{N}%
\varepsilon
_{i,T+1}^{2}+T^{-1}\psi _{NT}+O_{p}
\bigl(N^{-1/2}\bigr)+O_{p} \bigl( T^{-1/2}N^{-1/2}
\bigr) , \label{MSF_PI}
\end{equation}%
where
$\psi _{NT}=T  [ \mathrm{E}  ( g_{1,nT}  ) -\mathrm{E}
  ( g_{2,NT}  )   ] $, or more specifically,%
%
\begin{eqnarray}
\psi _{NT} &=&TN^{-1}\sum_{i=1}^{N}\text{E}
\bigl[ \boldsymbol{\varepsilon }%
_{i}^{\prime }
\boldsymbol{W}_{i}\bigl(\boldsymbol{W}_{i}^{\prime }
\boldsymbol{W}%
_{i}\bigr)^{-1}
\boldsymbol{w}_{i,T+1}\boldsymbol{w}_{i,T+1}^{\prime }
\bigr] \boldsymbol{\bar{Q}}_{N}^{-1}\boldsymbol{
\bar{q}}_{N}\notag
\\
&&{}-TN^{-1}\sum_{i=1}^{N}
\mathrm{E} \bigl[ \boldsymbol{\varepsilon }%
_{i}^{\prime }
\boldsymbol{W}_{i}\bigl(\boldsymbol{W}_{i}^{\prime }
\boldsymbol{W}%
_{i}\bigr)^{-1}
\boldsymbol{w}_{i,T+1}\boldsymbol{w}_{i,T+1}^{\prime }
\boldsymbol{%
\eta }_{i} \bigr] . \label{epsiNT}
\end{eqnarray}
Finally, using (\ref{MSF_PI}), together with (\ref{MSFEI}) and (\ref{MSFEP1}%
), in (\ref{wstar}) now yields (\ref{w*NT}).

\subsection{Proof of Proposition~\protect\ref{prop:combinedFE}}\label%
{ProofCombinedFE}

To compare the FE forecast to the individual forecasts, rewrite (\ref{ehat})
as
$\hat{e}_{i,T+1}=\varepsilon _{i,T+1}-  ( \hat{\alpha}_{i}-
\alpha _{i}  ) -\boldsymbol{x}_{i,T+1}^{\prime }(
\boldsymbol{\hat{\beta}}_{i}-%
\boldsymbol{\beta }_{i})$, and note that
$\hat{\alpha}_{i}-\alpha _{i}=\bar{%
\varepsilon}_{iT}-\boldsymbol{\bar{x}}_{iT}^{\prime }  (
\boldsymbol{\hat{%
\beta}}_{i}-\boldsymbol{\beta }_{i}  ) $. Therefore,
$ \hat{e}_{i,T+1}=\bar{\bar{\varepsilon}}_{i,T+1}-
\bar{\bar{\boldsymbol{x}}}%
_{i,T+1}^{\prime }(\boldsymbol{\hat{\beta}}_{i}-\boldsymbol{\beta }_{i})
$. Furthermore,
$\hat{e}_{i,T+1}^{%
\text{FE}}=\bar{\bar{\varepsilon}}_{i,T+1}-(\hat{\boldsymbol{\beta }}_{
\text{%
FE}}-\boldsymbol{\beta }_{i})^{\prime }\bar{\bar{\boldsymbol{x}}}_{i,T+1}$,
where
$\bar{\bar{\varepsilon}}_{i,T+1}=\varepsilon _{i,T+1}-
\bar{\varepsilon}%
_{iT}$ and
$\bar{\bar{\boldsymbol{x}}}_{i,T+1}=\boldsymbol{x}_{iT+1}-
\bar{%
\boldsymbol{x}}_{iT}$,
\begin{eqnarray*}
\hat{\boldsymbol{\beta }}_{i}-\boldsymbol{\beta }_{i}
&=&\bigl( \boldsymbol{X}%
_{i}^{\prime }
\boldsymbol{M}_{T}\boldsymbol{X}_{i}
\bigr)^{-1} \boldsymbol{X}%
_{i}^{\prime }
\boldsymbol{M}_{T}\boldsymbol{\varepsilon }_{i}=
\boldsymbol{Q}%
_{iT,\beta }^{-1}\boldsymbol{\xi
}_{iT,\beta },
\\
\hat{\boldsymbol{\beta }}_{\text{FE}}-\boldsymbol{\beta }_{i}
&=&-%
\boldsymbol{\eta }_{i,\beta }+\bar{\boldsymbol{Q}}_{NT,\beta }^{-1}
\bar{%
\boldsymbol{q}}_{NT,\beta }+\bar{\boldsymbol{Q}}_{NT,\beta }^{-1}
\bar{%
\boldsymbol{\xi }}_{NT,\beta },
\\
\boldsymbol{Q}_{iT,\beta }&=&T^{-1}\boldsymbol{X}_{i}^{\prime }
\boldsymbol{M%
}_{T}\boldsymbol{X}_{i}, \quad \bar{\boldsymbol{Q}}_{NT} = N^{-1} \sum_{i=1}^N \boldsymbol{Q}_{iT,\beta},
\\
\boldsymbol{\xi }_{iT,\beta }&=&T^{-1}%
\boldsymbol{X}_{i}^{\prime }\boldsymbol{M}_{T}
\boldsymbol{\varepsilon }_{i}%
, \quad \text{and}\quad \bar{
\boldsymbol{\xi }}_{NT,\beta }=N^{-1}\sum
_{i=1}^{N}%
\boldsymbol{\xi
}_{iT,\beta }=O_{p}\bigl(N^{-1/2}T^{-1/2}\bigr).
\end{eqnarray*}%
Hence,%
%
\begin{eqnarray}
N^{-1}\sum_{i=1}^{N}
\hat{e}_{i,T+1}^{\text{FE}}\hat{e}_{i,T+1}
&=&N^{-1} \sum_{i=1}^{N}
\bar{\bar{\varepsilon}}_{i,T+1}^{2}\notag
\\
&&{}+N^{-1}\sum_{i=1}^{N}(
\hat{\boldsymbol{\beta }}_{i}- \boldsymbol{\beta }%
_{i})^{\prime }
\bar{\bar{\boldsymbol{x}}}_{i,T+1} \bar{\bar{\boldsymbol{x}}}%
_{i,T+1}^{\prime }(
\hat{\boldsymbol{\beta }}_{\text{FE}}- \boldsymbol{\beta }%
_{i})
\notag
\\
&&{}-N^{-1}\sum_{i=1}^{N}(
\varepsilon _{i,T+1}-\bar{\varepsilon}_{iT})
\bar{%
\bar{\boldsymbol{x}}}_{i,T+1}^{\prime } \bigl[ (
\hat{\boldsymbol{\beta }%
}_{\text{FE}}-\boldsymbol{\beta
}_{i} ) +( \hat{\boldsymbol{\beta }}_{i}-%
\boldsymbol{\beta }_{i}) \bigr] . \label{FEMSFE}
\end{eqnarray}%
Under Assumptions \ref{ass:weak_exogeneity_b} and \ref{ass:6}, we have
\begin{equation*}
\frac{1}{N}\sum_{i=1}^{N} ( \hat{
\boldsymbol{\beta }}_{
\text{FE}}-%
\boldsymbol{\beta
}_{i} ) ^{\prime }\bar{\bar{\boldsymbol{x}}}_{i,T+1}%
\bar{\varepsilon}_{iT}=c_{NT}^{\text{FE}}+O_{p}
\bigl(N^{-1/2}\bigr).
\end{equation*}%
Additionally,
\begin{equation*}
N^{-1}\sum_{i=1}^{N}\bar{
\bar{\boldsymbol{x}}}_{i,T+1}^{\prime }( \boldsymbol{%
\hat{\beta}}_{i}-\boldsymbol{\beta }_{i})\bar{\bar{
\varepsilon}}%
_{i,T+1}=c_{NT,\beta }+O_{p}
\bigl(N^{-1/2}\bigr),
\end{equation*}%
where
$ c_{NT,\beta }=N^{-1}\sum_{i=1}^{N}\mathrm{E}  [
\bar{\bar{\boldsymbol{x}}}%
_{i,T+1}^{\prime }(\boldsymbol{X}_{i}^{\prime }\boldsymbol{M}_{T}
\boldsymbol{%
X}_{i})^{-1}\boldsymbol{X}_{i}^{\prime }\boldsymbol{M}_{T}
\boldsymbol{%
\varepsilon }_{i}\bar{\varepsilon}_{iT+1}  ]  $. Details are in the
Supplemental Appendix, where we also show that (see (S.22))
\begin{equation}
N^{-1}\sum_{i=1}^{N}
\hat{e}_{i,T+1}^{2}=N^{-1}\sum
_{i=1}^{N} \bar{\bar{%
\varepsilon}}_{i,T+1}^{2}+T^{-1}h_{NT,\beta }-2c_{NT,\beta }+O_{p}
\bigl(N^{-1/2}\bigr), \label{MSFEI_FE}
\end{equation}%
where
$ h_{NT,\beta }=N^{-1}\sum_{i=1}^{N}\mathrm{E}  \bigl[
\bar{\bar{\boldsymbol{x}}}%
_{i,T+1}^{\prime }\boldsymbol{Q}_{iT,\beta }^{-1}  \bigl(
\frac{\boldsymbol{X}%
_{i}^{\prime }\boldsymbol{M}_{T}\boldsymbol{\varepsilon }_{i}\boldsymbol{%
\varepsilon }_{i}^{\prime }\boldsymbol{M}_{T}\boldsymbol{X}_{i}}{T}
  \bigr) \boldsymbol{Q}_{iT,\beta }^{-1}\bar{\bar{\boldsymbol{x}}}_{i,T+1}
  \bigr]  $.

Using this result, we have%
%
\begin{equation}
N^{-1}\sum_{i=1}^{N}\bar{
\bar{\varepsilon}}_{i,T+1} \bar{\bar{\boldsymbol{x}}}%
_{i,T+1}^{\prime }
\bigl[ ( \hat{\boldsymbol{\beta }}_{
\text{FE}}-%
\boldsymbol{\beta
}_{i} ) +(\hat{\boldsymbol{\beta }}_{i}-
\boldsymbol{%
\beta }_{i}) \bigr] =c_{NT}^{\text{FE}}+c_{NT,\beta }+O_{p}
\bigl(N^{-1/2}\bigr). \label{FEMSFE1}
\end{equation}%
Also,
\begin{eqnarray*}
&&N^{-1}\sum_{i=1}^{N}(
\hat{\boldsymbol{\beta }}_{i}- \boldsymbol{\beta }%
_{i})^{\prime }
\bigl( \bar{\bar{\boldsymbol{x}}}_{i,T+1} \bar{\bar{%
\boldsymbol{x}}}_{i,T+1}^{\prime } \bigr) (\hat{\boldsymbol{\beta
}}_{
\text{%
FE}}-\boldsymbol{\beta }_{i})
\\
&&\quad =T^{-1/2}N^{-1}\sum_{i=1}^{N}
\bigl( T^{-1/2} \boldsymbol{\varepsilon }%
_{i}^{\prime }
\boldsymbol{M}_{T}\boldsymbol{X}_{i} \bigr)
\\
&&{}\qquad \times
\boldsymbol{Q}%
_{iT,\beta }^{-1} \bigl( \bar{\bar{
\boldsymbol{x}}}_{i,T+1} \bar{\bar{%
\boldsymbol{x}}}_{i,T+1}^{\prime }
\bigr) \bigl( - \boldsymbol{\eta }%
_{i,\beta }+\bar{
\boldsymbol{Q}}_{NT,\beta }^{-1} \bar{\boldsymbol{q}}%
_{NT,\beta }+
\bar{\boldsymbol{Q}}_{NT,\beta }^{-1} \bar{\boldsymbol{\xi
}}%
_{NT,\beta } \bigr) .
\end{eqnarray*}%
$\bar{\boldsymbol{\xi }}_{NT,\beta }=O_{p}(N^{-1/2}T^{-1/2})$ and
$\bar{%
\boldsymbol{Q}}_{NT,\beta }^{-1}=\bar{\boldsymbol{Q}}_{N,\beta }^{-1}+O_{p}(N^{-1/2})$,
where
$\bar{\boldsymbol{Q}}_{N,\beta }=E  ( \bar{%
\boldsymbol{Q}}_{NT,\beta }  ) $, (see Lemma~\ref{Lemma_1_VTEX1}). Hence,
for a fixed $T>T_{0}$,
\begin{equation*}
\Biggl[ N^{-1}\sum_{i=1}^{N}
\bigl( T^{-1/2} \boldsymbol{\varepsilon }%
_{i}^{\prime }
\boldsymbol{M}_{T}\boldsymbol{X}_{i} \bigr)
\boldsymbol{Q}%
_{iT,\beta }^{-1} \bigl( \bar{\bar{
\boldsymbol{x}}}_{i,T+1} \bar{\bar{%
\boldsymbol{x}}}_{i,T+1}^{\prime }
\bigr) \Biggr] \bar{\boldsymbol{Q}}%
_{NT,\beta }^{-1}
\bar{\boldsymbol{\xi }}_{NT,\beta }=O_{p}\bigl(N^{-1/2}T^{-1/2}\bigr).
\end{equation*}%
Also, under Assumption~\ref{ass:6},
\begin{eqnarray*}
&&N^{-1}\sum_{i=1}^{N}
\bigl( T^{-1/2}\boldsymbol{\varepsilon }_{i}^{
\prime }%
\boldsymbol{M}_{T}\boldsymbol{X}_{i} \bigr)
\boldsymbol{Q}_{iT,
\beta }^{-1} \bigl( \bar{\bar{
\boldsymbol{x}}}_{i,T+1} \bar{\bar{\boldsymbol{x}}}%
_{i,T+1}^{\prime }
\bigr) \boldsymbol{\eta }_{i,\beta }
\\
&&\quad =N^{-1}\sum_{i=1}^{N}
\mathrm{E} \bigl[ \bigl( T^{-1/2} \boldsymbol{%
\varepsilon
}_{i}^{\prime }\boldsymbol{M}_{T}
\boldsymbol{X}_{i} \bigr) \boldsymbol{Q}_{iT,\beta }^{-1}
\bigl( \bar{\bar{\boldsymbol{x}}}_{i,T+1}\bar{%
\bar{
\boldsymbol{x}}}_{i,T+1}^{\prime } \bigr) \boldsymbol{\eta
}_{i,
\beta }%
 \bigr] +O_{p}\bigl(N^{-1/2}
\bigr),
\end{eqnarray*}%
and
\begin{eqnarray*}
&&N^{-1}\sum_{i=1}^{N}
\bigl( T^{-1/2}\boldsymbol{\varepsilon }_{i}^{
\prime }%
\boldsymbol{M}_{T}\boldsymbol{X}_{i} \bigr)
\boldsymbol{Q}_{iT,
\beta }^{-1} \bigl( \bar{\bar{
\boldsymbol{x}}}_{i,T+1} \bar{\bar{\boldsymbol{x}}}%
_{i,T+1}^{\prime }
\bigr)
\\
&&\quad =N^{-1}\sum_{i=1}^{N}
\mathrm{E} \bigl[ \bigl( T^{-1/2} \boldsymbol{%
\varepsilon
}_{i}^{\prime }\boldsymbol{M}_{T}
\boldsymbol{X}_{i} \bigr) \boldsymbol{Q}_{iT,\beta }^{-1}
\bigl( \bar{\bar{\boldsymbol{x}}}_{i,T+1}\bar{%
\bar{
\boldsymbol{x}}}_{i,T+1}^{\prime } \bigr) \bigr] +O_{p}
\bigl(N^{-1/2}\bigr).
\end{eqnarray*}%
Then%
%
\begin{equation}
N^{-1}\sum_{i=1}^{N}(\hat{
\boldsymbol{\beta }}_{i}- \boldsymbol{\beta }%
_{i})^{\prime }
\bar{\bar{\boldsymbol{x}}}_{i,T+1} \bar{\bar{\boldsymbol{x}}}%
_{i,T+1}^{\prime }(
\hat{\boldsymbol{\beta }}_{\text{FE}}- \boldsymbol{\beta }%
_{i})=T^{-1}
\psi _{NT}^{\text{FE}}+O_{p}\bigl(T^{-1/2}N^{-1/2}
\bigr), \label{FEMSFE2}
\end{equation}%
where
%
\begin{eqnarray}
\psi _{NT}^{\text{FE}} &=&TN^{-1}\sum
_{i=1}^{N}\mathrm{E} \bigl[ \bigl(
T^{-1}\boldsymbol{\varepsilon }_{i}^{\prime }
\boldsymbol{M}_{T} \boldsymbol{%
X}_{i} \bigr)
\boldsymbol{Q}_{iT,\beta }^{-1} \bigl( \bar{\bar{
\boldsymbol{x}}%
}_{i,T+1}\bar{\bar{\boldsymbol{x}}}_{i,T+1}^{\prime }
\bigr) \bigr] \bar{%
\boldsymbol{Q}}_{N,\beta }^{-1}
\bar{\boldsymbol{q}}_{N,\beta } \notag
\\
&&{}-TN^{-1}\sum_{i=1}^{N}
\mathrm{E} \bigl[ \bigl( T^{-1} \boldsymbol{%
\varepsilon
}_{i}^{\prime }\boldsymbol{M}_{T}
\boldsymbol{X}_{i} \bigr) \boldsymbol{Q}_{iT,\beta }^{-1}
\bigl( \bar{\bar{\boldsymbol{x}}}_{i,T+1}\bar{%
\bar{
\boldsymbol{x}}}_{i,T+1}^{\prime } \bigr) \boldsymbol{\eta
}_{i,
\beta }%
 \bigr] . \label{FEepsiAPP}
\end{eqnarray}%
Using (\ref{FEMSFE1}) and (\ref{FEMSFE2}) in (\ref{FEMSFE}) yields
\begin{eqnarray*}
N^{-1}\sum_{i=1}^{N}
\hat{e}_{i,T+1}^{\text{FE}}\hat{e}_{i,T+1}
&=&N^{-1} \sum_{i=1}^{N}
\bar{\bar{\varepsilon}}_{i,T+1}^{2} +T^{-1}\psi
_{NT}^{
\text{FE}}
\\
&&{}- \bigl( c_{NT}^{\text{FE}}+c_{NT,\beta } \bigr)
+O_{p}\bigl(N^{-1/2}\bigr)+O_{p}
\bigl(T^{-1/2}N^{-1/2}\bigr).
\end{eqnarray*}%
Substituting this result together with (\ref{FEmsfe}) and (\ref{MSFEI_FE}) in (\ref{wNT-FE}) now yields equation (\ref{wFE}).
This expression corresponds to (\ref{MSF_PI}) for comparison with pooled estimates. 

\section{Estimation of combination weights}\label{app:estimationOfWeights}

There are three components in the forecast combination weights, given by
(%
\ref{w*NT}), namely $\Delta _{NT}$, $h_{NT,}$ and ${\psi }_{NT}$. To establish
that
$\hat{\Delta}_{NT}  ( \tilde{\boldsymbol{\eta }}  ) =N^{-1}
\sum_{i=1}^{N}\boldsymbol{w}_{i,T+1}^{\prime }
\tilde{\boldsymbol{\eta }}_{i}\tilde{\boldsymbol{\eta }}_{i}^{\prime }
\boldsymbol{w}_{i,T+1}$ is a consistent estimator of $\Delta _{NT}$, recall
that
\begin{equation*}
\tilde{\boldsymbol{\eta }}_{i}=\boldsymbol{\eta }_{i}+
\boldsymbol{Q}%
_{iT}^{-1}\boldsymbol{\xi
}_{iT}-\boldsymbol{\bar{Q}}_{NT}^{-1}
\boldsymbol{%
\bar{q}}_{NT}-\boldsymbol{
\bar{Q}}_{NT}^{-1}\boldsymbol{\bar{\xi}}_{NT},
\end{equation*}%
where
\begin{equation*}
\begin{split}\boldsymbol{\xi }_{iT }&=T^{-1}\boldsymbol{W}_{i}^{\prime }
\boldsymbol{%
\varepsilon }_{i}=T^{-1}\sum
_{t=1}^{T}\boldsymbol{w}_{it}
\varepsilon _{it}=O_{p} \bigl( T^{-1/2} \bigr)
,  \quad \text{and }
\\
  \bar{%
\boldsymbol{\xi
}}_{NT}&=N^{-1}\sum_{i=1}^{N}
\boldsymbol{\xi }%
_{iT}=O_{p} \bigl(
N^{-1/2}T^{-1/2} \bigr) .
\end{split}\end{equation*}%
Then
\begin{eqnarray*}
\hat{\Delta}_{NT} ( \tilde{\boldsymbol{\eta }} )
&=&N^{-1} \sum_{i=1}^{N}
\boldsymbol{w}_{i,T+1}^{\prime } \tilde{\boldsymbol{%
\eta }}_{i}\tilde{\boldsymbol{\eta }}_{i}^{\prime }
\boldsymbol{w}_{i,T+1}
\\
&=&N^{-1}\sum_{i=1}^{N}
\boldsymbol{w}_{i,T+1}^{\prime } \bigl( \boldsymbol{%
\eta }_{i}+\boldsymbol{Q}_{iT}^{-1}
\boldsymbol{\xi }_{iT}- \boldsymbol{\bar{Q}%
}_{NT}^{-1}
\boldsymbol{\bar{q}}_{NT}-\boldsymbol{\bar{Q}}_{NT}^{-1}%
\boldsymbol{\bar{\xi}}_{NT} \bigr)
\\
&&{}\times \bigl( \boldsymbol{\eta }_{i}+\boldsymbol{Q}_{iT}^{-1}
\boldsymbol{%
\xi }_{iT}-\boldsymbol{\bar{Q}}_{NT}^{-1}
\boldsymbol{\bar{q}}_{NT}-%
\boldsymbol{
\bar{Q}}_{NT}^{-1}\boldsymbol{\bar{\xi}}_{NT}
\bigr) ^{
\prime }%
\boldsymbol{w}_{i,T+1}
\\
&=&N^{-1}\sum_{i=1}^{N}
\boldsymbol{w}_{i,T+1}^{\prime } \boldsymbol{\eta
}_{i}%
\boldsymbol{\eta }_{i}^{\prime }
\boldsymbol{w}_{i,T+1}+N^{-1}\sum
_{i=1}^{N}%
\boldsymbol{w}_{i,T+1}^{\prime }
\boldsymbol{Q}_{iT}^{-1} \boldsymbol{\xi
}%
_{iT}\boldsymbol{\xi }_{iT}^{\prime }
\boldsymbol{Q}_{iT}^{-1} \boldsymbol{w}%
_{i,T+1}
\\
&&{}+N^{-1}\sum_{i=1}^{N}
\boldsymbol{w}_{i,T+1}^{\prime } \boldsymbol{\bar{Q}}%
_{NT}^{-1}
\boldsymbol{\bar{q}}_{NT}\boldsymbol{\bar{q}}_{NT}
\boldsymbol{\bar{%
Q}}_{NT}^{-1}
\boldsymbol{w}_{i,T+1}
\\
&&{}+N^{-1}\sum
_{i=1}^{N} \boldsymbol{w}%
_{i,T+1}^{\prime }
\boldsymbol{\bar{Q}}_{NT}^{-1} \boldsymbol{\bar{
\xi}}_{NT}%
\boldsymbol{\bar{\xi}}_{NT}^{\prime }
\boldsymbol{\bar{Q}}_{NT}^{-1}%
\boldsymbol{w}_{i,T+1}
\\
&&{}+2N^{-1}\sum_{i=1}^{N}
\boldsymbol{w}_{i,T+1}^{\prime } \boldsymbol{\eta
}%
_{i}\boldsymbol{\xi }_{iT}^{\prime }
\boldsymbol{Q}_{iT}^{-1} \boldsymbol{w}%
_{i,T+1}-2N^{-1}
\sum_{i=1}^{N}\boldsymbol{w}_{i,T+1}^{\prime }
\boldsymbol{%
\eta }_{i}\boldsymbol{\bar{q}}_{NT}^{\prime }
\boldsymbol{\bar{Q}}_{NT}^{-1}%
\boldsymbol{w}_{i,T+1}
\\
&&{}-2N^{-1}\sum_{i=1}^{N}
\boldsymbol{w}_{i,T+1}^{\prime } \boldsymbol{\eta
}%
_{i}\boldsymbol{\bar{\xi}}_{NT}^{\prime }
\boldsymbol{\bar{Q}}_{NT}^{-1}%
\boldsymbol{w}_{i,T+1}
\\
&&{}-2N^{-1}\sum
_{i=1}^{N}\boldsymbol{w}_{i,T+1}^{
\prime }%
\boldsymbol{Q}_{iT}^{-1}\boldsymbol{\xi }_{iT}
\boldsymbol{\bar{q}}%
_{NT}^{\prime }\boldsymbol{
\bar{Q}}_{NT}^{-1}\boldsymbol{w}_{i,T+1}
\\
&&{}-N^{-1}\sum_{i=1}^{N}
\boldsymbol{w}_{i,T+1}^{\prime } \boldsymbol{Q}%
_{iT}^{-1}
\boldsymbol{\xi }_{iT}\boldsymbol{\bar{\xi}}_{NT}^{\prime }%
\boldsymbol{\bar{Q}}_{NT}^{-1}\boldsymbol{w}_{i,T+1}
\\
&&{}+2N^{-1}
\sum_{i=1}^{N}%
\boldsymbol{w}_{i,T+1}^{\prime }\boldsymbol{\bar{Q}}_{NT}^{-1}
\boldsymbol{%
\bar{q}}_{NT}\boldsymbol{\bar{
\xi}}_{NT}^{\prime } \boldsymbol{\bar{Q}}%
_{NT}^{-1}
\boldsymbol{w}_{i,T+1},
\end{eqnarray*}%
and we have that
\begin{eqnarray*}
N^{-1}\sum_{i=1}^{N}
\boldsymbol{w}_{i,T+1}^{\prime }\boldsymbol{Q}_{iT}^{-1}%
\boldsymbol{\xi }_{iT}\boldsymbol{\xi }_{iT}^{\prime }
\boldsymbol{Q}%
_{iT}^{-1}\boldsymbol{w}_{i,T+1}
&=&N^{-1}\sum_{i=1}^{N}
\boldsymbol{w}%
_{i,T+1}^{\prime }\mathrm{E} \bigl(
\boldsymbol{Q}_{iT}^{-1} \boldsymbol{\xi
}%
_{iT}\boldsymbol{\xi }_{iT}^{\prime }
\boldsymbol{Q}_{iT}^{-1} \bigr) \boldsymbol{w}_{i,T+1}
\\
&&{}+O_{p}
\bigl( N^{-1/2} \bigr),
\\
\mathrm{E} \bigl( \boldsymbol{Q}_{iT}^{-1}\boldsymbol{
\xi }_{iT} \boldsymbol{%
\xi }_{iT}^{\prime }
\boldsymbol{Q}_{iT}^{-1} \bigr) &=&T^{-2}
\mathrm{E}%
 \bigl( \boldsymbol{Q}_{iT}^{-1}
\boldsymbol{W}_{i}^{\prime } \boldsymbol{%
\varepsilon }_{i}\boldsymbol{\varepsilon }_{i}^{\prime }
\boldsymbol{W}%
_{i}^{\prime }\boldsymbol{Q}_{iT}^{-1}
\bigr) =T^{-1}\sigma _{i}^{2}%
\boldsymbol{Q}_{i}^{-1},
\\
N^{-1}\sum_{i=1}^{N}
\boldsymbol{w}_{i,T+1}^{\prime } \boldsymbol{\xi
}_{iT}%
\boldsymbol{\xi }_{iT}^{\prime }
\boldsymbol{w}_{i,T+1} &=&N^{-1}\sum
_{i=1}^{N} \boldsymbol{w}_{i,T+1}^{\prime }
\mathrm{E} \bigl( \boldsymbol{Q}_{iT}^{-1} \boldsymbol{\xi
}_{iT}\boldsymbol{\xi }_{iT}^{\prime }
\boldsymbol{Q}_{iT}^{-1} \bigr) \boldsymbol{w}_{i,T+1}
\\
&&{}+O_{p}
\bigl( N^{-1/2} \bigr)
\\
&=&O_{p} \bigl( N^{-1/2} \bigr) +O_{p}
\bigl(T^{-1}\bigr),
\\
N^{-1}\sum_{i=1}^{N}
\boldsymbol{w}_{i,T+1}^{\prime } \boldsymbol{\bar{Q}}%
_{NT}^{-1}
\boldsymbol{\bar{q}}_{NT}\boldsymbol{\bar{q}}_{NT}
\boldsymbol{\bar{%
Q}}_{NT}^{-1}
\boldsymbol{w}_{i,T+1} &=&\boldsymbol{\bar{q}}_{NT}
\boldsymbol{%
\bar{Q}}_{NT}^{-1} \Biggl(
N^{-1}\sum_{i=1}^{N}
\boldsymbol{w}_{i,T+1}%
\boldsymbol{w}_{i,T+1}^{\prime }
\Biggr) \boldsymbol{\bar{Q}}_{NT}^{-1}%
\boldsymbol{\bar{q}}_{NT}
\\
&=&\boldsymbol{\bar{q}}_{N}\boldsymbol{\bar{Q}}_{N}^{-1}
\boldsymbol{\bar{q}}%
_{N}+O_{p}
\bigl(N^{-1/2}\bigr),
\\
N^{-1}\sum_{i=1}^{N}
\boldsymbol{w}_{i,T+1}^{\prime } \boldsymbol{\bar{Q}}%
_{NT}^{-1}
\boldsymbol{\bar{\xi}}_{NT}\boldsymbol{\bar{\xi}}_{NT}^{
\prime }%
\boldsymbol{\bar{Q}}_{NT}^{-1}\boldsymbol{w}_{i,T+1}&=&
\boldsymbol{\bar{\xi}}%
_{NT}^{\prime }\boldsymbol{
\bar{Q}}_{NT}^{-1} \Biggl( N^{-1}\sum
_{i=1}^{N}%
\boldsymbol{w}_{i,T+1}
\boldsymbol{w}_{i,T+1}^{\prime } \Biggr) \boldsymbol{%
\bar{Q}}_{NT}^{-1}\boldsymbol{\bar{\xi}}_{NT}
\\
&=&O_{p}
\bigl(N^{-1}T^{-1}\bigr),
\end{eqnarray*}
and
\begin{eqnarray*}
N^{-1}\sum_{i=1}^{N}
\boldsymbol{w}_{i,T+1}^{\prime } \boldsymbol{\eta
}_{i}%
\boldsymbol{\bar{q}}_{NT}^{\prime }
\boldsymbol{\bar{Q}}_{NT}^{-1} \boldsymbol{%
w}_{i,T+1}&=&
\boldsymbol{\bar{q}}_{NT}^{\prime }\boldsymbol{
\bar{Q}}%
_{NT}^{-1} \Biggl( N^{-1}
\sum_{i=1}^{N}\boldsymbol{w}_{i,T+1}
\boldsymbol{w}%
_{i,T+1}^{\prime }\boldsymbol{\eta
}_{i} \Biggr)
\\
 &=& \boldsymbol{\bar{q}}_{N}%
\boldsymbol{\bar{Q}}_{N}^{-1}\boldsymbol{
\bar{q}}_{N}+O_{p}\bigl(N^{-1/2}\bigr).
\end{eqnarray*}%
Also, since E$  ( \boldsymbol{Q}_{iT}^{-1}\boldsymbol{\xi }_{iT}
 \llvert  \boldsymbol{w}_{i,T+1},\boldsymbol{\eta }_{i}
  ) =0$,
\begin{eqnarray*}
N^{-1}\sum_{i=1}^{N}
\boldsymbol{w}_{i,T+1}^{\prime } \boldsymbol{\eta
}_{i}%
\boldsymbol{\xi }_{iT}^{\prime }
\boldsymbol{Q}_{iT}^{-1} \boldsymbol{w}%
_{i,T+1}&=&N^{-1}
\sum_{i=1}^{N}\boldsymbol{\xi
}_{iT}^{\prime } \boldsymbol{Q}%
_{iT}^{-1}
\bigl( \boldsymbol{w}_{i,T+1}\boldsymbol{w}_{i,T+1}^{
\prime }%
\boldsymbol{\eta }_{i} \bigr) =O_{p} \bigl(
N^{-1/2} \bigr) ,%
\\
N^{-1}\sum_{i=1}^{N}
\boldsymbol{w}_{i,T+1}^{\prime } \boldsymbol{\eta
}_{i}%
\bar{\boldsymbol{\xi }}_{NT}^{\prime }
\bar{\boldsymbol{Q}}_{NT}^{-1}%
\boldsymbol{w}_{i,T+1}&=&N^{-3}\sum
_{i=1}^{N}\boldsymbol{\xi }_{iT}^{
\prime }%
\boldsymbol{Q}_{iT}^{-1} \bigl( \boldsymbol{w}_{i,T+1}
\boldsymbol{w}%
_{i,T+1}^{\prime }\boldsymbol{\eta
}_{i} \bigr) +O_{p} \bigl( N^{-3} \bigr)
\\
&=&O_{p} \bigl( N^{-5/2} \bigr) ,%
\\
N^{-1}\sum_{i=1}^{N}
\boldsymbol{w}_{i,T+1}^{\prime }\boldsymbol{Q}_{iT}^{-1}%
\boldsymbol{\xi }_{iT}\boldsymbol{\bar{q}}_{NT}^{\prime }
\boldsymbol{\bar{Q}}%
_{NT}^{-1}
\boldsymbol{w}_{i,T+1}&=&\boldsymbol{\bar{q}}_{NT}^{\prime }%
\boldsymbol{\bar{Q}}_{NT}^{-1} \Biggl( N^{-1}
\sum_{i=1}^{N} \bigl(
\boldsymbol{w%
}_{i,T+1}\boldsymbol{w}_{i,T+1}^{\prime }
\bigr) \boldsymbol{Q}_{iT}^{-1}%
\boldsymbol{
\xi }_{iT} \Biggr)
\\
 &=&O_{p} \bigl( N^{-1/2} \bigr)
,%
\\
N^{-1}\sum_{i=1}^{N}
\boldsymbol{w}_{i,T+1}^{\prime }\boldsymbol{Q}_{iT}^{-1}%
\boldsymbol{\xi }_{iT}\boldsymbol{\bar{\xi}}_{NT}^{\prime }
\boldsymbol{\bar{Q%
}}_{NT}^{-1}
\boldsymbol{w}_{i,T+1}&=&\boldsymbol{\bar{\xi}}_{NT}^{
\prime }%
\boldsymbol{\bar{Q}}_{NT}^{-1} \Biggl( N^{-1}
\sum_{i=1}^{N} \boldsymbol{w}%
_{i,T+1}
\boldsymbol{w}_{i,T+1}^{\prime }\boldsymbol{Q}_{iT}^{-1}
\boldsymbol{%
\xi }_{iT} \Biggr)
\\
 &=&O_{p}
\bigl(N^{-1}T^{-1/2}\bigr),%
\end{eqnarray*}
and
\begin{equation*}
\begin{split}N^{-1}\sum_{i=1}^{N}
\boldsymbol{w}_{i,T+1}^{\prime } \boldsymbol{\bar{Q}}%
_{NT}^{-1}
\boldsymbol{\bar{q}}_{NT}\boldsymbol{\bar{\xi}}_{NT}^{
\prime }%
\boldsymbol{\bar{Q}}_{NT}^{-1}\boldsymbol{w}_{i,T+1}&=
\boldsymbol{\bar{\xi}}%
_{NT}^{\prime }\boldsymbol{
\bar{Q}}_{NT}^{-1} \Biggl( N^{-1}\sum
_{i=1}^{N}%
\boldsymbol{w}_{i,T+1}
\boldsymbol{w}_{i,T+1}^{\prime } \Biggr) \boldsymbol{%
\bar{Q}}_{NT}^{-1}\boldsymbol{\bar{q}}_{NT}
\\
&=O_{p}
\bigl(T^{-1/2}N^{-1/2}\bigr).
\end{split}\end{equation*}%
Overall
\begin{equation*}
\hat{\Delta}_{NT} ( \tilde{\boldsymbol{\eta }} ) =N^{-1}
\sum_{i=1}^{N}\boldsymbol{w}_{i,T+1}^{\prime }
\boldsymbol{\eta }_{i}%
\boldsymbol{\eta
}_{i}^{\prime }\boldsymbol{w}_{i,T+1}- \boldsymbol{
\bar{q}}%
_{N}\boldsymbol{\bar{Q}}_{N}^{-1}
\boldsymbol{\bar{q}}_{N}+O_{p} \bigl( N^{-1/2}
\bigr) +O_{p}\bigl(T^{-1}\bigr),
\end{equation*}%
or equivalently
\begin{equation*}
=N^{-1}\sum_{i=1}^{N}
\mathrm{E} \bigl( \boldsymbol{w}_{i,T+1}^{
\prime }%
\boldsymbol{\eta }_{i}\boldsymbol{\eta }_{i}^{\prime }
\boldsymbol{w}%
_{i,T+1} \bigr) -\boldsymbol{
\bar{q}}_{N}\boldsymbol{\bar{Q}}_{N}^{-1}%
\boldsymbol{\bar{q}}_{N}+O_{p} \bigl( N^{-1/2}
\bigr) +O_{p}\bigl(T^{-1}\bigr).
\end{equation*}

In combination,
$\hat{\Delta}_{NT}  ( \tilde{\boldsymbol{\eta }}  ) $ is a consistent
estimator of $\Delta _{NT}$ for large $N$ and $T$. In the case of strictly
exogenous regressors,
$\hat{\Delta}_{NT}  ( \tilde{%
\boldsymbol{\eta }}  ) $ is a consistent estimator of
$\Delta _{NT}$ for a fixed $T$, so long as
$\mathrm E  ( \hat{\boldsymbol{\theta }}_{i}  ) $ exits, since
in that case
$\mathrm{E}  ( \hat{\boldsymbol{\theta }}_{j}-%
\boldsymbol{\theta }_{j}  ) =\mathbf{0}$ for all $j$.

Now consider the second component of the weights, namely $h_{NT}$. 
We will show that a consistent estimator of $h_{NT}$ is given by
\begin{equation*}
\hat{h}_{NT}=N^{-1}\sum_{i=1}^{N}
\boldsymbol{w}_{i,T+1}^{\prime } \boldsymbol{%
Q}_{iT}^{-1}
\boldsymbol{\hat{H}}_{iT}\boldsymbol{Q}_{iT}^{-1}
\boldsymbol{w}%
_{i,T+1}=h_{NT}+O_{p}
\bigl( N^{-1/2} \bigr) +O_{p} \biggl( \frac{\ln
(N)}{%
\sqrt{T}} \biggr) ,
\end{equation*}%
where
$\boldsymbol{\hat{H}}_{iT}=\hat{\sigma}_{i}^{2}T^{-1}\sum_{t=1}^{T}%
\boldsymbol{w}_{it}\boldsymbol{w}_{it}^{\prime } $, and
$\hat{\sigma}%
_{i}^{2}=\sum_{t=1}^{T}\hat{\varepsilon}_{it}^{2}/(T-K)$. Note that
\begin{equation*}
\hat{h}_{NT}-h_{NT}=N^{-1}\sum
_{i=1}^{N}\boldsymbol{w}_{i,T+1}^{
\prime }%
\boldsymbol{Q}_{iT}^{-1}\boldsymbol{\hat{H}}_{iT}
\boldsymbol{Q}_{iT}^{-1}%
\boldsymbol{w}_{i,T+1},-N^{-1}
\sum_{i=1}^{N}\mathrm{E} (
s_{iT
 } ) ,
\end{equation*}%
where $s_{iT }$ is defined by (\ref{s_iT}). Since
$N^{-1}\sum_{i=1}^{N}%
\mathrm{E}  ( s_{iT }  ) =N^{-1}\sum_{i=1}^{N}s_{iT }+O_{p}
  ( N^{-1/2}  ) $,
\begin{eqnarray*}
&&\hat{h}_{NT}-h_{NT}
\\
 &&\quad =N^{-1}\sum
_{i=1}^{N}\boldsymbol{w}_{i,T+1}^{
\prime }%
\boldsymbol{Q}_{iT}^{-1} \Biggl[ \hat{
\sigma}_{i}^{2}T^{-1}\sum
_{t=1}^{T}%
\boldsymbol{w}_{it}
\boldsymbol{w}_{it}^{\prime }- \frac{\boldsymbol{W}%
_{i}^{\prime }
\boldsymbol{\varepsilon }_{i}\boldsymbol{\varepsilon
}%
_{i}^{\prime }\boldsymbol{W}_{i}}{T}
\Biggr] \boldsymbol{Q}_{iT}^{-1}%
\boldsymbol{w}_{i,T+1}+O_{p} \bigl( N^{-1/2} \bigr)
\\
&&\quad =N^{-1}\sum_{i=1}^{N}
\hat{\sigma}_{i}^{2}\boldsymbol{w}_{i,T+1}^{
\prime }%
\boldsymbol{Q}_{iT}^{-1}\boldsymbol{w}_{i,T+1}
\\
&&\qquad {}-N^{-1}
\sum_{i=1}^{N}%
\boldsymbol{w}_{i,T+1}^{\prime }\boldsymbol{Q}_{iT}^{-1}
\biggl( \frac{%
\boldsymbol{W}_{i}^{\prime }
\boldsymbol{\varepsilon }_{i}\boldsymbol{%
\varepsilon
}_{i}^{\prime }\boldsymbol{W}_{i}}{T} \biggr)
\boldsymbol{Q}%
_{iT}^{-1}\boldsymbol{w}_{i,T+1}+O_{p}
\bigl( N^{-1/2} \bigr)
\\
&&\quad =D_{1,NT}-D_{2,NT}+O_{p} \bigl(
N^{-1/2} \bigr) .
\end{eqnarray*}%
Now consider the decomposition
\begin{eqnarray*}
D_{1,NT} &=&N^{-1}\sum_{i=1}^{N}
\hat{\sigma}_{i}^{2}\boldsymbol{w}%
_{i,T+1}^{\prime }
\boldsymbol{Q}_{iT}^{-1}\boldsymbol{w}_{i,T+1}
\\
&=&N^{-1}%
\sum_{i=1}^{N}\hat{\sigma}_{i}^{2}
\boldsymbol{w}_{i,T+1}^{\prime } \bigl( \boldsymbol{Q}_{iT}^{-1}-
\boldsymbol{Q}_{i}^{-1}+ \boldsymbol{Q}%
_{i}^{-1}
\bigr) \boldsymbol{w}_{i,T+1}
\\
&=&N^{-1}\sum_{i=1}^{N}
\hat{\sigma}_{i}^{2}\boldsymbol{w}_{i,T+1}^{
\prime }%
\boldsymbol{Q}_{i}^{-1}\boldsymbol{w}_{i,T+1}+N^{-1}
\sum_{i=1}^{N} \hat{
\sigma%
}_{i}^{2}\boldsymbol{w}_{i,T+1}^{\prime }
\bigl( \boldsymbol{Q}_{iT}^{-1}-%
\boldsymbol{Q}_{i}^{-1} \bigr) \boldsymbol{w}_{i,T+1}.
\end{eqnarray*}%
Also,
\begin{eqnarray*}
\Biggl\llVert N^{-1}\sum_{i=1}^{N}
\hat{\sigma}_{i}^{2}\boldsymbol{w}%
_{i,T+1}^{\prime }
\bigl( \boldsymbol{Q}_{iT}^{-1}-\boldsymbol{Q}%
_{i}^{-1}
\bigr) \boldsymbol{w}_{i,T+1} \Biggr\rrVert &\leq& \sup
_{i} \llVert \boldsymbol{w}_{i,T+1} \rrVert
^{2}\sup_{i} \hat{\sigma}%
_{i}^{2}
\sup_{i} \bigl\llVert \bigl( \boldsymbol{Q}_{iT}^{-1}-
\boldsymbol{Q}%
_{i}^{-1} \bigr) \bigr\rrVert
\\
&=&O_{p} \biggl( \frac{\ln (N)}{\sqrt{T}} \biggr) ,
\end{eqnarray*}%
and
$\sup_{i}  ( \hat{\sigma}_{i}^{2}-\sigma _{i}^{2}  ) =O_{p}
  ( \frac{\ln (N)}{\sqrt{T}}  ) $. Hence,
\begin{eqnarray*}
D_{1,NT} &=&N^{-1}\sum_{i=1}^{N}
\hat{\sigma}_{i}^{2}\boldsymbol{w}%
_{i,T+1}^{\prime }
\boldsymbol{Q}_{i}^{-1}\boldsymbol{w}_{i,T+1}+O_{p}
\biggl( \frac{\ln (N)}{\sqrt{T}} \biggr)
\\
&=&N^{-1}\sum_{i=1}^{N}
\sigma _{i}^{2}\boldsymbol{w}_{i,T+1}^{
\prime }%
\boldsymbol{Q}_{i}^{-1}\boldsymbol{w}_{i,T+1}+N^{-1}
\sum_{i=1}^{N} \bigl( \hat{
\sigma}_{i}^{2}-\sigma _{i}^{2}
\bigr) \boldsymbol{w}_{i,T+1}^{
\prime }%
\boldsymbol{Q}_{i}^{-1}\boldsymbol{w}_{i,T+1}
\\
&&{}+O_{p}
\biggl( \frac{\ln (N)}{%
\sqrt{T}} \biggr)
\\
&=&N^{-1}\sum_{i=1}^{N}
\sigma _{i}^{2}\boldsymbol{w}_{i,T+1}^{
\prime }%
\boldsymbol{Q}_{i}^{-1}\boldsymbol{w}_{i,T+1}+O_{p}
\biggl( \frac{\ln (N)}{%
\sqrt{T}} \biggr).
\end{eqnarray*}%
Similarly, and noting that
$\mathrm{E}  \biggl(
\frac{\boldsymbol{W}%
_{i}^{\prime }\boldsymbol{\varepsilon }_{i}\boldsymbol{\varepsilon }%
_{i}^{\prime }\boldsymbol{W}_{i}}{T}  \biggr) =\sigma _{i}^{2}
\boldsymbol{Q}%
_{i}$, we have
\begin{eqnarray*}
D_{2,NT} &=&N^{-1}\sum_{i=1}^{N}
\boldsymbol{w}_{i,T+1}^{\prime } \boldsymbol{Q%
}_{iT}^{-1}
\biggl( \frac{\boldsymbol{W}_{i}^{\prime }\boldsymbol{
\varepsilon }%
_{i}\boldsymbol{\varepsilon
}_{i}^{\prime }\boldsymbol{W}_{i}}{T} \biggr)
\boldsymbol{Q}_{iT}^{-1}\boldsymbol{w}_{i,T+1}
\\
&=&N^{-1}\sum_{i=1}^{N}
\boldsymbol{w}_{i,T+1}^{\prime } \boldsymbol{Q}%
_{i}^{-1}
\biggl( \frac{\boldsymbol{W}_{i}^{\prime }\boldsymbol{
\varepsilon }%
_{i}\boldsymbol{\varepsilon
}_{i}^{\prime }\boldsymbol{W}_{i}}{T} \biggr)
\boldsymbol{Q}_{i}^{-1}\boldsymbol{w}_{i,T+1}+O_{p}
\biggl( \frac{\ln (N)}{%
\sqrt{T}} \biggr)
\\
&=&N^{-1}\sum_{i=1}^{N}
\boldsymbol{w}_{i,T+1}^{\prime } \boldsymbol{Q}%
_{i}^{-1}
\mathrm{E} \biggl( \frac{\boldsymbol{W}_{i}^{\prime }
\boldsymbol{%
\varepsilon }_{i}\boldsymbol{\varepsilon
}_{i}^{\prime }\boldsymbol{W}_{i}}{T%
}
\biggr) \boldsymbol{Q}_{i}^{-1}\boldsymbol{w}_{i,T+1}+O_{p}
\bigl(N^{-1/2}\bigr)+O_{p}%
 \biggl( \frac{\ln
(N)}{\sqrt{T}} \biggr)
\\
&=&N^{-1}\sum_{i=1}^{N}
\sigma _{i}^{2}\boldsymbol{w}_{i,T+1}^{
\prime }%
\boldsymbol{Q}_{i}^{-1}\boldsymbol{w}_{i,T+1}+O_{p}
\bigl(N^{-1/2}\bigr)+O_{p} \biggl( \frac{\ln (N)}{\sqrt{T}}
\biggr) .
\end{eqnarray*}%
Hence,
\begin{equation*}
\hat{h}_{NT}-h_{NT}=O_{p}
\bigl(N^{-1/2}\bigr)+O_{p} \biggl( \frac{\ln (N)}{
\sqrt{T}}%
 \biggr) ,
\end{equation*}%
as desired.

Finally, turn to $\psi _{NT}$. The asymptotic bias,
$\hat{\boldsymbol{\theta }_{i}}-\boldsymbol{\theta }_{i}$, for each
$i$ can then be estimated using bootstrap or half-jackknifing. The sieve
bootstrap could be used for a pure panel AR model but generally not with
weakly exogenous regressors. However, the half-jackknife estimator can
work more generally. For a give $T$%
, split the sample in two equal parts, one observation is dropped if
$T$ is an odd number. Denote the estimators based on the two subsamples
by $\hat{%
\boldsymbol{\theta }}_{ia}$ and $\hat{\boldsymbol{\theta }}_{ib}$. Then
$%
\mathrm{E}  ( \hat{\boldsymbol{\theta}}_{i}-
\boldsymbol{\theta }%
_{i}  ) $ can be estimated by
\begin{equation*}
\hat{\boldsymbol{\theta }}_{i}- \biggl[ 2\hat{\boldsymbol{\theta
}}_{i}- \frac{1%
}{2} ( \hat{\boldsymbol{\theta
}}_{ia}+ \hat{\boldsymbol{\theta }}%
_{ib} )
\biggr] = \biggl[ \frac{1}{2} ( \hat{\boldsymbol{\theta }}%
_{ia}+
\hat{\boldsymbol{\theta }}_{ib} ) - \hat{\boldsymbol{\theta
}}_{i}%
 \biggr].
\end{equation*}%
A consistent estimator of $\psi _{NT}$ is then given by
\begin{eqnarray*}
{\hat{\psi}_{NT}} &=& \Biggl[ TN^{-1}\sum
_{i=1}^{N} \biggl[ \frac{1}{2} ( \hat{
\boldsymbol{\theta }}_{ia}+ \hat{\boldsymbol{\theta
}}_{ib} ) -\hat{%
\boldsymbol{\theta }}_{i}
\biggr] ^{\prime }\boldsymbol{w}_{i,T+1} \boldsymbol{%
w}_{i,T+1}^{\prime }
\Biggr] \boldsymbol{\bar{Q}}_{NT}^{-1} \boldsymbol{
\bar{q}%
}_{NT} ( \mathring{\boldsymbol{\eta}} )
\\
&&{}-TN^{-1}\sum_{i=1}^{N}
\biggl[ \frac{1}{2} ( \hat{\boldsymbol{\theta }}%
_{ia}+
\hat{\boldsymbol{\theta }}_{ib} ) - \hat{\boldsymbol{\theta
}}_{i}%
 \biggr] ^{\prime }\boldsymbol{w}_{i,T+1}
\boldsymbol{w}_{i,T+1}^{
\prime }%
\mathring{\boldsymbol{
\eta}}.
\end{eqnarray*}

\subsubsection*{Estimating the weights in combination of individual and
fixed effects forecasts}

Similar to the derivations above, it can be shown that the components of
the weights in Proposition~\ref{prop:combinedFE} can be estimated as follows.
First,
%
\begin{equation}
\label{DeltaFE} \hat{\Delta}_{NT}^{\text{FE}}=\frac{1}{N}\sum
_{i=1}^{N} \bar{\bar{%
\boldsymbol{x}}}_{i,T+1}^{\prime} \mathring{\boldsymbol{\eta
}}_{i,
\beta} \mathring{\boldsymbol{\eta}}_{i,\beta}^{\prime}
\bar{\bar{\boldsymbol{x}}}%
_{i,T+1} - \hat{\bar{
\boldsymbol{q}}}_{NT,\beta}^{\prime } ({ \boldsymbol{%
\mathring{\eta}}_{i,\beta}} ) \hat{\bar{\boldsymbol{Q}}}%
_{NT,\beta}^{-1}
\hat{\bar{\boldsymbol{q}}}_{NT,\beta} ({ \boldsymbol{%
\mathring{\eta}}_{i,\beta}} ),
\end{equation}%
$\mathring{\boldsymbol{\eta}}_{i,\beta}=\hat{\boldsymbol{\beta}}_{i} -
\frac {1}{N}\sum_{i=1}^{N}\hat{\boldsymbol{\beta }}_{i}$,  $
\bar{\bar{\boldsymbol{x}}}%
_{i,T+1}=\boldsymbol{x}_{iT+1}-\bar{\boldsymbol{x}}_{iT}$,  $
\hat{\bar{%
\boldsymbol{Q}}}_{NT,\beta} = N^{-1}T^{-1}\sum_{i=1}^{N}
\boldsymbol{X}%
_{i}^{\prime}\boldsymbol{M}_{T} \boldsymbol{X}_{i}$  and\break
$\hat{\bar{%
\boldsymbol{q}}}_{NT,\beta}  ({\boldsymbol{\mathring{\eta}}_{i,
\beta}}%
  ) = N^{-1}T^{-1}\sum_{i=1}^{N} \boldsymbol{X}_{i}^{\prime}
\boldsymbol{M}_{T} \boldsymbol{X}_{i} \mathring{\boldsymbol{\eta}}_{i,
\beta}.$ Next,
%
\begin{equation}
\hat{h}_{NT,\beta }=N^{-1}\sum_{i=1}^{N}
\bar{\bar{\boldsymbol{x}}}%
_{i,T+1}^{\prime }
\boldsymbol{Q}_{iT,\beta }^{-1} \mathring{\boldsymbol{H}}_{iT,\beta }%
\boldsymbol{Q}_{iT,\beta }^{-1}\bar{\bar{\boldsymbol{x}}}_{i,T+1},
\label{hhatFE}
\end{equation}%
$\mathring{\boldsymbol{H}}_{iT,\beta }=\hat{\sigma}_{i}^{2}
\frac{1}{T}\sum_{t=1}^{T}%
\bar{\bar{\boldsymbol{x}}}_{it}\bar{\bar{\boldsymbol{x}}}_{it}^{
\prime }$,
$%
\hat{\sigma}_{i}^{2}=\hat{\boldsymbol{\varepsilon }}_{i}^{\prime }
\hat{%
\boldsymbol{\varepsilon }}_{i}/(T-K)$, and
$\hat{\varepsilon}_{it}=y_{it}-%
\hat{\boldsymbol{\theta }}_{i}^{\prime }\boldsymbol{w}_{it}$. Furthermore,
%
\begin{eqnarray}
\hat{\psi}_{NT}^{\text{FE}} &=&TN^{-1}\sum
_{i=1}^{N} ( \hat{\boldsymbol{%
\beta
}}_{\text{FE}}-\hat{\boldsymbol{\beta }}_{\text{FEJK}} )
^{
\prime }\bar{\bar{\boldsymbol{x}}}_{i,T+1} \bar{\bar{
\boldsymbol{x}}}%
_{i,T+1}^{\prime }\bar{
\boldsymbol{Q}}_{N,\beta }^{-1} \bar{\boldsymbol{q}}%
_{N,\beta }(
\mathring{\boldsymbol{\eta }}_{i,\beta })\notag
\\
&&{}-TN^{-1}\sum_{i=1}^{N} (
\hat{\boldsymbol{\beta }}_{\text{FE}}- \hat{%
\boldsymbol{\beta
}}_{\text{FEJK}} ) ^{\prime } \bar{\bar{\boldsymbol{x}}%
}_{i,T+1}
\bar{\bar{\boldsymbol{x}}}_{i,T+1}^{\prime } \mathring{\boldsymbol{
\eta }}%
_{i,\beta }, \label{psihatFE}
\end{eqnarray}%
where
$\hat{\boldsymbol{\beta }}_{\text{FEJK}}=2\hat{\boldsymbol{\beta }}_{%
\text{FE}}-\frac{1}{2}  ( \hat{\boldsymbol{\beta }}_{\text{FE},a}+
\hat{%
\boldsymbol{\beta }}_{\text{FE},b}  ) $ is the half-jackknife estimator
of \cite{Chuetal2018}.

Finally, the weights include the difference between
\begin{equation*}
c_{NT}^{\text{FE}}=\frac{1}{N}\sum
_{i=1}^{N} ( \hat{\boldsymbol{\beta
}}%
_{\text{FE}}-\boldsymbol{\beta }_{i} )
^{\prime } \bar{\bar{\boldsymbol{x%
}}}_{i,T+1}\bar{
\hat{\varepsilon}}_{iT}  \quad \text{and}\quad
\hat{c}%
_{NT,\beta }=\frac{1}{N}\sum_{i=1}^{N}
( \hat{\boldsymbol{\beta }}_{i}-{%
\boldsymbol{\beta
}}_{i} ) ^{\prime }\bar{\bar{\boldsymbol{x}}}_{i,T+1}%
\bar{\hat{\varepsilon}}_{iT}.
\end{equation*}%
However,
\begin{equation*}
\hat{c}_{NT}^{\text{FE}}-\hat{c}_{NT,\beta }=\frac{1}{N}
\sum_{i=1}^{N} ( \hat{\boldsymbol{\beta
}}_{\text{FE}}- \hat{\boldsymbol{\beta }}_{i} )
^{\prime } \bar{\bar{\boldsymbol{x}}}_{i,T+1}\bar{\hat{
\varepsilon}}%
_{iT}=O_{p}\bigl(T^{-1}
\bigr).
\end{equation*}

\section{Panel AR(1): An example of correlated heterogeneity}\label{PanelAR}

Correlated heterogeneity can arise in many contexts. One important example
is dynamic panel data models where, barring special cases, heterogeneity
is correlated by design. As a simple example, consider the stationary panel
AR(1) case where $y_{it}=\beta _{i}y_{i,t-1}+\varepsilon _{it}$, for
$%
t=\cdots -2-1,0,1,\dots ,T,T+1,\dots $, and
$\sup_{i} \llvert  \beta _{i} \rrvert  \leq c$ for some positive
$c<1$, and $\beta _{i}$ follows a random coefficient model
$\beta _{i}=\beta _{0}+\eta _{i}$, where $\beta _{0}=E(\beta _{i})$, and
$\eta _{i}$ is suitably truncated such that the stationary condition
$\sup_{i} \llvert  \beta _{i} \rrvert  \leq c$ is met.

Suppose our objective is to forecast $y_{iT+1}$ based on the observations
$%
  \{ y_{it},t=0,1,2,\ldots ,T  \} $.\footnote{The assumption that the process for $y_{it}$ has started a long time prior
to date $0$, is equivalent to assuming that $y_{i0}$ is drawn from a distribution
with zero mean and variance $\sigma _{i}^{2}/(1-\beta _{i}^{2})$%
.} In the context of the general linear model analyzed in the paper,
$%
\boldsymbol{w}_{it}=y_{i,t-1}$ and
$\boldsymbol{\theta }_{i}=\beta _{i}$. It is easily verified that our Assumptions
\ref{ass:1}--\ref{ass:6} cover the dynamic case where one or more elements
of $\boldsymbol{w}_{it}$ are lagged values of $y_{it}$. Forecasts based
on pooled estimates, which incorrectly assume
$\beta _{i}=\beta _{0}$ generate a heterogeneity bias, $\Delta _{N}$, given
by (\ref{DeltaSt_N}). In the present example $q_{i}=$ E$  ( y_{i,t-1}^{2}
\eta _{i}  ) $, $Q_{i}=$ E$  ( y_{i,t-1}^{2}  ) $, and
\begin{equation*}
\Delta _{N}=N^{-1}\sum_{i=1}^{N}
\mathrm{E} \bigl( y_{it}^{2}\eta _{i}^{2}
\bigr) - \frac{ \Biggl[ N^{-1}\sum_{i=1}^{N}
\mathrm{E} \bigl( y_{i,t-1}^{2}\eta _{i} \bigr)
\Biggr] ^{2}}{N^{-1}\sum_{i=1}^{N}
\mathrm{E}%
 \bigl( y_{i,t-1}^{2} \bigr) },
\end{equation*}%
where $q_{i}$ measures the degree of correlated heterogeneity. To derive
$%
q_{i}$ for the AR model, note that
%
\begin{equation}
y_{it}=\sum_{s=0}^{\infty }\beta
_{i}^{s}\varepsilon _{i,t-s}=\sum
_{s=0}^{
\infty } ( \beta _{0}+\eta
_{i} ) ^{s}\varepsilon _{i,t-s},
\label{yitAR}
\end{equation}%
so $y_{it}$ is a nonlinear function of $\eta _{i}$, and, in general,
$q_{i}=%
\mathrm{E}  ( y_{i,t-1}^{2}\eta _{i}  ) \neq 0$. This shows that
heterogeneity in panel AR models generates correlated heterogeneity as
is also implicit in the analysis of \citet{PesSmi1995}. Using (\ref{yitAR}%
), we have
\begin{eqnarray*}
\mathrm{E}(y_{it})&=&0,
\\
  Q_{i}&=&\mathrm{E}
\bigl(y_{it}^{2}\bigr)=\mathrm{E} \bigl(
y_{i,t-1}^{2} \bigr) =\mathrm{E} \biggl( \frac{\sigma
_{i}^{2}}{1-\beta _{i}^{2}}
\biggr),\quad  \text{for all }t,%
\\
q_{i}&=&\mathrm{E} \bigl( y_{i,t-1}^{2}\eta
_{i} \bigr) =\sum_{s=0}^{
\infty }%
\mathrm{E} \bigl( \beta _{i}^{2s}\eta _{i}
\sigma _{i}^{2} \bigr) = \mathrm{E}%
 \biggl(
\frac{\eta _{i}\sigma _{i}^{2}}{1-\beta
_{i}^{2}} \biggr) ,\quad  \text{and}
\\
\mathrm{E} \bigl( y_{it}^{2}\eta _{i}^{2}
\bigr) &=&\mathrm{E} \biggl( \frac{\eta _{i}^{2}\sigma
_{i}^{2}}{1-\beta _{i}^{2}}
\biggr) .%
\end{eqnarray*}In this simple example, heterogeneity is uncorrelated only if
$\beta _{0}=0$ and $\eta _{i}$ is symmetrically distributed around
$0$. This follows since when $\beta _{0}=0$ we have
$q_{i}=\mathrm{E}  (
\frac{\eta _{i}\sigma _{i}^{2}}{1-\eta _{i}^{2}}  ) $ and under symmetry
$\eta _{i}\sigma _{i}^{2}/  ( 1-\eta _{i}^{2}  ) $ is an odd function
of $\eta _{i}$, which yields $q_{i}=0$. But when $\beta _{0}\neq 0$, then
$q_{i}\neq 0$ even if $\eta _{i}$ has a symmetric distribution. The expression
for $\Delta _{N}$ is strictly positive irrespective of whether
$q_{i}=0$ or not. Under stationarity, $\Delta _{N}$ simplifies to
%
\begin{eqnarray}
\Delta _{\text{AR}} &=&\mathrm{E} \bigl( y_{it}^{2}
\eta _{i}^{2} \bigr) -%
\frac{ \bigl[
\mathrm{E} \bigl( y_{i,t-1}^{2}\eta _{i} \bigr)
\bigr] ^{2}}{%
\mathrm{E} \bigl( y_{i,t-1}^{2}
\bigr) } \notag
\\
&=& \frac{\mathrm{E} \biggl( \frac{\eta _{i}^{2}\sigma
_{i}^{2}}{1-\beta _{i}^{2}}
\biggr) \mathrm{E} \biggl( \frac{\sigma _{i}^{2}}{1-\beta
_{i}^{2}}%
 \biggr) - \biggl[ \mathrm{E}
\biggl( \frac{\eta _{i}\sigma _{i}^{2}}{1-
\beta _{i}^{2}} \biggr) \biggr] ^{2}}{
\mathrm{E} \biggl( \frac{\sigma _{i}^{2}}{%
1-
\beta _{i}^{2}} \biggr) }. \label{deltaAR}
\end{eqnarray}%
Let $f_{i}=\sigma _{i}\eta _{i}/\sqrt{1-\beta _{i}^{2}}$ and
$g_{i}=\sigma _{i}/\sqrt{1-\beta _{i}^{2}}$, and note that the numerator
of $\Delta _{%
\text{AR}}$ can be written as E$(f_{i}^{2})$E($g_{i}^{2})-  [ E(f_{i}g_{i})
  ] ^{2}\geq 0$, which establishes that
$\Delta _{\text{AR}%
} \geq 0$, in line with part (c) of Proposition~\ref{prop:pooled}.

The magnitude of $\Delta _{\text{AR}} $ depends on the joint distribution
of $\beta _{i}$ and $\sigma _{i}^{2}$. As an example, consider the case
where $\sigma _{i}^{2}$ and $\beta _{i}$ are independently distributed,
E$%
(\sigma _{i}^{2})=\sigma ^{2}$ and
$\eta _{i}\sim \operatorname{Uniform}(-a/2,a/2)$%
, for $a>0$.\footnote{Note in this case $\eta _{i}$ is symmetrically distributed around
$0$.} Then
\begin{equation*}
q_{i}=\sigma ^{2}\mathrm{E} \biggl( \frac{\eta
_{i}}{1-\beta _{i}^{2}} \biggr)
=%
\frac{\sigma ^{2}}{2} \biggl[ \mathrm{E} \biggl( \frac{
\eta _{i}}{1-\beta _{0}-\eta _{i}} \biggr)
+\mathrm{E} \biggl( \frac{\eta _{i}}{1+\beta _{0}+\eta
_{i}} \biggr) \biggr].
\end{equation*}%
To derive the expectations in the above expression, note that for a given
$B$%
, such that $B^{2}-a^{2}/4>0$, we have
%
\begin{equation}
\mathrm{E} \biggl( \frac{\eta _{i}}{B+\eta _{i}}
\biggr) = \frac{1}{a}%
\int _{-a/2}^{a/2} \biggl( \frac{\eta }{B+\eta }
\biggr)\,d\eta =1- \biggl( \frac{B%
}{a} \biggr) \ln \biggl(
\frac{B+a/2}{B-a/2} \biggr) . \label{Eetainv}
\end{equation}%
Using this result, and setting $B=1+\beta _{0}$, we have, for
$(1+\beta _{0})^{2}>a^{2}/4$,
\begin{equation*}
\mathrm{E} \biggl( \frac{\eta _{i}}{1+\beta _{0}+\eta
_{i}} \biggr) =1- \biggl( \frac{1+\beta _{0}}{a}
\biggr) \ln \biggl( \frac{1+\beta _{0}+a/2}{1+\beta
_{0}-a/2} \biggr) .
\end{equation*}%
Similarly, again for $(\beta _{0}-1)^{2}>a^{2}/4$,
\begin{equation*}
\mathrm{E} \biggl( \frac{\eta _{i}}{1-\beta _{0}-\eta
_{i}} \biggr) =- \mathrm{E}%
 \biggl( \frac{\eta
_{i}}{\beta -1+\eta _{i}} \biggr) =- \biggl[ 1- \biggl(
\frac{%
\beta _{0}-1}{a} \biggr) \ln \biggl( \frac{\beta
_{0}-1+a/2}{\beta _{0}-1-a/2}%
 \biggr)
\biggr] .
\end{equation*}%
Overall, assuming that $a/2<1- \llvert  \beta _{0} \rrvert  $, we
have%
%
\begin{eqnarray}
\mathrm{E} \bigl( y_{i,t-1}^{2}\eta _{i}
\bigr) &=& \frac{\sigma ^{2}}{2} \biggl[ \biggl( \frac{\beta
_{0}-1}{a} \biggr) \ln \biggl( \frac{\beta _{0}-1+a/2}{%
\beta _{0}-1-a/2} \biggr) - \biggl( \frac{1+\beta
_{0}}{a} \biggr) \ln \biggl( \frac{1+\beta _{0}+a/2}{1+
\beta _{0}-a/2} \biggr) \biggr] \notag
\\
&=&\frac{\sigma ^{2}}{2a} \biggl[ - ( 1-\beta _{0} )
\ln \biggl( \frac{%
1-\beta _{0}-a/2}{1-\beta
_{0}+a/2} \biggr) - ( 1+\beta _{0} ) \ln \biggl(
\frac{1+\beta _{0}+a/2}{1+\beta _{0}-a/2} \biggr)
\biggr] . \label{y2eta}
\end{eqnarray}%
To ensure that
$ \llvert  \beta _{i} \rrvert  = \llvert  \beta _{0}+\eta _{i}
 \rrvert  <1$, we require that $a$ is sufficiently small relative to
$%
\beta _{0}$ and $ \llvert  \beta _{0} \rrvert  <1$. A sufficient condition
for this to hold is that
\begin{equation*}
\llvert \beta _{i} \rrvert = \llvert \beta _{0}+\eta
_{i} \rrvert \leq \llvert \beta _{0} \rrvert + \llvert
\eta _{i} \rrvert = \llvert \beta _{0} \rrvert +a/2<1.
\end{equation*}

\begin{table}[tbp]
\caption{Numerical values for $\mathrm{E}\left( y_{i,t-1}^{2}\protect\eta %
_{i}\right)$ and $\Delta_\text{AR}$ for the panel AR(1) model}
\label{tbl:numericalValuesAR1}\centering
\begin{tabular}{lcc}
\hline\hline
$\beta_{0}$ & $\mathrm{E}( y_{i,t-1}^{2}\eta _{i})$ & $\Delta_\mathrm{AR}$
\\ \hline
0.3 & 0.100 & 0.117\\
0.45 & 0.316 & 0.163\\
0.49 & 0.657 & 0.211\\
0.4999 & 1.783 & 0.328\\
\hline\hline
\multicolumn{3}{p{4.6cm}}{{\footnotesize {Note: The numerical values are
based on $a=\sigma^2=1$.}}}%
\end{tabular}%
\end{table}

In general, for $a>0$ and $ \llvert  \beta _{0} \rrvert  <1$,
$\mathrm{E}%
  ( y_{i,t-1}^{2}\eta _{i}  ) \neq 0$,
$\mathrm{E}  ( y_{i,t-1}^{2}\eta _{i}  ) \rightarrow 0$, only
if $a\rightarrow 0$. Since
$\eta _{i}\sim \mathit{iid}\operatorname{Uniform}(-a/2,a/2)$ is symmetrically
distributed, then
$\mathrm{E}  ( y_{i,t-1}^{2}\eta _{i}  ) =0$ for
$\beta _{0}=0$. But
$\operatorname{Cov}(y_{i,t-1}^{2},\eta _{i}^{2})\neq 0$, even under symmetry
and $y_{i,t-1}^{2}$ and $\eta _{i}$ are not independently distributed.
For example, when $\beta _{0}=0$, we have
\begin{equation*}
\mathrm{E}\bigl(y_{it}^{2}\eta _{i}^{2}
\bigr)=\sigma ^{2}\mathrm{E} \biggl( \frac{\eta _{i}^{2}}{1-
\eta _{i}^{2}} \biggr) \neq \mathrm{E}
\bigl(y_{it}^{2}\bigr) \mathrm{E}\bigl(\eta
_{i}^{2}\bigr)=\sigma ^{2}\mathrm{E} \biggl(
\frac{1}{1-\eta _{i}^{2}} \biggr) \mathrm{E}\bigl(\eta
_{i}^{2}\bigr).
\end{equation*}%
When $\beta _{i}$ and $\sigma _{i}^{2}$ are independently distributed,
using (\ref{deltaAR}), we have
\begin{equation*}
\sigma ^{-2}\Delta _{\text{AR}}= \frac{\mathrm{E} \biggl(
\frac{\eta _{i}^{2}}{%
1-\beta
_{i}^{2}} \biggr) \mathrm{E} \biggl( \frac{1}{1-\beta
_{i}^{2}} \biggr) -%
 \biggl[ \mathrm{E}
\biggl( \frac{\eta _{i}}{1-\beta _{i}^{2}}
\biggr) \biggr] ^{2}}{\mathrm{E} \biggl( \frac{1}{1-\beta
_{i}^{2}} \biggr) }.
\end{equation*}%
We can derive an analytical expression for
$\mathrm{E}  ( \frac{1}{%
1-\beta _{i}^{2}}  ) $, noting that
\begin{equation*}
\mathrm{E} \biggl( \frac{1}{B+\eta _{i}} \biggr) =
\frac{1}{a}%
\int _{-a/2}^{a/2} \biggl( \frac{1}{B+\eta }
\biggr)\,d\eta = \frac{1}{a}\ln \biggl( \frac{B+a/2}{B-a/2} \biggr) .
\end{equation*}%
Hence,
\begin{eqnarray*}
\mathrm{E} \biggl( \frac{1}{1-\beta _{i}^{2}} \biggr)
&=&\frac{1}{2} \biggl[ -%
\mathrm{E} \biggl( \frac{1}{-1+\beta
_{0}+\eta _{i}} \biggr) + \mathrm{E} \biggl(
\frac{1}{1+\beta _{0}+\eta _{i}} \biggr) \biggr]
\\
&=&-\frac{1}{2a}\ln \biggl( \frac{\beta _{0}-1+a/2}{\beta
_{0}-1-a/2} \biggr) +%
\frac{1}{2a}\ln \biggl( \frac{\beta
_{0}+1+a/2}{\beta _{0}+1-a/2} \biggr) ,
\end{eqnarray*}%
or%
%
\begin{equation}
\mathrm{E} \biggl( \frac{1}{1-\beta _{i}^{2}} \biggr) =
\frac{1}{2a} \biggl[ \ln \biggl( \frac{1+\beta _{0}+a/2}{1+\beta
_{0}-a/2} \biggr) - \ln \biggl( \frac{%
1-\beta
_{0}-a/2}{1-\beta _{0}+a/2} \biggr) \biggr] .
\label{Evar}
\end{equation}

Using (\ref{Evar}) and simulated values of
$\mathrm{E}  ( \frac{\eta _{i}%
}{1-\beta _{i}^{2}}  )$ and
$\mathrm{E}  ( \frac{\eta _{i}^{2}}{%
1-\beta _{i}^{2}}  ) $, we obtain the values of
$\Delta _{\text{AR}}$ for $%
\alpha =1$ and $\sigma ^{2}=1$ that are reported in Table~\ref{tbl:numericalValuesAR1}
for $10{,}000$ replications.

\newpage
\bibliographystyle{apalike}
\bibliography{panelForecasting}

\newpage \setcounter{section}{0} \setcounter{equation}{0} %
\setcounter{page}{1} \setcounter{table}{0} 
\renewcommand{\thesection}{S.\arabic{section}} \renewcommand{%
\theequation}{S.\arabic{equation}} \renewcommand{\thepage}{S-\arabic{page}} %
\renewcommand{\thetable}{S.\arabic{table}}

\begin{center}
{\LARGE Supplement to}

{\LARGE Forecasting with panel data: Estimation uncertainty versus parameter
heterogeneity}
 \\ \ \\

\bigskip

M. Hashem Pesaran\footnote{%
University of Cambridge, UK, and University of Southern California, USA
Email: mhp1@cam.ac.uk}

Andreas Pick\footnote{%
Erasmus University Rotterdam, Erasmus School of Economics, Burgemeester
Oudlaan 50, 3000DR Rotterdam, and Tinbergen Institute. Email:
andreas.pick@cantab.net}

and

Allan Timmermann\footnote{%
UC San Diego, Rady School of Management, 9500 Gilman Drive, La Jolla CA
92093-0553. Email: atimmermann@ucsd.edu.}

\bigskip

\end{center}
\bigskip

\hrule

\vspace{2em}\newpage

\section{Introduction}

This supplementary appendix provides additional material underpinning the
analysis in the main paper along with a set of extensions to the Monte Carlo
simulations and empirical results. We begin by deriving in Section \ref%
{app:R2} the pooled R-squared, $\mathit{PR}_{N}^{2}$, used in the Monte
Carlo simulations to target the predictive power of our panel forecasting
models. We characterize $\mathit{PR}_{N}^{2}$ as a function of the
underlying parameters of the DGPs and use this to calibrate the parameters
used in the simulations. Next, Section \ref{sec:MCestimators} provides
details of how we implement the estimators used in our analysis. Section \ref%
{addMC} provides additional simulation and empirical results.

\section{Derivation of the pooled R-squared $\mathit{PR}_{N}^{2}$}

\label{app:R2}

Consider the panel data model%
\begin{equation}
y_{it}=\alpha _{i}+\beta _{i}y_{i,t-1}+\gamma _{i}x_{it}+\varepsilon _{it},
\label{eq:ARDL}
\end{equation}%
\begin{equation*}
x_{it}=\mu _{xi}+\xi _{it},\text{ \ }\xi _{it}=\rho _{xi}\xi _{i,t-1}+\sigma
_{xi}\sqrt{1-\rho _{xi}^{2}}\nu _{it}.
\end{equation*}%
Further, Var($\varepsilon _{it})=1$, and Var($\nu _{it})=1$ as set out in
further detail in Section~\ref{sec:MC} of the paper. To simplify the derivations, we
treat $x_{it}$ as strictly exogenous (no feedback from $y_{i,t-1})$ and
assume that $y_{it}$ is stationary and started a long time in the past. To
deal with the heterogeneity across the different equations in the panel, we
use the following average measure of fit, for a given $N$, 
\begin{equation}
\mathit{PR}_{N}^{2}=1-\frac{N^{-1}\sum_{i=1}^{N}\mathrm{Var}\left(
\varepsilon _{it}\left\vert \boldsymbol{\theta }_{i}\text{, }x_{it}\right.
\right) }{N^{-1}\sum_{i=1}^{N}\mathrm{Var}(y_{it}\left\vert \boldsymbol{%
\theta }_{i}\text{, }x_{it},\right. )},  \label{PR2}
\end{equation}%
where as before $\boldsymbol{\theta }_{i}=(\alpha _{i},\beta _{i},\gamma
_{i})^{\prime }$. For the numerator we have 
\begin{equation}
\mathrm{Var}\left( \varepsilon _{it}\left\vert \boldsymbol{\ \theta }_{i}%
\text{, }\sigma _{i}^{2}\text{, }x_{it}\text{ }\right. \right) =\sigma
_{i}^{2}.  \label{v1a}
\end{equation}

To derive $\mathrm{Var}(y_{it}\left\vert \boldsymbol{\theta }_{i}\text{, }%
x_{it}\right. )$, we note that 
\begin{eqnarray*}
\mathrm{Var}(y_{it}\left\vert \boldsymbol{\theta }_{i},\sigma
_{i}^{2},x_{it}\right. ) &=&\mathrm{E}\left[ \mathrm{Var}(y_{it}\left\vert 
\boldsymbol{\theta }_{i},\sigma _{i}^{2}\text{,}y_{i,t-1},x_{it}\right. )%
\right] +\mathrm{Var}\left[ \mathrm{E}(y_{it}\left\vert \boldsymbol{\theta }%
_{i},\sigma _{i}^{2},y_{i,t-1},x_{it}\right. )\right] , \\
\mathrm{E}(y_{it}\left\vert \boldsymbol{\theta }_{i},\sigma
_{i}^{2},y_{i,t-1},x_{it}\right. ) &=&\alpha _{i}+\beta _{i}y_{i,t-1}+\gamma
_{i}x_{it},\text{ }\mathrm{Var}(y_{it}\left\vert \boldsymbol{\theta }%
_{i},\sigma _{i}^{2},y_{i,t-1},x_{it}\right. )=\sigma _{i}^{2}, \\
\mathrm{Var}\left[ \mathrm{E}(y_{it}\left\vert \boldsymbol{\theta }%
_{i},\sigma _{i}^{2},y_{i,t-1},x_{it}\right. )\right] &=&\beta _{i}^{2}%
\mathrm{Var}(y_{it}\left\vert \boldsymbol{\theta }_{i},\sigma
_{i}^{2},x_{it}\right. )+\gamma _{i}^{2}\mathrm{Var}\left( x_{it}\right) .
\end{eqnarray*}%
Noting that $\mathrm{Var}(x_{it}) = \mathrm{Var}(\xi_{it})$, we have
\begin{equation}
\mathrm{Var}(y_{it}\left\vert \boldsymbol{\theta }_{i},\sigma
_{i}^{2},x_{it}\right. )=\frac{\gamma _{i}^{2}\mathrm{Var}(\xi _{it})+\sigma
_{i}^{2}}{1-\beta _{i}^{2}}.  \label{v1b}
\end{equation}

Now using (\ref{v1a}) and (\ref{v1b}) in (\ref{PR2}), we obtain 
\begin{equation*}
\mathit{PR}_{N}^{2}=1-\left( \frac{N^{-1}\sum_{i=1}^{N}\sigma _{i}^{2}}{%
N^{-1}\sum_{i=1}^{N}\frac{\gamma _{i}^{2}\sigma _{xi}^{2}+\sigma _{i}^{2}}{%
1-\beta _{i}^{2}}}\right) ,
\end{equation*}%
where $\sigma _{xi}^{2}=\mathrm{Var}(\xi _{it})$. After some simplifications
we have 
\begin{equation}
\mathit{PR}_{N}^{2}=\frac{b_{N}+\left( c_{N}-a_{N}\right) }{b_{N}+c_{N}},
\label{PRN}
\end{equation}%
where $a_{N}=N^{-1}\sum_{i=1}^{N}\sigma _{i}^{2},\quad
b_{N}=N^{-1}\sum_{i=1}^{N}\frac{\gamma _{i}^{2}\sigma _{xi}^{2}}{1-\beta
_{i}^{2}},$ and $c_{N}=N^{-1}\sum_{i=1}^{N}\frac{\sigma _{i}^{2}}{1-\beta
_{i}^{2}}$.

When these parameters are distributed independently, as $N\rightarrow \infty 
$, we obtain 
\begin{eqnarray*}
a_{N} &\overset{p}{\rightarrow }&\mathrm{E}(\sigma _{i}^{2}),\text{ }b_{N}%
\overset{p}{\rightarrow }\mathrm{E}(\gamma _{i}^{2})\mathrm{E}(\sigma
_{xi}^{2})\mathrm{E}\left( \frac{1}{1-\beta _{i}^{2}}\right) , \\
c_{N} &\overset{p}{\rightarrow }&\mathrm{E}(\sigma _{i}^{2})\mathrm{E}\left( 
\frac{1}{1-\beta _{i}^{2}}\right) .
\end{eqnarray*}%
Hence, using (\ref{PRN}), we note that (as $N\rightarrow \infty $) 
\begin{equation*}
\mathit{PR}_{N}^{2}\rightarrow \mathit{PR}^{2}=\frac{\mathrm{E}(\gamma
_{i}^{2})\mathrm{E}(\sigma _{xi}^{2})\mathrm{E}\left( \frac{1}{1-\beta
_{i}^{2}}\right) +\left[ \mathrm{E}(\sigma _{i}^{2})\mathrm{E}\left( \frac{1%
}{1-\beta _{i}^{2}}\right) -\mathrm{E}(\sigma _{i}^{2})\right] }{\mathrm{E}%
(\gamma _{i}^{2})\mathrm{E}(\sigma _{xi}^{2})\mathrm{E}\left( \frac{1}{%
1-\beta _{i}^{2}}\right) +\mathrm{E}(\sigma _{i}^{2})\mathrm{E}\left( \frac{1%
}{1-\beta _{i}^{2}}\right) }.
\end{equation*}%
Under our design $\mathrm{E}(\sigma _{i}^{2})=1$, $\mathrm{E}(\sigma
_{xi}^{2})=1$, and the above expression further simplifies to 
\begin{equation}
\mathit{PR}^{2}=\frac{\mathrm{E}(\gamma _{i}^{2})\mathrm{E}\left( \frac{1}{%
1-\beta _{i}^{2}}\right) +\left[ \mathrm{E}\left( \frac{1}{1-\beta _{i}^{2}}%
\right) -\text{\textrm{1}}\right] }{\mathrm{E}(\gamma _{i}^{2})\mathrm{E}%
\left( \frac{1}{1-\beta _{i}^{2}}\right) +\mathrm{E}\left( \frac{1}{1-\beta
_{i}^{2}}\right) }.  \label{eq:ExpPR2}
\end{equation}

In the general case where $\sigma _{i}^{2}$ is not distributed independently
of $\beta _{i}$ and $N$ is finite, we have 
\begin{equation*}
\mathit{PR}_{N}^{2}>1-a_{N}/c_{N}=1-\frac{N^{-1}\sum_{i=1}^{N}\sigma _{i}^{2}%
}{N^{-1}\sum_{i=1}^{N}\frac{\sigma _{i}^{2}}{1-\beta _{i}^{2}}}.
\end{equation*}

In the case where $\beta _{i}=\beta _{0}+\eta _{i\beta }$, and $\eta
_{i\beta }\sim \mathrm{iid}\ \mathrm{Uniform}(-\alpha _{\beta }/2,\alpha _{\beta }/2)$, $%
\alpha _{\beta }>0$, we have (see also (\ref{Evar}) in the Appendix to the
paper): 
\begin{equation}
\begin{array}{ll}
\mathrm{E}\left( \frac{1}{1-\beta _{i}^{2}}\right) & =\frac{1}{a_{\beta }}%
\int_{-a_{\beta }/2}^{a_{\beta }/2}\frac{1}{1-\left( \beta _{0}+\eta _{\beta
}\right) ^{2}}d\eta _{\beta } \\ 
& =\frac{1}{2a_{\beta }}\int_{-a_{\beta }/2}^{a_{\beta }/2}\left[ \frac{1}{%
1+\beta _{0}+\eta _{\beta }}+\frac{1}{1-\beta _{0}-\eta _{\beta }}\right]
d\eta _{\beta } \\ 
& =\frac{1}{2a_{\beta }}\left[ ln(1+\beta _{0}+\eta _{\beta })-ln(1-\beta
_{0}-\eta _{\beta })\right] _{-a_{\beta }/2}^{a_{\beta }/2} \\ 
& =\frac{1}{2a_{\beta }}\left[ ln\left( \frac{1+\beta _{0}+a_{\beta }/2}{%
1+\beta _{0}-a_{\beta }/2}\right) -ln\left( \frac{1-\beta _{0}-a_{\beta }/2}{%
1-\beta _{0}+a_{\beta }/2}\right) \right] ,%
\end{array}
\label{eq:expectation1}
\end{equation}%
assuming that 
\begin{equation*}
\left( 1+\beta _{0}+a_{\beta }/2\right) \left( 1+\beta _{0}-a_{\beta
}/2\right) >0\text{ and }\left( 1-\beta _{0}-a_{\beta }/2\right) \left(
1-\beta _{0}+a_{\beta }/2\right) >0,
\end{equation*}%
or if%
\begin{equation}
0\leq a_{\beta }<2\left( 1-\left\vert \beta _{0}\right\vert \right) .
\label{Conrho}
\end{equation}%
It is easily established that $\mathrm{E}\left( \frac{1}{1-\beta _{i}^{2}}%
\right) \rightarrow \frac{1}{1-\beta _{0}^{2}}$, as $a_{\beta }\rightarrow 0$%
.

Our Monte Carlo simulations target an $\mathit{PR}^{2}$ of 0.6. We do so by calibrating the values of
the $a_{\beta}$ and $\beta_ {0}$ parameters. The values of the parameters
used to this end are reported in Table~\ref{tbl:parametersMCARX}.

\begin{table}[h!]
\caption{\textit{PR}$^2$ for parameters of Monte Carlo models}
\label{tbl:parametersMCARX}\centering
\begin{tabular}{llcc}
\hline\hline
$a_\beta$ & $\beta_0$ & $\mathit{PR}_{ARX}^2(\rho_{\gamma x}=0)$ & $\mathit{%
PR} _{ARX}^2(\rho_{\gamma x}=0.5)$ \\ \hline
0 & 0.775 & 0.605 & 0.605 \\ 
0.5 & 0.688 & 0.640 & 0.651 \\ 
1 & 0.486 & 0.669 & 0.686 \\ \hline\hline
\multicolumn{4}{p{8.5cm}}{{\footnotesize {Note: The table reports the
parameters for $a_\beta$ and $\beta_0$ in the first two columns and the
implied values for \textit{PR}$^2$ in the remaining columns.}}}%
\end{tabular}%
\end{table}


\section{Details of the estimators\label{sec:MCestimators}}

This section provides details on the implementation of the estimators and
forecasts used in the Monte Carlo experiments and empirical applications.
Recall that the DGP, equation~(\ref{eq:MC1}) in the paper, in the Monte Carlo experiments is%
\begin{equation}
y_{it}=\alpha _{i}+\beta _{i}y_{i,t-1}+\gamma _{i}x_{it}+\varepsilon
_{it}=\alpha _{i}+\boldsymbol{\beta }_{i}^{\prime }\mathbf{x}%
_{it}+\varepsilon _{it}=\boldsymbol{\theta }_{i}^{\prime }\boldsymbol{w}%
_{it}+\varepsilon _{it},\quad \varepsilon _{it}\sim (0,\sigma _{i}^{2}),
\label{eq:model_forecasting_methods}
\end{equation}%
for $t=1,2,\ldots ,T$ and $i=1,2,\ldots ,N$, where $\boldsymbol{\beta }%
_{i}=(\beta _{i},\gamma _{i})^{\prime }$, $\boldsymbol{\theta }_{i}=(\alpha
_{i},\boldsymbol{\beta }_{i}^{\prime })^{\prime }$, $\mathbf{x}%
_{it}=(y_{i,t-1},x_{it})^{\prime }$, and $\boldsymbol{w}_{it}=(1,\mathbf{x}%
_{it}^{\prime })^{\prime }$. Here we consider a more general case where the
dimension of $\boldsymbol{x}_{it}$ is $k\times 1$ and that of $\boldsymbol{w}%
_{it}$ is $K\times 1$, where $K=k+1$. In principle, $\mathbf{x}_{it}$ could
include higher order lags of $y_{it}$ and $x_{it}$, and other covariates. As
in the paper, for simplicity we do not explicitly refer to the
forecast horizon, $h$, but it is assumed that $\boldsymbol{x}_{it}$ contains
information known at time $t-h$. Below we assume a forecast horizon of $h=1$.

\begin{description}
\item[Individual forecasts] \ The individual-specific forecasts based on the
data of a given cross-sectional unit are 
\begin{equation}
\hat{y}_{i,T+1}=\hat{\alpha}_{i,T}+\hat{\boldsymbol{\beta }}_{i,T}^{\prime }%
\boldsymbol{x}_{i,T+1}=\boldsymbol{\hat{\theta}}_{i,T}^{\prime }\boldsymbol{w%
}_{i,T+1}  \label{eq:indForeMC}
\end{equation}%
The parameters are estimated using the estimation sample containing $T$
observations: $\boldsymbol{y}_{i}=(y_{i1},y_{i2},\ldots ,y_{iT})^{\prime }$
and $\boldsymbol{X}_{i}=(\boldsymbol{x}_{i1},\boldsymbol{x}_{i2},\ldots ,%
\boldsymbol{x}_{iT})^{\prime }$. In matrix notation, the model is 
\begin{equation*}
\boldsymbol{y}_{i}=\alpha _{i}\boldsymbol{\tau }_{T}+\boldsymbol{X}_{i}%
\boldsymbol{\beta }_{i}+\boldsymbol{\varepsilon }_{i}=\boldsymbol{W}_{i}%
\boldsymbol{\theta }_{i}+\boldsymbol{\varepsilon }_{i},
\end{equation*}%
where $\boldsymbol{\tau }$ is a $T\times 1$ unit vector, $\boldsymbol{W}%
_{i}=(\boldsymbol{w}_{i1},\boldsymbol{w}_{i2},\ldots ,\boldsymbol{w}%
_{iT})^{\prime }$, $\boldsymbol{w}_{it}=(1,\mathbf{x}_{it}^{\prime
})^{\prime }$, and $\boldsymbol{\varepsilon }_{i}=(\varepsilon
_{i1},\varepsilon _{i2},\ldots ,\varepsilon _{iT})^{\prime }$.

The parameters are estimated as 
\begin{equation*}
\hat{\boldsymbol{\beta }}_{i,T}=\left( \boldsymbol{X}_{i}^{\prime }%
\boldsymbol{M}_{T}\boldsymbol{X}_{i}\right) ^{-1}\boldsymbol{X}_{i}%
\boldsymbol{M}_{T}\boldsymbol{y}_{i},
\end{equation*}%
\begin{equation*}
\hat{\alpha}_{i,T}=\left( \boldsymbol{\tau }_{T}^{\prime }\boldsymbol{M}_{ix}%
\boldsymbol{\tau }_{T}\right) ^{-1}\boldsymbol{\tau }_{T}^{\prime }%
\boldsymbol{M}_{ix}\boldsymbol{y}_{i},
\end{equation*}%
\begin{equation*}
\boldsymbol{M}_{T}=\boldsymbol{I}_{T}-\boldsymbol{\tau }_{T}\left( 
\boldsymbol{\tau }_{T}^{\prime }\boldsymbol{\tau }_{T}\right) ^{-1}%
\boldsymbol{\tau }_{T}^{\prime }\text{, }\boldsymbol{M}_{ix}=\boldsymbol{I}%
_{T}-\boldsymbol{X}_{i}\left( \boldsymbol{X}_{i}^{\prime }\boldsymbol{X}%
_{i}\right) ^{-1}\boldsymbol{X}_{i}^{\prime }.
\end{equation*}%
Written in more compact form, we have 
\begin{equation}
\boldsymbol{\hat{\theta}}_{i,T}=\left( \boldsymbol{W}_{i}^{\prime }%
\boldsymbol{W}_{i}\right) ^{-1}\boldsymbol{W}_{i}^{\prime }\boldsymbol{y}%
_{i}.  \label{eq:indbeta}
\end{equation}%
The \textquotedblleft individual\textquotedblright\ forecasts in (\ref%
{eq:indForeMC}), for $i=1,2,\ldots ,N$, will be used as the reference
forecast and the MSFE of all other methods are reported as ratios relative
to the MSFE of this forecast, defined by 
\begin{equation}
\mathrm{MSFE}_{\mathit{ref}}=N^{-1}\sum_{i=1}^{N}\left( y_{i,T+1}-%
\boldsymbol{\hat{\theta}}_{i,T}^{\prime }\boldsymbol{w}_{i,T+1}\right) ^{2}.
\label{eq:referenceMSFE}
\end{equation}

\item[Pooled forecasts] \ The forecasts that use the pooled information of
all units in the panel are 
\begin{equation}
\tilde{y}_{i,T+1}=\boldsymbol{\tilde{\theta}}_{\text{pool}}^{\prime }%
\boldsymbol{w}_{i,T+1},  \label{eq:poolForeMC}
\end{equation}%
where 
\begin{equation}
\boldsymbol{\tilde{\theta}}_{\text{pool}}=\left( \boldsymbol{W}^{\prime }%
\boldsymbol{W}\right) ^{-1}\boldsymbol{Wy}=\left( \sum_{i=1}^{N}\boldsymbol{W%
}_{i}^{\prime }\boldsymbol{W}_{i}\right) ^{-1}\sum_{i=1}^{N}\boldsymbol{W}%
_{i}^{\prime }\boldsymbol{y}_{i},  \label{thetapool}
\end{equation}%
and $\boldsymbol{W}=(\boldsymbol{W}_{1}^{\prime },\boldsymbol{W}_{2}^{\prime
},\ldots ,\boldsymbol{W}_{N}^{\prime })^{\prime }$ and $\boldsymbol{y}=(%
\boldsymbol{y}_{1}^{\prime },\boldsymbol{y}_{2}^{\prime },\ldots ,%
\boldsymbol{y}_{N}^{\prime })^{\prime }$.

\item[Fixed effects forecast] \ The FE forecasts are given by 
\begin{equation}
\hat{y}_{i,T+1}^{FE}=\hat{\alpha}_{i,\text{FE}}+\boldsymbol{\hat{\beta}}_{%
\text{FE}}^{\prime }\boldsymbol{x}_{i,T+1},  \label{ForFESup}
\end{equation}%
where%
\begin{equation*}
\boldsymbol{\hat{\beta}}_{\text{FE}}=\left( \sum_{i=1}^{N}\boldsymbol{X}%
_{i}^{\prime }\boldsymbol{M}_{T}\boldsymbol{X}_{i}\right) ^{-1}\sum_{i=1}^{N}%
\boldsymbol{X}_{i}^{\prime }\boldsymbol{M}_{T}\boldsymbol{y}_{i},
\end{equation*}%
and 
\begin{equation*}
\hat{\alpha}_{i,\text{FE}}=\boldsymbol{\tau }_{T}^{\prime }(\boldsymbol{y}%
_{i}-\boldsymbol{\hat{\beta}}_{\text{FE}}^{\prime }\boldsymbol{X}_{i})/T
\end{equation*}

\item[Goldberger's random effects BLUP] \ This forecast uses the best linear
unbiased predictor (BLUP) of Goldberger (1962). For this forecast, the model
is given by: 
\begin{equation*}
y_{i,t+1}=\alpha +\boldsymbol{\beta }^{\prime }\boldsymbol{x}%
_{i,t+1}+\varepsilon _{i,t+1},
\end{equation*}%
where $\varepsilon _{i,t+1}=\eta _{i}+u_{i,t+1}$. The BLUP forecasts are
given as 
\begin{equation}
\hat{y}_{i,T+1}^{RE}=\hat{\alpha}_{\text{RE}}+\hat{\boldsymbol{\beta }}_{%
\text{RE}}^{\prime }\boldsymbol{x}_{i,T+1}+\frac{T\hat{\sigma}_{\eta }^{2}}{T%
\hat{\sigma}_{\eta }^{2}+\hat{\sigma}_{u}^{2}}\bar{\hat{\varepsilon}}_{i},
\label{ForGoldberger}
\end{equation}%
where $\overline{\widehat{\varepsilon }}_{i}$ $=T^{-1}\sum_{t=1}^{T}\widehat{%
\varepsilon }_{it}$ and $\widehat{\varepsilon }_{it}=y_{it}-\hat{\alpha}_{%
\text{RE}}-\boldsymbol{x}_{it}^{\prime }\hat{\boldsymbol{\beta }}_{\text{RE}%
} $. $\hat{\alpha}_{\text{RE}},$ and $\hat{\boldsymbol{\beta }}_{\text{RE}}$
are estimated by GLS using 
\begin{equation*}
\boldsymbol{\hat{\Sigma}}^{-1}=\hat{\sigma}_{u}^{-2}\left( \boldsymbol{M}%
_{T}+\hat{\rho}\boldsymbol{P}_{T}\right)
\end{equation*}%
where $\boldsymbol{P}_{T}=\boldsymbol{I}_{T}-\boldsymbol{M}_{T}$, \quad $%
\hat{\rho}=\hat{\sigma}_{u}^{2}/(T\hat{\sigma}_{\eta }^{2}+\hat{\sigma}%
_{u}^{2})$, 
\begin{equation*}
\hat{\sigma}_{u}^{2}=\frac{1}{N(T-1)-K}\sum_{i=1}^{N}(\boldsymbol{y}_{i}-%
\hat{\alpha}_{i,\mathrm{FE}}-\boldsymbol{X}_{i}\hat{\boldsymbol{\beta }}_{%
\mathrm{FE}})^{\prime }\boldsymbol{M}_{T}(\boldsymbol{y}_{i}-\hat{\alpha}_{i,%
\mathrm{FE}}-\boldsymbol{X}_{i}\hat{\boldsymbol{\beta }}_{\mathrm{FE}})
\end{equation*}%
\begin{equation*}
\hat{\sigma}_{\eta }^{2}=\frac{1}{N-K}\sum_{i=1}^{N}(\bar{y}_{i}-\hat{%
\boldsymbol{\beta }}_{\mathrm{FE}}^{\prime }\bar{\boldsymbol{x}}_{i})^{2}-%
\hat{\sigma}_{u}^{2}/T,
\end{equation*}%
\begin{eqnarray*}
\hat{\boldsymbol{\beta }}_{\text{RE}} &=&\left[ \frac{1}{NT}\sum_{i=1}^{N}%
\boldsymbol{X}_{i}^{\prime }\boldsymbol{M}_{T}\boldsymbol{X}_{i}+\frac{\hat{%
\rho}}{N}\sum_{i=1}^{N}\left( \bar{\boldsymbol{x}}_{i}-\bar{\boldsymbol{x}}%
\right) \left( \bar{\boldsymbol{x}}_{i}-\bar{\boldsymbol{x}}\right) ^{\prime
}\right] ^{-1}\times \\
&&\left[ \frac{1}{NT}\sum_{i=1}^{N}\boldsymbol{X}_{i}^{\prime }\boldsymbol{M}%
_{T}\boldsymbol{y}_{i}+\frac{\hat{\rho}}{N}\sum_{i=1}^{N}\left( \bar{%
\boldsymbol{x}}_{i}-\bar{\boldsymbol{x}}\right) \left( \bar{y}_{i}-\bar{y}%
\right) ^{\prime }\right] ,
\end{eqnarray*}%
and $\hat{\alpha}_{\text{RE}}=\bar{y}-\hat{\boldsymbol{\beta }}_{\text{RE}%
}^{\prime }\bar{\boldsymbol{x}}, $ where%
\begin{equation*}
\bar{\boldsymbol{x}}_{i}=T^{-1}\sum_{t=1}^{T}\boldsymbol{x}_{i,t}\text{, }%
\quad \bar{\boldsymbol{x}}=N^{-1}\sum_{i=1}^{N}\bar{\boldsymbol{x}}_{i},%
\text{ \ \ }\bar{y}_{i}=T^{-1}\sum_{t=1}^{T}y_{it}\text{, }\quad \bar{y}%
=N^{-1}\sum_{i=1}^{N}\bar{y}_{i}.
\end{equation*}%
See Baltagi (2013, pp.\ 999--1001) and Pesaran (2015, pp.~646--649) for
further details.

\item[Combination of individual and pooled forecasts] \ 
\begin{equation*}
\hat{y}_{i,T+1}^{c}=\hat{\omega}_{NT}^{\ast }\hat{y}_{i,T+1}+(1-\hat{\omega}%
_{NT}^{\ast })\tilde{y}_{i,T+1},  \label{ComIP}
\end{equation*}%
where $\hat{y}_{i,T+1}$ and $\tilde{y}_{i,T+1}$ are the individual and
pooled forecasts in~\eqref{eq:indForeMC} and~\eqref{eq:poolForeMC} with
weights (see equation~\eqref{w*NT}
in the paper)
\begin{equation*}
\hat{\omega}_{NT}^{\ast }=\frac{\hat{\Delta}_{NT} - T^{-1}\hat{\psi}_{NT}}{%
\hat{\Delta}_{NT}+T^{-1}\hat{h}_{NT} -2T^{-1}\hat{\psi}_{NT}},
\label{wComIP}
\end{equation*}%
where $\hat{\Delta}_{NT}$, $\hat{h}_{NT}$, and $\hat{\psi}_{NT}$ are given
by (\ref{eq:Deltahat}), (\ref{eq:hhat}) and (\ref{eq:psihat}).

\item[Combination of individual and FE forecasts ] 
\begin{equation*}
y_{i,T+1}^{\ast }(\hat{\omega}_{\text{FE},NT}^{\ast })=\hat{\omega}_{\text{FE%
},NT}^{\ast }\hat{y}_{i,T+1}+(1-\hat{\omega}_{\text{FE},NT}^{\ast })\hat{y}%
_{i,T+1,\text{FE}},  \label{ComIFE}
\end{equation*}%
where $\hat{y}_{i,T+1}$ and $\hat{y}_{i,T+1,\text{FE}}$ are the individual
and FE forecasts in~\eqref{eq:indForeMC} and (\ref{ForFESup}) with the
weight 
\begin{equation*}
\hat{\omega}_{\text{FE},NT}^{\ast }=\frac{\hat{\Delta}_{NT}^{\text{FE}}
-T^{-1}\hat{\psi}_{NT}^{\text{FE}} - (\hat{c}^{\text{FE}}_{NT} -\hat{c}_{NT,\beta})}
{\hat{\Delta}_{NT}^{\text{FE}}+T^{-1}\hat{%
h}_{NT,\beta}-2 T^{-1}\hat{\psi}_{NT}^{\text{FE}}}.  \label{wComIFE2}
\end{equation*}%
$\hat{\Delta}_{NT}^{\text{FE}}$, $\hat{h}_{NT,\beta }$ and 
$\hat{\psi}_{NT}^{\text{FE}}$, are given by (\ref{DeltaFE}), 
(\ref{hhatFE}), and (\ref{psihatFE}), and $\hat{c}^{\text{FE}}_{NT}$ and
$\hat{c}_{NT,\beta}$ are defined below equation~\eqref{psihatFE} in the paper.

\item[Combination with individual weights]
\begin{equation*}
\hat{y}_{i,T+1}^{c}=\hat{\omega}_{i}^{\ast }\hat{y}_{i,T+1}+(1-\hat{\omega}%
_{i}^{\ast })\tilde{y}_{i,T+1},  \label{ComIP}
\end{equation*}%
where $\hat{y}_{i,T+1}$ and $\tilde{y}_{i,T+1}$ are the individual and
pooled forecasts in~\eqref{eq:indForeMC} and~\eqref{eq:poolForeMC} with
weights 
\begin{equation*}
\omega _{i}^{\ast} = \frac{ \boldsymbol{w}_{i,T+1}^\prime \hat{\boldsymbol{%
\Omega}}_\eta\boldsymbol{w}_{i,T+1} }{ \boldsymbol{w}_{i,T+1}^\prime(T^{-1}%
\hat\sigma^2_i\boldsymbol{Q}_{iT}^{-1} + \hat{\boldsymbol{\Omega}}_\eta)%
\boldsymbol{w}_{i,T+1} }
\end{equation*}
where $\hat{\boldsymbol{\Omega}}_\eta = N^{-1}\sum_{i=1}^N (\hat{\boldsymbol{%
\theta}}_{i,T}-\bar{\hat{\boldsymbol{\theta}}})(\hat{\boldsymbol{\theta}}_{i,T}-\bar{%
\hat{\boldsymbol{\theta}}})^\prime$, $\bar{\hat{\boldsymbol{\theta}}}%
=N^{-1}\sum_{i=1}^N \hat{\boldsymbol{\theta}}_i$, and the estimator of $%
\hat\sigma^2_i$ is given in Pesaran et al.~(2022).

\item[Empirical Bayes forecast] 

The empirical Bayes forecast using the estimator of Hsiao et al.~(1999) is $
\hat{y}_{i,T+1}^{EB}=\hat{\boldsymbol{\theta }}_{i,\mathit{EB}}^{\prime }%
\boldsymbol{w}_{i,T+1},\label{ForEB}$ 
where 
\begin{equation*}
\hat{\boldsymbol{\theta }}_{i,\mathit{EB}}^{\prime }=(\hat{\sigma}_{i}^{-2}%
\boldsymbol{W}_{i}^{\prime }\boldsymbol{W}_{i}+\boldsymbol{\hat{\Omega}}%
_{\eta }^{-1})^{-1}(\hat{\sigma}_{i}^{-2}\boldsymbol{W}_{i}^{\prime }%
\boldsymbol{y}_{i}+\boldsymbol{\hat{\Omega}}_{\eta }^{-1}\text{ }\bar{\hat{%
\boldsymbol{\theta }}}),
\end{equation*}%
\begin{equation*}
\bar{\hat{\boldsymbol{\theta }}}=\frac{1}{N}\sum_{i=1}^{N}\hat{\boldsymbol{%
\theta }}_{i,T},\quad \hat{\sigma}_{i}^{2}=\widehat{\boldsymbol{\varepsilon }%
}_{i}^{\prime }\widehat{\boldsymbol{\varepsilon }}_{i}/(T-K),
\end{equation*}%
$\boldsymbol{\hat{\Omega}}_{\eta }$ is given above,
and $\widehat{\boldsymbol{\varepsilon }}=\boldsymbol{y}_{i}-\boldsymbol{W}%
_{i}\hat{\boldsymbol{\theta }}_{i,T}$ with $\hat{\boldsymbol{\theta }}_{i,T}$
given in~\eqref{eq:indbeta}.

\item[Hierarchical Bayesian forecast] 

In this supplement, we additionally apply the hierarchical Bayesian model of
Lindley and Smith (1972) which assumes $\varepsilon _{it}\sim iid\mathrm{N}%
(0,\sigma ^{2})$, using the following priors: 
\begin{eqnarray*}
\boldsymbol{\theta }_{i} &\sim &\mathrm{N}(\bar{\boldsymbol{\theta }},%
\boldsymbol{\Sigma }_{\boldsymbol{\theta }}), \\
\bar{\boldsymbol{\theta }} &\sim &\mathrm{N}(\boldsymbol{d},{\boldsymbol{S}}%
_{\bar{\theta}}), \\
\boldsymbol{\Sigma }_{\boldsymbol{\theta }}^{-1} &\sim &\mathrm{Wishart}(\nu
_{\Sigma },(\nu _{\Sigma }\boldsymbol{S}_{\Sigma })^{-1}), \\
\sigma ^{2} &\sim &\mathrm{invGamma}(\nu _{\sigma }/2,\nu _{\sigma }s^{2}/2).
\end{eqnarray*}

The Gibbs sampler uses the conditional posteriors (Gelfand et al., 1990) as
set out below, where $|\cdot $ denotes conditional on the other parameters
in the Gibbs sampler, for $r_{b}=1,2,\ldots ,R_{b}$, where $R_{b}$ denotes
the number of random draws used in the Gibbs sampler:

\begin{itemize}
\item $\boldsymbol{\theta }_{i,r_{b}}|\cdot \sim \mathrm{N(\boldsymbol{b}%
_{i},\boldsymbol{S}_{i})}$, where $\boldsymbol{b}_{i}=\boldsymbol{S}%
_{i}\left( \sigma _{r_{b}-1}^{-2}\boldsymbol{W}_{i}^{\prime }\boldsymbol{y}%
_{i}+\boldsymbol{\Sigma }_{\boldsymbol{\theta },\text{ }r_{b}-1}^{-1}\bar{%
\boldsymbol{\theta }}_{r_{b}-1}\right) $,\newline
and $\boldsymbol{S}_{i}=\left( \sigma _{r_{b}-1}^{-2}\boldsymbol{W}%
_{i}^{\prime }\boldsymbol{W}_{i}+\boldsymbol{\Sigma }_{\boldsymbol{\theta },%
\text{ }r_{b}-1}^{-1}\right) ^{-1} $

\item $\sigma _{r_{b}}^{2}|\cdot \sim \mathrm{invGamma}\left( [NT+\nu
_{\sigma }]/2,\frac{1}{2}\left[ \sum_{i=1}^{N}(\boldsymbol{y}_{i}-%
\boldsymbol{W}_{i}\boldsymbol{\theta }_{i,r_{b}})^{\prime }(\boldsymbol{y}%
_{i}-\boldsymbol{W}_{i}\boldsymbol{\theta }_{i,r_{b}})+\nu _{\sigma }s^{2}%
\right] \right) $

\item $\bar{\boldsymbol{\theta }}_{r_{b}}|\cdot \sim \mathrm{N}(\boldsymbol{h%
},\boldsymbol{S}_{h}),$ where $\boldsymbol{h}=\boldsymbol{S}_{h}\left( 
\boldsymbol{\Sigma }_{\boldsymbol{\theta },r_{b}-1}^{-1}\sum_{i=1}^{N}%
\boldsymbol{\theta }_{i,r_{b}}+{\boldsymbol{S}}_{\bar{\theta}}^{-1}{%
\boldsymbol{d}}\right) $ and $\boldsymbol{S}_{h}=\left( N\boldsymbol{\Sigma }%
_{\boldsymbol{\theta },r_{b}-1}^{-1}+{\boldsymbol{S}}_{\bar{\theta}%
}^{-1}\right) ^{-1}$

\item $\boldsymbol{\Sigma }_{\boldsymbol{\theta },r_{b}}^{-1}|\cdot \sim 
\mathrm{Wishart}\left( N+\nu _{\Sigma },\text{ }\left[ \sum_{i=1}^{N}\left( 
\boldsymbol{\theta }_{i,r_{b}}-\bar{\boldsymbol{\theta }}_{r_{b}}\right)
\left( \boldsymbol{\theta }_{i,r_{b}}-\bar{\boldsymbol{\theta }}%
_{r_{b}}\right) ^{\prime }+\nu _{\Sigma }\boldsymbol{S}_{\Sigma }\right]
^{-1}\right) $
\end{itemize}

The Gibbs sampler draws iteratively from the conditional posterior
distributions, starting with the following initial values ($r_b=0$) 
\begin{equation*}
\sigma _{0}^{2}=\widehat{\boldsymbol{\varepsilon }}^{\prime }\widehat{%
\boldsymbol{\varepsilon }}/(NT-K),\quad \widehat{\boldsymbol{\varepsilon }}=(%
\widehat{\boldsymbol{\varepsilon }}_{1},\widehat{\boldsymbol{\varepsilon }}%
_{2},\ldots ,\widehat{\boldsymbol{\varepsilon }}_{N})^{\prime },\quad 
\widehat{\boldsymbol{\varepsilon }}_{i}=\boldsymbol{y}_{i}-\boldsymbol{W}_{i}%
\hat{\boldsymbol{\theta }}_{i,T}
\end{equation*}%
\begin{equation*}
\bar{\boldsymbol{\theta }}_{0}=\frac{1}{N}\sum_{i=1}^{N}\hat{\boldsymbol{%
\theta }}_{i,T},\quad \text{ and }\quad \boldsymbol{\Sigma }_{\boldsymbol{%
\theta },0}^{-1}=\frac{1}{N}\sum_{i=1}^{N}(\hat{\boldsymbol{\theta }}_{i,T}-%
\bar{\boldsymbol{\theta }}_{0})(\hat{\boldsymbol{\theta }}_{i,T}-\bar{%
\boldsymbol{\theta }}_{0})^{\prime }.
\end{equation*}%
Estimates from the Gibbs sampler are obtained from 1500 iterations with the
first 500 discarded as a burn-in sample. In each iteration, we calculate 
\begin{equation}
\hat{y}_{i,T+1,r_{b}}^{HB}=\hat{\boldsymbol{\theta }}_{i,\mathit{r}%
_{b}}^{\prime }\boldsymbol{w}_{i,T+1},  \label{ForHB}
\end{equation}%
for $i=1,2,\ldots ,N$ and the forecast is then $\hat{y}_{i,T+1}^{HB}=\frac{1%
}{R_{b}}\sum_{r_{b}=1}^{R_{b}}\hat{y}_{i,T+1,r_{b}}^{HB}$.

We use the following hyperpriors: $\boldsymbol{d}=\boldsymbol{0}$, $\nu
_{\Sigma }=K$, $\nu _{\sigma }=0.1$, and $s^{2}=0.1$. For the prior
covariance matrices $\boldsymbol{S}_{\bar{\theta}}$ and $\boldsymbol{S}%
_{\Sigma}$ we provide the results for three settings: 
(1) $\boldsymbol{S}_{%
\bar{\theta}}=\boldsymbol{I}_{K}10^{6}$, $\boldsymbol{S}_{\Sigma }=\boldsymbol{I}_{K}10$, 
(2) $\boldsymbol{S}_{\bar{\theta}} = \boldsymbol{I}%
_{K}10^{2}$, $\boldsymbol{S}_{\Sigma }=\boldsymbol{I}_{K}10^{2}$, and 
(3) 
$\boldsymbol{S}_{\bar{\theta}} = \boldsymbol{I}_{K}$, $\boldsymbol{S}_{\Sigma
}=\boldsymbol{I}_{K}$. These are proper, weakly informative priors that
avoid the use of uninformative priors that appear to be difficult to attain
in hierarchical models (Gelman, 2006).
\end{description}

Monte Carlo results for the hierarchical Bayesian model are given in Table~%
\ref{tbl:MC_hierBayes}. Since the MCMC approach to the hierarchical Bayesian
model is computationally quite expensive, we restrict the Monte Carlo
experiments to 1000 iterations and report some of the remaining methods as a
reference. It can be seen from the table that the accuracy of the forecasts
largely depends on the serendipitous choice of the prior.

\input{./tables/MC_hierBayes_appendix_April2025.tex}

Results for the applications are reported in Table~\ref%
{tbl:applications_msfe_appendix}. The results suggest that the choice of
prior for the error variance has relatively little influence, whereas the
prior choices for the parameter covariances can substantially alter the
forecast accuracy.

\section{Additional Monte Carlo applications and empirical results}

\label{addMC}

In Section~\ref{sec:MC} of the paper, we restricted our analysis to the case of $N=100$.
The results for $N=1000$ can be found in Table~\ref{tbl:MC_ARX_N1000}. The
results from $N=100$ clearly carry over and the influence of the number of
cross-section units is small.

\input{./tables/MC_N1000_appendix_April2025.tex}

As a practical alternative to the combination forecasts in Section~\ref%
{sec:combinations}, which are based on estimates of the optimal combination
weights, forecast combinations using equal weights have a long history in
the literature (Timmermann, 2006). We therefore considered how this forecast
combination scheme performs both in the Monte Carlo simulations and for the
empirical applications. As in the paper, we separately consider combination
schemes for the individual-pooled forecasts and for the individual-FE
forecasts.

\input{./tables/MC_appendix.tex}

In Table S.4 we also report
a complete suite of Monte Carlo simulation results based on an
equal-weighted combination scheme for our two combination schemes. The
predictive accuracy of the equal-weighted combination scheme is comparable
to that of the combinations based on estimated weights in the presence of
modest levels of parameter heterogeneity. Conversely, equal-weighted
combinations underperform forecast combinations with estimated weights when
the level of parameter heterogeneity is either very low or very high. In
either case, one approach (individual estimation or pooling) dominates the
other by a sufficiently large margin that equal-weighting becomes
sub-optimal.

We also considered the performance of an (infeasible) oracle combination
scheme that uses the true parameter values to compute the optimal
combination weights. Compared against our feasible estimates of the
combination weights, this oracle scheme shows the impact of parameter
estimation error on forecasting performance. We find that the cost of
estimation error is only sizable if $T$ is small ($T=20$) and the parameters
are homogeneous. For this case, the oracle scheme reduces the MSFE of the
pooled-individual combination by 0.051 (0.906 versus 0.856) and by 0.037 for
the FE-individual combination. Differences are much smaller (0.005 and
0.011) in the heterogeneous case even when $T=20$\ and are further reduced
for $T=100$ where, in many cases, only the third decimal of the MSFE ratio
is affected.

Overall, we conclude from these Monte Carlo simulations that the optimal
forecast combination scheme introduced in our paper produces more accurate
forecasts that are notably more robust to parameter heterogeneity than the
equal-weighted combination schemes considered here.

Table~\ref{tbl:applications_msfe_appendix} shows the performance of the
equal-weighted forecasts for the application to house price inflation. For
comparison, we also show the forecasting results for our optimal combination
scheme. In this application pooling beats individual forecasts, which
suggests a low degree of parameter heterogeneity. The equal-weighted
forecast combinations perform correspondingly well. In fact, the combination
of individual and pooled forecasts has the lowest average MSFE, offers the
most precise forecasts for 10.2\% (SAR model) and 14.9\% (SARX) of MSAs and
never produces the worst forecast. This performance is marginally better
than that of the optimal combination schemes with estimated weights.

\input{./tables/applications_appendix_April_2025.tex}

The results for the CPI application in Table~\ref%
{tbl:applications_msfe_appendix} show that in a similar fashion the
equal-weighted combination provides precise forecasts, which are more
accurate, on average, than the optimal forecast combination, though beaten
by a small margin by the empirical Bayes forecasts.

Table~\ref{tbl:DM_appendix} shows the results from the panel and individual
DM test statistics. For both applications, the panel DM test show
significant improvements over the individual forecasts. For the house price
applications, somewhat fewer forecasts for MSAs are significantly better
than the individual forecast compared to what we find for the optimal
combination scheme. For the CPI application, in contrast, the pooled
forecast with equal weights is significantly more precise than the benchmark
for slightly more series than under the optimal combination scheme.

\input{./tables/DMtestStats_Appendix_7April25.tex}

\section{Derivation of results for fixed effect estimation}

\label{app:FE_estimator}

Following the derivations for the pooled estimates, it is easily seen that%
\begin{equation*}
\hat{\boldsymbol{\beta }}_{\text{FE}}-\boldsymbol{\beta }_{i}=-\boldsymbol{%
\eta }_{i,\beta }+\bar{\boldsymbol{Q}}_{NT,\beta }^{-1}\bar{\boldsymbol{q}}%
_{NT,\beta }+\bar{\boldsymbol{Q}}_{NT,\beta }^{-1}\bar{\boldsymbol{\xi }}%
_{NT,\beta },
\end{equation*}%
where $\boldsymbol{\eta }_{i,\beta }=\boldsymbol{\beta }_{i}-\boldsymbol{%
\beta },$ $\bar{\boldsymbol{\xi }}_{NT,\beta }=N^{-1}\sum_{i=1}^{N}T^{-1}%
\boldsymbol{X}_{i}^{\prime }\boldsymbol{M}_{T}\boldsymbol{\varepsilon }_{i}$,%
\begin{equation*}
\ \bar{\boldsymbol{Q}}_{NT,\beta }=N^{-1}\sum_{i=1}^{N}T^{-1}\boldsymbol{X}%
_{i}^{\prime }\boldsymbol{M}_{T}\boldsymbol{X}_{i},\text{ and }\bar{%
\boldsymbol{q}}_{NT,\beta }=N^{-1}\sum_{i=1}^{N}\left( T^{-1}\boldsymbol{X}%
_{i}^{\prime }\boldsymbol{M}_{T}\boldsymbol{X}_{i}\right) \boldsymbol{\eta }%
_{i,\beta }\text{.}
\end{equation*}%
With one exception, the derivation of the average MSFE for the FE estimation
closely parallels the case of the pooled estimator with $\boldsymbol{\eta }%
_{i,\beta }$ in place of $\boldsymbol{\eta }_{i}$, $\bar{\boldsymbol{Q}}%
_{NT,\beta }$ replacing $\bar{\boldsymbol{Q}}_{NT}$, $\bar{\boldsymbol{q}}%
_{NT,\beta }$ replacing $\bar{\boldsymbol{q}}_{NT}$, $\bar{\boldsymbol{\xi }}%
_{NT,\beta }$ replacing $\bar{\boldsymbol{\xi }}_{NT}$, and $\bar{\bar{%
\boldsymbol{x}}}_{i,T+1}=\boldsymbol{x}_{i,T+1}-\bar{\boldsymbol{x}}_{iT}$
in place of $\boldsymbol{x}_{i,T+1}$. The exception arises due to the fact
that in the case of weakly exogenous regressors, $\bar{\varepsilon}_{iT}$
(and hence $\bar{\bar{\varepsilon}}_{i,T+1}$) is not distributed
independently of $(\hat{\boldsymbol{\beta }}_{\text{FE}}-\boldsymbol{\beta }%
_{i})^{\prime }\bar{\bar{\boldsymbol{x}}}_{i,T+1}.$ To account for this
dependence, we first note that, under Assumption \ref{ass:2.2}, $\bar{%
\boldsymbol{\xi }}_{NT,\beta }=O_{p}\left( N^{-1/2}T^{-1/2}\right) $, and%
\begin{eqnarray*}
&&N^{-1}\sum_{i=1}^{N}\left( \hat{\boldsymbol{\beta }}_{\text{FE}}-%
\boldsymbol{\beta }_{i}\right) ^{\prime }\bar{\bar{\boldsymbol{x}}}_{i,T+1}%
\bar{\varepsilon}_{iT}=N^{-1}\sum_{i=1}^{N}\left( -\boldsymbol{\eta }%
_{i,\beta }+\bar{\boldsymbol{Q}}_{NT,\beta }^{-1}\bar{\boldsymbol{q}}%
_{NT,\beta }+\bar{\boldsymbol{Q}}_{NT,\beta }^{-1}\bar{\boldsymbol{\xi }}%
_{NT,\beta }\right) ^{\prime }\bar{\bar{\boldsymbol{x}}}_{i,T+1}\bar{%
\varepsilon}_{iT} \\
&=&-N^{-1}\sum_{i=1}^{N}\boldsymbol{\eta }_{i,\beta }^{\prime }\bar{\bar{%
\boldsymbol{x}}}_{i,T+1}\bar{\varepsilon}_{iT}+\bar{\boldsymbol{q}}%
_{NT,\beta }^{\prime }\bar{\boldsymbol{Q}}_{NT,\beta }^{-1}\left(
N^{-1}\sum_{i=1}^{N}\bar{\bar{\boldsymbol{x}}}_{i,T+1}\bar{\varepsilon}%
_{iT}\right) +O_{p}\left( N^{-1/2}\right) .
\end{eqnarray*}%
Also, under Assumptions \ref{ass:weak_exogeneity_b} and \ref{ass:6} we have%
\begin{equation}
N^{-1}\sum_{i=1}^{N}\left( \hat{\boldsymbol{\beta }}_{\text{FE}}-%
\boldsymbol{\beta }_{i}\right) ^{\prime }\bar{\bar{\boldsymbol{x}}}_{i,T+1}%
\bar{\varepsilon}_{iT}=c_{NT}^{\text{FE}}+O_{p}(N^{-1/2}),  \label{biasFE}
\end{equation}%
where $c_{NT}^{\text{FE}}$ is given in Section~B.4 of the paper.

The expression for $c_{NT}^{\text{FE}}$ simplifies somewhat by noting that
under Assumption \ref{ass:weak_exogeneity_a}, $\mathrm{E}\left( \boldsymbol{x%
}_{iT+1}\bar{\varepsilon}_{iT}\right) =\boldsymbol{0}$, and using Lemma \ref%
{Lemma_1_VTEX1} we have $\bar{\boldsymbol{q}}_{NT,\beta }^{\prime }\bar{%
\boldsymbol{Q}}_{NT,\beta }^{-1}=\bar{\boldsymbol{q}}_{N,\beta }^{\prime }%
\bar{\boldsymbol{Q}}_{N,\beta }^{-1}+O_{p}\left( N^{-1/2}\right) $. Note
that under Assumption \ref{ass:5}, $\boldsymbol{\eta }_{i,\beta }$ and $%
\varepsilon _{it}$ are independently distributed. Using these results, the
MSFE under fixed effects estimation in~(\ref{FEmsfe}) follows.

\subsubsection*{A comparison of forecasts based on individual and fixed
effects estimates}

Since,
\begin{equation}
\hat{e}_{i,T+1}=\bar{\bar{\varepsilon}}_{i,T+1}-\bar{\bar{\boldsymbol{x}}}%
_{i,T+1}^{\prime }(\boldsymbol{\hat{\beta}}_{i}-\boldsymbol{\beta }_{i}).
\label{ei}
\end{equation}%
The derivation of the average MSFE, $N^{-1}\sum_{i=1}^{N}\hat{e}_{i,T+1}^{2}$
can now proceed as before, except that under weak exogeneity the two
components of $\hat{e}_{i,T+1},$ in (\ref{ei}), are no longer independently
distributed and, as in the FE estimation, we need to consider the additional
term%
\begin{eqnarray*}
&&N^{-1}\sum_{i=1}^{N}\bar{\bar{\boldsymbol{x}}}_{i,T+1}^{\prime }(%
\boldsymbol{\hat{\beta}}_{i}-\boldsymbol{\beta }_{i})\bar{\bar{\varepsilon}}%
_{i,T+1}=N^{-1}\sum_{i=1}^{N}\bar{\bar{\boldsymbol{x}}}_{i,T+1}^{\prime }(%
\boldsymbol{X}_{i}^{\prime }\boldsymbol{M}_{T}\boldsymbol{X}_{i})^{-1}%
\boldsymbol{X}_{i}^{\prime }\boldsymbol{M}_{T}\boldsymbol{\varepsilon }_{i}%
\bar{\bar{\varepsilon}}_{i,T+1} \\
&=&-N^{-1}\sum_{i=1}^{N}\mathrm E\left[ \bar{\bar{\boldsymbol{x}}}_{i,T+1}^{\prime }(%
\boldsymbol{X}_{i}^{\prime }\boldsymbol{M}_{T}\boldsymbol{X}_{i})^{-1}%
\boldsymbol{X}_{i}^{\prime }\boldsymbol{M}_{T}\boldsymbol{\varepsilon }_{i}%
\bar{\varepsilon}_{i,T+1}\right]+O_{p}(N^{-1/2}).
\end{eqnarray*}%
Using this, we have%
\begin{equation}
N^{-1}\sum_{i=1}^{N}\bar{\bar{\boldsymbol{x}}}_{i,T+1}^{\prime }(\boldsymbol{%
\hat{\beta}}_{i}-\boldsymbol{\beta }_{i})\bar{\bar{\varepsilon}}%
_{i,T+1}=c_{NT,\beta }+O_{p}(N^{-1/2}),  \label{biasIND}
\end{equation}%
where%
\begin{equation}
c_{NT,\beta }=N^{-1}\sum_{i=1}^{N}\mathrm{E}\left[ \bar{\bar{\boldsymbol{x}}}%
_{i,T+1}^{\prime }(\boldsymbol{X}_{i}^{\prime }\boldsymbol{M}_{T}\boldsymbol{%
X}_{i})^{-1}\boldsymbol{X}_{i}^{\prime }\boldsymbol{M}_{T}\boldsymbol{%
\varepsilon }_{i}\bar{\varepsilon}_{iT}\right] .  \label{cNTbeta_app}
\end{equation}%
Taking this term into account we obtain%
\begin{equation}
N^{-1}\sum_{i=1}^{N}\hat{e}_{i,T+1}^{2}=N^{-1}\sum_{i=1}^{N}\bar{\bar{%
\varepsilon}}_{i,T+1}^{2}+T^{-1}h_{NT,\beta }-2c_{NT,\beta }+O_{p}(N^{-1/2}),
\end{equation}%
where%
\begin{equation}
h_{NT,\beta }=N^{-1}\sum_{i=1}^{N}\mathrm{E}\left[ \bar{\bar{\boldsymbol{x}}}%
_{i,T+1}^{\prime }\boldsymbol{Q}_{iT,\beta }^{-1}\left( \frac{\boldsymbol{X}%
_{i}^{\prime }\boldsymbol{M}_{T}\boldsymbol{\varepsilon }_{i}\boldsymbol{%
\varepsilon }_{i}^{\prime }\boldsymbol{M}_{T}\boldsymbol{X}_{i}}{T}\right) 
\boldsymbol{Q}_{iT,\beta }^{-1}\bar{\bar{\boldsymbol{x}}}_{i,T+1}\right] ,
\label{eq:hNT_FEapp}
\end{equation}%
and $\boldsymbol{Q}_{iT,\beta }=T^{-1}\left( \boldsymbol{X}_{i}^{\prime }%
\boldsymbol{M}_{T}\boldsymbol{X}_{i}\right) $. As with the term $c_{NT}^{%
\text{FE}}$ in the average MSFE of the FE forecasts, $c_{NT,\beta }=0$ when $%
\boldsymbol{x}_{it}\,$is strictly exogenous. To see why this is so, note
that in this case, $\mathrm{E}\left( \boldsymbol{\varepsilon }_{i}\bar{%
\varepsilon}_{iT}\left\vert \boldsymbol{X}_{i}\right. \right) =(\sigma
_{i}^{2}/T)\boldsymbol{\tau }_{T}$ and%
\begin{equation*}
\mathrm{E}\left[ \bar{\bar{\boldsymbol{x}}}_{i,T+1}^{\prime }(\boldsymbol{X}%
_{i}^{\prime }\boldsymbol{M}_{T}\boldsymbol{X}_{i})^{-1}\boldsymbol{X}%
_{i}^{\prime }\boldsymbol{M}_{T}\boldsymbol{\varepsilon }_{i}\bar{\varepsilon%
}_{iT}\left\vert \boldsymbol{X}_{i}\right. \right] =\bar{\bar{\boldsymbol{x}}%
}_{i,T+1}^{\prime }(\boldsymbol{X}_{i}^{\prime }\boldsymbol{M}_{T}%
\boldsymbol{X}_{i})^{-1}\boldsymbol{X}_{i}^{\prime }\boldsymbol{M}_{T}%
\mathrm{E}\left[ \boldsymbol{\varepsilon }_{i}\bar{\varepsilon}%
_{iT}\left\vert \boldsymbol{X}_{i},\bar{\bar{\boldsymbol{x}}}_{i,T+1}\right. %
\right] =0,
\end{equation*}%
so unconditionally $\mathrm{E}\left[ \bar{\bar{\boldsymbol{x}}}%
_{i,T+1}^{\prime }(\boldsymbol{X}_{i}^{\prime }\boldsymbol{M}_{T}\boldsymbol{%
X}_{i})^{-1}\boldsymbol{X}_{i}^{\prime }\boldsymbol{M}_{T}\boldsymbol{%
\varepsilon }_{i}\bar{\varepsilon}_{iT}\right] =0$, and $c_{NT,\beta }=0$.

Apart from the error term, $\varepsilon _{i,T+1}-\bar{\varepsilon}_{iT}$,
which is common to the individual and FE forecasts, the squared forecast
errors are analogous to those in the comparison of individual and pooled
forecasts except that we work with demeaned data and allow for the
additional terms $c_{NT}^{\text{FE}}$ and $c_{NT,\beta }$ if the regressors
are weakly exogenous.

{\ifx \undefined \bysame \fi}

\end{document}

%% file: tables/MC_N100_mainBody_April2025.tex
\begin{sidewaystable}\thisfloatpagestyle{empty}
\caption{Monte Carlo results}
\label{tbl:MC_ARX_N100} 
\hspace{-3em} 
\resizebox{23cm}{!}{
\setlength\tabcolsep{3pt}
\begin{tabular}{lllllllllllllllllllllllllllllll}
\hline\hline
$a_\beta$&$\sigma^2_\alpha$  &&
\multicolumn{3}{c}{Pooled}&&\multicolumn{3}{c}{RE}&&\multicolumn{3}{c}{FE} &&
\multicolumn{3}{c}{Empirical Bayes}&&
\multicolumn{3}{c}{Comb.\ (pool)} &&
\multicolumn{3}{c}{Comb.\ (FE)}&& 
\multicolumn{3}{c}{Comb.\ $\omega_i^*$}\\
\cline{1-2}\cline{4-6}\cline{8-10}\cline{12-14}\cline{16-18}\cline{20-22}\cline{24-26}\cline{28-30}
&$T$&&\multicolumn{1}{c}{20}&\multicolumn{1}{c}{50}&\multicolumn{1}{c}{100}&&\multicolumn{1}{c}{20}&\multicolumn{1}{c}{50}&\multicolumn{1}{c}{100}&&\multicolumn{1}{c}{20}&\multicolumn{1}{c}{50}&\multicolumn{1}{c}{100}&&\multicolumn{1}{c}{20}&\multicolumn{1}{c}{50}&\multicolumn{1}{c}{100}&&\multicolumn{1}{c}{20}&\multicolumn{1}{c}{50}&\multicolumn{1}{c}{100}&&\multicolumn{1}{c}{20}&\multicolumn{1}{c}{50}&\multicolumn{1}{c}{100}&&\multicolumn{1}{c}{20}&\multicolumn{1}{c}{50}&\multicolumn{1}{c}{100}\\
\hline
\rowsep{-0.975}\\
\multicolumn{30}{c}{Conditional on $\kappa_i=0$}\\
\hline
\multicolumn{30}{c}{$\rho_{\gamma x}=0$}\\
0.0&0.5&&0.864&0.985&1.010&&0.911&0.985&0.996&&0.923&0.987&0.997&&0.935&0.989&0.997&&0.913&0.982&0.995&&0.955&0.992&0.998&&0.947&0.994&0.999\\
0.5&0.5&&0.860&0.981&1.006&&0.953&1.004&1.005&&0.978&1.009&1.007&&0.944&0.992&0.998&&0.911&0.981&0.995&&0.979&0.999&1.000&&0.939&0.994&0.999\\
1.0&1.0&&0.819&0.964&0.994&&1.089&1.096&1.061&&1.167&1.119&1.068&&0.937&0.990&0.998&&0.895&0.975&0.992&&1.022&1.006&1.000&&0.900&0.986&0.998%
\rowsep{0.1}\\
\multicolumn{30}{c}{$\rho_{\gamma x}=0.5$}\\
0.0&0.5&&0.862&0.982&1.007&&0.910&0.985&0.996&&0.923&0.987&0.997&&0.937&0.989&0.997&&0.912&0.981&0.995&&0.955&0.992&0.998&&0.950&0.994&0.999\\
0.5&0.5&&0.856&0.977&1.001&&0.950&1.003&1.005&&0.977&1.009&1.007&&0.946&0.993&0.998&&0.910&0.980&0.994&&0.979&0.999&1.000&&0.942&0.994&0.999\\
1.0&1.0&&0.820&0.964&0.994&&1.098&1.103&1.065&&1.174&1.125&1.073&&0.941&0.991&0.998&&0.895&0.975&0.992&&1.024&1.006&1.000&&0.900&0.986&0.998%
\rowsep{0.25}\\
\multicolumn{30}{c}{Conditional on $\kappa_i=\pm 1$}\\
\hline
\multicolumn{30}{c}{$\rho_{\gamma x}=0$}\\
0.0&0.5&&0.744&0.997&1.065&&0.733&0.924&0.971&&0.752&0.928&0.972&&0.804&0.945&0.979&&0.806&0.948&0.985&&0.842&0.951&0.980&&0.800&0.957&0.989\\
0.5&0.5&&0.867&1.163&1.243&&0.910&1.113&1.162&&0.949&1.122&1.164&&0.849&0.969&0.992&&0.832&0.964&0.991&&0.913&0.985&0.996&&0.808&0.965&0.992\\
1.0&1.0&&1.049&1.455&1.572&&1.274&1.504&1.529&&1.388&1.540&1.540&&0.881&0.977&0.995&&0.855&0.975&0.995&&0.983&0.999&1.000&&0.786&0.963&0.993\rowsep{0.1}\\
\multicolumn{30}{c}{$\rho_{\gamma x}=0.5$}\\
0.0&0.5&&0.764&1.023&1.093&&0.731&0.923&0.971&&0.752&0.928&0.972&&0.803&0.945&0.979&&0.809&0.950&0.986&&0.842&0.951&0.980&&0.801&0.958&0.989\\
0.5&0.5&&0.894&1.196&1.278&&0.912&1.120&1.168&&0.955&1.130&1.171&&0.846&0.968&0.991&&0.836&0.966&0.992&&0.914&0.986&0.997&&0.809&0.965&0.992\\
1.0&1.0&&1.068&1.479&1.597&&1.287&1.515&1.535&&1.401&1.552&1.547&&0.878&0.974&0.993&&0.856&0.975&0.995&&0.986&0.999&1.000&&0.785&0.960&0.992\\
\hline\hline
\multicolumn{30}{p{25.5cm}}{\footnotesize{Notes:
The table reports the ratio of average MSFE for a given forecasting method over the average MSFE of the forecasts based on 
individual estimates.
The forecasts are:
`Pooled' based on pooled estimation, `RE' based on the random effects estimation,
`FE' based on the fixed effects estimation, `Empirical Bayes' based on the empirical Bayes estimation,
`Comb.~(pool)' refers to the combination
of forecasts based on individual and pooled estimation, 
`Comb.~(FE)' the combination of forecasts based on
 individual and fixed effects estimation, 
 and `Comb.\ $\omega_i^*$' the combination forecasts using the individual weights of Pesaran et al.~(2022).
  The parameters
$a_\beta$ and $\sigma_\alpha^2$ determine the heterogeneity of the slope coefficient and the intercept.
The results in the upper panel are for $\kappa_i = 0$ where $\bs w_{i,T+1}$ equals the expected values of the regressors.
The results in the lower panel are for $\kappa_i=\pm 1$ where $\bs w_{i,T+1}$ equals the expected values of the regressors plus or minus one standard deviation.
Results are for \textit{PR}$^2$ of approximately 0.6 and  $N=100$.
The DGP is set out in Section~\ref{sec:MCDesignARX}.}}
\end{tabular}
}
\end{sidewaystable}  

%% file: tables/applications_April_2025.tex
\begin{sidewaystable}\thisfloatpagestyle{empty}
\caption{Results for the applications}
\label{tbl:applications_msfe}
\centering
{\footnotesize
\hspace*{-1.5cm}
\begin{tabular}{llllllllllllllll}
\hline\hline
 & \multicolumn{3}{l}{Ratio of} && \multicolumn{3}{l}{Freq.\ beating} &&\multicolumn{3}{l}{Freq.\ smallest}&&\multicolumn{3}{l}{Freq.\ largest}\\
 & \multicolumn{3}{l}{ave.~MSFE} &&\multicolumn{3}{l}{benchmark} &&\multicolumn{3}{l}{MSFE}&&\multicolumn{3}{l}{MSFE}\\
\cline{2-4}\cline{6-8}\cline{10-12}\cline{14-16}
Observations             & all & $\kappa_i=0$ & $\kappa_i=\pm 1$ && all & $\kappa_i=0$ & $\kappa_i=\pm 1$&& all & $\kappa_i=0$ & $\kappa_i=\pm 1$ && all& $\kappa_i=0$ & $\kappa_i=\pm 1$ \\
\hline
\rowsep{-.5}\\
\multicolumn{16}{l}{House price inflation forecasts}\\
\hline
 Individual                 & 2.822&2.520 & 3.542 && --   & --   & --    && 0.008&0.273 & 0.146 && 0.569 &0.282 & 0.420\rowsep{.2}\\
 Pooled                     & 0.920&1.162 & 0.947 && 0.613&0.381 & 0.536 && 0.116&0.171 & 0.249 && 0.157 &0.188 & 0.160\\
 RE                         & 0.924&1.166 & 0.960 && 0.619&0.376 & 0.528 && 0.108&0.022 & 0.047 && 0.003 &0.019 & 0.008\\
 FE                         & 0.936&1.186 & 0.980 && 0.591&0.381 & 0.517 && 0.055&0.108 & 0.105 && 0.251 &0.376 & 0.296\\
 Emp.Bayes                  & 0.901&0.955 & 0.881 && 0.942&0.519 & 0.652 && 0.185&0.157 & 0.127 && 0.003 &0.064 & 0.052\\
 Comb.\ (pool)              & 0.920&0.961 & 0.932 && 0.939&0.522 & 0.688 && 0.157&0.077 & 0.099 && 0.000 &0.011 & 0.017\\
 Comb.\ (FE)                & 0.937&0.977 & 0.940 && 0.917&0.494 & 0.677 && 0.019&0.072 & 0.069 && 0.011 &0.039 & 0.044\\
 Comb.\ ($\omega^*_i$)      & 0.921&0.957 & 0.909 && 0.936&0.541 & 0.713 && 0.044&0.028 & 0.052 && 0.006 &0.003 & 0.006\rowsep{.5}\\
\multicolumn{16}{l}{CPI inflation forecasts}\\
\hline
 Individual                 &15.501&10.451 &11.295&& --   & --   &  --   && 0.005&0.134 & 0.070 && 0.439&0.214 & 0.316\rowsep{.2}\\
 Pooled                     & 0.878& 1.013& 0.971 && 0.444&0.374 & 0.417 && 0.203&0.118 & 0.112 && 0.396&0.406 & 0.380\\
 RE                         & 0.880& 1.001& 0.957 && 0.508&0.390 & 0.401 && 0.016&0.070 & 0.064 && 0.000&0.037 & 0.032\\
 FE                         & 0.883& 0.992& 0.959 && 0.508&0.401 & 0.401 && 0.000&0.086 & 0.102 && 0.166&0.230 & 0.225\\
 Emp.Bayes                  & 0.892& 0.991& 0.926 && 0.984&0.652 & 0.818 && 0.390&0.278 & 0.278 && 0.000&0.070 & 0.011\\
 Comb.\ (pool)              & 0.930& 0.987& 0.953 && 0.733&0.481 & 0.572 && 0.128&0.091 & 0.123 && 0.000&0.027 & 0.016\\
 Comb.\ (FE)                & 0.935& 0.980& 0.967 && 0.791&0.524 & 0.583 && 0.053&0.064 & 0.070 && 0.000&0.016 & 0.021\\
 Comb.\ ($\omega^*_i$)      & 0.897& 0.972& 0.931 && 0.973&0.695 & 0.813 && 0.203&0.160 & 0.182 && 0.000&0.000 & 0.000\\
 \hline\hline
\multicolumn{16}{p{19.3cm}}{\footnotesize{
 results for the house price application and the bottom panel reports the
 results for the CPI subindices application. The first three columns report
 the ratio of average MSFE of the respective method in the row relative to
 that of the individual forecast. The exception is the individual forecast,
 which reports the average MSFE (times $10^5$ in the case of CPI). The second
 three columns report the proportion of cross-section units for which the
 respective method in the rows have a lower MSFE than the individual
 forecast. The third three columns report the proportion of cross-section
 units for which the respective method in the row has the lowest MSFE. The
 last three columns report the proportion of units for which the respective
 method in the row has the highest MSFE. In each block, the first column
 averages over all forecasts, the second over the forecasts for which $d_
 {i,T+1}=\hat{\boldsymbol{\theta }}_{i}^{\prime }\boldsymbol{w}_{i,T+1}$ is
 close to its mean in the estimation sample. The third column averages over
 the forecast for which $d_{i,T+1}$ is close to plus or minus one standard
 deviation from its mean in the estimation sample. The methods in the rows
 are listed in the footnote of Table~\ref{tbl:MC_ARX_N100}.}}
\end{tabular}
}
\end{sidewaystable}


%% file: tables/DMtestStats_20April25.tex




\begin{table}
\caption{Diebold-Mariano test statistics for equal predictive accuracy}
\label{tbl:DM}\centering
\resizebox{17cm}{!}{
\begin{tabular}{lrrrrrrr}
\hline\hline
 & Pooled & RE & FE  & Emp.Bay.& Comb(pool)& Comb(FE) & Comb($\omega^*_i$) \\
 \hline
\multicolumn{8}{l}{House Prices: all forecasts}\\
\hline
Panel DM   &  $-$9.45 &  $-$9.11  & $-$7.57& $-$24.63 & $-$22.93 & $-$21.19 &$-$27.59\\
$\text{DM}<-1.96$/$\text{DM}>1.96$  &  60/6      &  62/6    &  57/8 & 209/2&  189/0   &  169/0 & 240/0
\rowsep{0.25}\\
\multicolumn{8}{l}{CPI: all forecasts}\\
\hline
Panel DM                           & $-$7.95& $-$7.78& $-$7.56& $-$11.25&$-$11.67& $-$10.59 & $-$11.48\\
$\text{DM}<-1.96$/$\text{DM}>1.96$ & 35/60 & 33/45 & 32/42 & 134/0 & 56/23 & 50/10 & 137/0\\
\hline\hline
\multicolumn{8}{p{17.3cm}}{\footnotesize{Notes:
The row ``Panel DM'' reports the results of the panel version of the
Diebold-Mariano test of Pesaran et al.~(2013). The second row report unit by
unit Diebold-Mariano test results: ``$\text{DM}<-1.96$'' reports the number of units
with a DM test statistic smaller than $-1.96$ and ``$\text{DM}>1.96$'' shows the
number of units whose test statistic exceeds 1.96. The remaining units have
insignificant DM test statistics. In total the house prices panel consist of
362 units and the CPI panel of 187 units. Each test is for the null
hypothesis that the forecasting method in the columns has equal forecast
accuracy as the forecasts based on individual estimates. The forecasting
methods are listed in the footnote of Table~\ref{tbl:MC_ARX_N100}.
}}
\end{tabular}}
\end{table}


%% file: tables/MC_hierBayes_appendix_April2025.tex
\begin{sidewaystable}\thisfloatpagestyle{empty}
\caption{Monte Carlo results including hierarchical Bayesian forecasts}
\label{tbl:MC_hierBayes}
\hspace{-3em} 
\resizebox{25.4cm}{!}{
\setlength\tabcolsep{3pt}
\begin{tabular}{llllllllllllllllllllllllllllllllllllllllll}
\hline\hline
$a_\beta$&$\sigma^2_\alpha$  &&
\multicolumn{3}{c}{Pooled}&&\multicolumn{3}{c}{RE}&&\multicolumn{3}{c}{FE} &&
\multicolumn{3}{c}{Empirical Bayes}&&
\multicolumn{3}{c}{hier.\ Bayes: prior 1} &&
\multicolumn{3}{c}{hier.\ Bayes: prior 2} &&
\multicolumn{3}{c}{hier.\ Bayes: prior 3} &&
\multicolumn{3}{c}{Comb.\ (pool)} &&
\multicolumn{3}{c}{Comb.\ (FE)}\\
\cline{1-2}\cline{4-6}\cline{8-10}\cline{12-14}\cline{16-18}\cline{20-22}\cline{24-26}\cline{28-30}\cline{32-34}\cline{36-38}
&$T$&&\multicolumn{1}{c}{20}&\multicolumn{1}{c}{50}&\multicolumn{1}{c}{100}&&\multicolumn{1}{c}{20}&\multicolumn{1}{c}{50}&\multicolumn{1}{c}{100}&&\multicolumn{1}{c}{20}&\multicolumn{1}{c}{50}&\multicolumn{1}{c}{100}&&\multicolumn{1}{c}{20}&\multicolumn{1}{c}{50}&\multicolumn{1}{c}{100}&&\multicolumn{1}{c}{20}&\multicolumn{1}{c}{50}&\multicolumn{1}{c}{100}&&\multicolumn{1}{c}{20}&\multicolumn{1}{c}{50}&\multicolumn{1}{c}{100}&&\multicolumn{1}{c}{20}&\multicolumn{1}{c}{50}&\multicolumn{1}{c}{100}&&\multicolumn{1}{c}{20}&\multicolumn{1}{c}{50}&\multicolumn{1}{c}{100}&&\multicolumn{1}{c}{20}&\multicolumn{1}{c}{50}&\multicolumn{1}{c}{100}\\
\hline\rowsep{-0.99}\\
\multicolumn{38}{c}{Conditional on $\kappa_i=0$}\\
\hline
\multicolumn{38}{c}{$N=100$, $\rho_{\gamma x}=0$}\\
0.0 & 0.5 && 0.864 & 0.984 & 1.008 && 0.912 & 0.984 & 0.996 && 0.925 & 0.987 & 0.997 && 0.936 & 0.989 & 0.997 && 0.937 & 0.992 & 0.998 && 0.970 & 0.997 & 0.999 && 0.906 & 0.986 & 0.997 && 0.913 & 0.982 & 0.995 && 0.956 & 0.992 & 0.998\\
0.5 & 0.5 && 0.860 & 0.980 & 1.004 && 0.956 & 1.003 & 1.006 && 0.982 & 1.009 & 1.007 && 0.946 & 0.992 & 0.998 && 0.932 & 0.992 & 0.998 && 0.963 & 0.997 & 0.999 && 0.908 & 0.988 & 0.998 && 0.911 & 0.981 & 0.994 && 0.981 & 0.999 & 1.000\\
1.0 & 1.0 && 0.815 & 0.963 & 0.992 && 1.085 & 1.096 & 1.064 && 1.162 & 1.119 & 1.072 && 0.932 & 0.991 & 0.998 && 0.891 & 0.987 & 0.998 && 0.918 & 0.992 & 0.999 && 0.878 & 0.986 & 0.997 && 0.893 & 0.974 & 0.991 && 1.018 & 1.005 & 1.001
\rowsep{0.01}\\
\multicolumn{38}{c}{$N=100$, $\rho_{\gamma x}=0.5$}\\
0.0 & 0.5 && 0.862 & 0.981 & 1.005 && 0.911 & 0.984 & 0.996 && 0.925 & 0.987 & 0.997 && 0.937 & 0.989 & 0.997 && 0.934 & 0.991 & 0.998 && 0.968 & 0.997 & 0.999 && 0.903 & 0.986 & 0.997 && 0.912 & 0.981 & 0.994 && 0.956 & 0.992 & 0.998\\
0.5 & 0.5 && 0.855 & 0.974 & 1.000 && 0.953 & 1.003 & 1.006 && 0.981 & 1.009 & 1.008 && 0.947 & 0.992 & 0.998 && 0.927 & 0.991 & 0.998 && 0.959 & 0.996 & 0.999 && 0.902 & 0.987 & 0.997 && 0.909 & 0.979 & 0.994 && 0.980 & 0.999 & 1.000\\
1.0 & 1.0 && 0.818 & 0.962 & 0.992 && 1.098 & 1.104 & 1.068 && 1.174 & 1.127 & 1.076 && 0.940 & 0.991 & 0.998 && 0.897 & 0.987 & 0.997 && 0.924 & 0.992 & 0.999 && 0.884 & 0.985 & 0.997 && 0.894 & 0.974 & 0.991 && 1.022 & 1.006 & 1.001
\rowsep{0.01}\\
\multicolumn{38}{c}{$N=1000$, $\rho_{\gamma x}=0$}\\
0.0 & 0.5 && 0.857 & 0.981 & 1.006 && 0.904 & 0.984 & 0.996 && 0.918 & 0.987 & 0.996 && 0.934 & 0.989 & 0.997 && 0.901 & 0.986 & 0.997 && 0.932 & 0.992 & 0.998 && 0.879 & 0.983 & 0.996 && 0.908 & 0.981 & 0.995 && 0.951 & 0.992 & 0.998\\
0.5 & 0.5 && 0.847 & 0.975 & 0.998 && 0.935 & 1.001 & 1.004 && 0.966 & 1.007 & 1.006 && 0.939 & 0.992 & 0.998 && 0.895 & 0.987 & 0.997 && 0.919 & 0.990 & 0.998 && 0.881 & 0.985 & 0.997 && 0.904 & 0.979 & 0.993 && 0.972 & 0.998 & 1.000\\
1.0 & 1.0 && 0.837 & 0.965 & 0.990 && 1.044 & 1.072 & 1.050 && 1.117 & 1.094 & 1.058 && 0.941 & 0.990 & 0.997 && 0.893 & 0.985 & 0.997 && 0.904 & 0.986 & 0.997 && 0.890 & 0.985 & 0.997 && 0.900 & 0.975 & 0.991 && 0.998 & 1.000 & 1.000
\rowsep{0.01}\\
\multicolumn{38}{c}{$N=1000$, $\rho_{\gamma x}=0.5$}\\
0.0 & 0.5 && 0.857 & 0.982 & 1.006 && 0.904 & 0.984 & 0.996 && 0.918 & 0.987 & 0.996 && 0.935 & 0.989 & 0.997 && 0.900 & 0.986 & 0.997 && 0.931 & 0.991 & 0.998 && 0.877 & 0.982 & 0.996 && 0.908 & 0.981 & 0.995 && 0.951 & 0.992 & 0.998\\
0.5 & 0.5 && 0.842 & 0.969 & 0.993 && 0.930 & 0.999 & 1.004 && 0.965 & 1.007 & 1.006 && 0.942 & 0.992 & 0.998 && 0.892 & 0.986 & 0.997 && 0.917 & 0.990 & 0.998 && 0.876 & 0.985 & 0.997 && 0.902 & 0.977 & 0.992 && 0.972 & 0.998 & 1.000\\
1.0 & 1.0 && 0.835 & 0.963 & 0.987 && 1.040 & 1.070 & 1.049 && 1.116 & 1.093 & 1.057 && 0.948 & 0.992 & 0.998 && 0.890 & 0.985 & 0.997 && 0.902 & 0.987 & 0.997 && 0.887 & 0.985 & 0.997 && 0.899 & 0.975 & 0.990 && 0.999 & 1.000 & 1.000
\rowsep{0.5}\\
\multicolumn{38}{c}{Conditional on $\kappa_i=\pm 1$}\\
\hline
\multicolumn{38}{c}{$N=100$, $\rho_{\gamma x}=0$}\\
0.0 & 0.5 && 0.743 & 0.997 & 1.062 && 0.733 & 0.924 & 0.972 && 0.751 & 0.928 & 0.973 && 0.803 & 0.945 & 0.979 && 0.834 & 0.971 & 0.993 && 0.920 & 0.989 & 0.998 && 0.756 & 0.949 & 0.987 && 0.806 & 0.948 & 0.984 && 0.842 & 0.951 & 0.981\\
0.5 & 0.5 && 0.872 & 1.164 & 1.239 && 0.917 & 1.117 & 1.163 && 0.956 & 1.126 & 1.166 && 0.850 & 0.970 & 0.992 && 0.839 & 0.974 & 0.994 && 0.915 & 0.989 & 0.998 && 0.788 & 0.964 & 0.992 && 0.832 & 0.965 & 0.991 && 0.914 & 0.985 & 0.996\\
1.0 & 1.0 && 1.051 & 1.459 & 1.567 && 1.273 & 1.514 & 1.537 && 1.386 & 1.551 & 1.549 && 0.874 & 0.978 & 0.995 && 0.806 & 0.971 & 0.994 && 0.868 & 0.983 & 0.997 && 0.788 & 0.970 & 0.994 && 0.854 & 0.975 & 0.995 && 0.980 & 0.999 & 1.000
\rowsep{0.01}\\
\multicolumn{38}{c}{$N=100$, $\rho_{\gamma x}=0.5$}\\
0.0 & 0.5 && 0.763 & 1.023 & 1.090 && 0.730 & 0.923 & 0.971 && 0.751 & 0.928 & 0.973 && 0.801 & 0.944 & 0.979 && 0.829 & 0.971 & 0.994 && 0.916 & 0.988 & 0.998 && 0.752 & 0.949 & 0.987 && 0.808 & 0.950 & 0.985 && 0.842 & 0.951 & 0.981\\
0.5 & 0.5 && 0.898 & 1.196 & 1.273 && 0.918 & 1.123 & 1.169 && 0.962 & 1.133 & 1.171 && 0.847 & 0.968 & 0.992 && 0.831 & 0.972 & 0.994 && 0.907 & 0.988 & 0.998 && 0.779 & 0.961 & 0.992 && 0.836 & 0.966 & 0.991 && 0.916 & 0.986 & 0.997\\
1.0 & 1.0 && 1.073 & 1.481 & 1.591 && 1.290 & 1.525 & 1.542 && 1.404 & 1.563 & 1.555 && 0.875 & 0.974 & 0.993 && 0.809 & 0.968 & 0.993 && 0.867 & 0.980 & 0.996 && 0.788 & 0.965 & 0.993 && 0.856 & 0.976 & 0.995 && 0.984 & 0.999 & 1.000
\rowsep{0.01}\\
\multicolumn{38}{c}{$N=1000$, $\rho_{\gamma x}=0$}\\
0.0 & 0.5 && 0.746 & 1.001 & 1.065 && 0.732 & 0.925 & 0.971 && 0.753 & 0.929 & 0.972 && 0.810 & 0.947 & 0.979 && 0.759 & 0.951 & 0.987 && 0.833 & 0.972 & 0.993 && 0.720 & 0.933 & 0.978 && 0.804 & 0.949 & 0.984 && 0.840 & 0.952 & 0.980\\
0.5 & 0.5 && 0.880 & 1.188 & 1.265 && 0.910 & 1.140 & 1.191 && 0.956 & 1.151 & 1.194 && 0.835 & 0.965 & 0.990 && 0.772 & 0.959 & 0.990 && 0.819 & 0.969 & 0.993 && 0.758 & 0.957 & 0.989 && 0.829 & 0.965 & 0.991 && 0.908 & 0.986 & 0.997\\
1.0 & 1.0 && 1.120 & 1.516 & 1.619 && 1.323 & 1.583 & 1.616 && 1.456 & 1.628 & 1.632 && 0.856 & 0.969 & 0.991 && 0.788 & 0.962 & 0.991 && 0.804 & 0.964 & 0.991 && 0.788 & 0.962 & 0.991 && 0.863 & 0.978 & 0.995 && 0.976 & 0.998 & 0.999
\rowsep{0.01}\\
\multicolumn{38}{c}{$N=1000$, $\rho_{\gamma x}=0.5$}\\
0.0 & 0.5 && 0.760 & 1.020 & 1.086 && 0.730 & 0.925 & 0.971 && 0.753 & 0.929 & 0.972 && 0.811 & 0.947 & 0.979 && 0.756 & 0.950 & 0.987 && 0.830 & 0.971 & 0.993 && 0.721 & 0.933 & 0.978 && 0.806 & 0.951 & 0.985 && 0.840 & 0.952 & 0.980\\
0.5 & 0.5 && 0.877 & 1.184 & 1.262 && 0.888 & 1.119 & 1.169 && 0.941 & 1.131 & 1.172 && 0.837 & 0.965 & 0.989 && 0.765 & 0.957 & 0.989 && 0.813 & 0.968 & 0.992 && 0.754 & 0.954 & 0.989 && 0.827 & 0.964 & 0.991 && 0.905 & 0.985 & 0.996\\
1.0 & 1.0 && 1.090 & 1.478 & 1.579 && 1.275 & 1.531 & 1.560 && 1.413 & 1.578 & 1.576 && 0.864 & 0.971 & 0.991 && 0.780 & 0.961 & 0.990 && 0.798 & 0.964 & 0.991 && 0.779 & 0.961 & 0.990 && 0.859 & 0.977 & 0.995 && 0.976 & 0.998 & 0.999\\
\hline\hline
\multicolumn{38}{p{32.5cm}}{\footnotesize{Notes: The hierarchical Bayesian forecasts differ in the following 
priors: 
(1) $\boldsymbol{S}_{\bar{\theta}}=\boldsymbol{I}_{K}10^{6}$, $\boldsymbol{S}_{\Sigma }=\boldsymbol{I}_{K}10$, 
(2)  $\boldsymbol{S}_{\bar{\theta}} = \boldsymbol{I}_{K}10^{2}$, $\boldsymbol{S}_{\Sigma }=\boldsymbol{I}_{K}10^{2}$, 
and 
(3) $\boldsymbol{S}_{\bar{\theta}} = \boldsymbol{I}_{K}$, $\boldsymbol{S}_{\Sigma }=\boldsymbol{I}_{K}$. The results are based on
1000 replications. For further
details see the footnote of Table~\ref{tbl:MC_ARX_N100} in the paper.}}
\end{tabular}
}
\end{sidewaystable}

%% file: tables/MC_N1000_appendix_April2025.tex
\begin{sidewaystable}\thisfloatpagestyle{empty}
\caption{Monte Carlo results, $N=1000$}
\label{tbl:MC_ARX_N1000} 
\hspace{-3em} 
\resizebox{25cm}{!}{
\setlength\tabcolsep{3pt}
\begin{tabular}{lllllllllllllllllllllllllllllll}
\hline\hline
$a_\beta$&$\sigma^2_\alpha$  &&
\multicolumn{3}{c}{Pooled}&&\multicolumn{3}{c}{RE}&&\multicolumn{3}{c}{FE} &&
\multicolumn{3}{c}{Empirical Bayes}&&
\multicolumn{3}{c}{Comb.\ (pool)} &&
\multicolumn{3}{c}{Comb.\ (FE)}&& 
\multicolumn{3}{c}{Comb.\ $\omega_i^*$}\\
\cline{1-2}\cline{4-6}\cline{8-10}\cline{12-14}\cline{16-18}\cline{20-22}\cline{24-26}\cline{28-30}
&$T$&&\multicolumn{1}{c}{20}&\multicolumn{1}{c}{50}&\multicolumn{1}{c}{100}&&\multicolumn{1}{c}{20}&\multicolumn{1}{c}{50}&\multicolumn{1}{c}{100}&&\multicolumn{1}{c}{20}&\multicolumn{1}{c}{50}&\multicolumn{1}{c}{100}&&\multicolumn{1}{c}{20}&\multicolumn{1}{c}{50}&\multicolumn{1}{c}{100}&&\multicolumn{1}{c}{20}&\multicolumn{1}{c}{50}&\multicolumn{1}{c}{100}&&\multicolumn{1}{c}{20}&\multicolumn{1}{c}{50}&\multicolumn{1}{c}{100}&&\multicolumn{1}{c}{20}&\multicolumn{1}{c}{50}&\multicolumn{1}{c}{100}\\
\hline
\rowsep{-0.975}\\
\multicolumn{30}{c}{Conditional on $\kappa_i=0$}\\
\hline
\multicolumn{30}{c}{$\rho_{\gamma x}=0$}\\
0.0&0.5&&0.859&0.981&1.006&&0.905&0.984&0.996&&0.919&0.986&0.997&&0.935&0.989&0.997&&0.909&0.981&0.995&&0.951&0.992&0.998&&0.951&0.995&0.999\\
0.5&0.5&&0.849&0.974&0.999&&0.936&1.000&1.004&&0.966&1.007&1.006&&0.939&0.991&0.998&&0.905&0.979&0.994&&0.972&0.998&1.000&&0.942&0.994&0.999\\
1.0&1.0&&0.839&0.964&0.990&&1.042&1.071&1.050&&1.114&1.094&1.058&&0.942&0.990&0.998&&0.901&0.975&0.992&&0.998&1.000&1.000&&0.930&0.991&0.999%
\rowsep{0.01}\\
\multicolumn{30}{c}{$\rho_{\gamma x}=0.5$}\\
0.0&0.5&&0.859&0.981&1.006&&0.905&0.984&0.996&&0.919&0.986&0.997&&0.936&0.989&0.997&&0.909&0.981&0.995&&0.951&0.992&0.998&&0.953&0.995&0.999\\
0.5&0.5&&0.845&0.968&0.993&&0.931&0.999&1.004&&0.965&1.007&1.006&&0.943&0.992&0.998&&0.903&0.977&0.993&&0.972&0.998&1.000&&0.947&0.994&0.999\\
1.0&1.0&&0.837&0.962&0.988&&1.039&1.070&1.049&&1.114&1.094&1.057&&0.949&0.991&0.998&&0.900&0.974&0.991&&0.998&1.000&1.000&&0.935&0.992&0.999%
\rowsep{0.25}\\
\multicolumn{30}{c}{Conditional on $\kappa_i=\pm 1$}\\
\hline
\multicolumn{30}{c}{$\rho_{\gamma x}=0$}\\
0.0&0.5&&0.750&1.000&1.065&&0.735&0.924&0.971&&0.755&0.929&0.972&&0.812&0.947&0.980&&0.806&0.949&0.984&&0.841&0.951&0.980&&0.819&0.963&0.990\\
0.5&0.5&&0.883&1.186&1.264&&0.912&1.138&1.190&&0.957&1.149&1.193&&0.836&0.964&0.990&&0.830&0.964&0.991&&0.909&0.985&0.997&&0.820&0.966&0.992\\
1.0&1.0&&1.121&1.514&1.617&&1.323&1.579&1.613&&1.456&1.625&1.628&&0.856&0.968&0.991&&0.863&0.977&0.995&&0.977&0.997&0.999&&0.827&0.968&0.994%
\rowsep{0.01}\\
\multicolumn{30}{c}{$\rho_{\gamma x}=0.5$}\\
0.0&0.5&&0.764&1.019&1.085&&0.733&0.924&0.971&&0.755&0.929&0.972&&0.813&0.947&0.980&&0.808&0.950&0.985&&0.841&0.951&0.980&&0.821&0.964&0.990\\
0.5&0.5&&0.880&1.182&1.261&&0.890&1.116&1.168&&0.942&1.129&1.171&&0.838&0.964&0.990&&0.828&0.963&0.991&&0.906&0.985&0.996&&0.828&0.966&0.992\\
1.0&1.0&&1.091&1.476&1.576&&1.275&1.527&1.557&&1.412&1.573&1.573&&0.864&0.970&0.992&&0.859&0.976&0.995&&0.977&0.998&0.999&&0.836&0.968&0.993\\
\hline\hline
\multicolumn{30}{p{25.3cm}}{\footnotesize{Notes: The results are for $N=1000$. For further details see Table~1 in the paper.}}
\end{tabular}
}
\end{sidewaystable}  

%% file: tables/MC_appendix.tex
\begin{table}\thisfloatpagestyle{empty}
\caption{Monte Carlo results for equally weighted and Oracle forecasts}
\label{tbl:MC_Appendix} 
\setlength\tabcolsep{5pt}
\begin{tabular}{llllllllllllllllll}
\hline\hline
$a_\beta$&$\sigma^2_\alpha$ &&
\multicolumn{3}{c}{eq.weight(pool)}&&
\multicolumn{3}{c}{eq.weight(pool)}&&
\multicolumn{3}{c}{Oracle(pool)} &&
\multicolumn{3}{c}{Oracle(FE)}\\
 \cline{1-2}\cline{4-6}\cline{8-10}\cline{12-14}\cline{16-18}
&$T$&&\multicolumn{1}{c}{20}&\multicolumn{1}{c}{50}&\multicolumn{1}{c}{100}&&\multicolumn{1}{c}{20}&\multicolumn{1}{c}{50}&\multicolumn{1}{c}{100}&&\multicolumn{1}{c}{20}&\multicolumn{1}{c}{50}&\multicolumn{1}{c}{100}&&\multicolumn{1}{c}{20}&\multicolumn{1}{c}{50}&\multicolumn{1}{c}{100}\\
\hline
\rowsep{-0.975}\\
\multicolumn{18}{c}{Conditional on $\kappa_i=0$}\\
\hline
\multicolumn{18}{c}{$N=100$, $\rho_{\gamma x}=0$}\\
0.0 & 0.5 && 0.888 & 0.977 & 0.995 && 0.950 & 0.992 & 0.998 && 0.872 & 0.983 & 0.997 && 0.944 & 0.995 & 0.999\\
0.5 & 0.5 && 0.886 & 0.976 & 0.994 && 0.970 & 0.999 & 1.001 && 0.869 & 0.980 & 0.996 && 0.968 & 0.999 & 1.003\\
1.0 & 1.0 && 0.860 & 0.968 & 0.990 && 1.038 & 1.036 & 1.020 && 0.857 & 0.974 & 0.994 && 1.017 & 1.021 & 1.016%
\rowsep{0.1}\\
\multicolumn{18}{c}{$N=100$, $\rho_{\gamma x}=0.5$}\\
0.0 & 0.5 && 0.887 & 0.976 & 0.995 && 0.950 & 0.992 & 0.998 && 0.870 & 0.982 & 0.997 && 0.944 & 0.995 & 0.999\\
0.5 & 0.5 && 0.885 & 0.975 & 0.993 && 0.970 & 0.999 & 1.001 && 0.870 & 0.980 & 0.996 && 0.968 & 0.999 & 1.003\\
1.0 & 1.0 && 0.860 & 0.968 & 0.990 && 1.041 & 1.038 & 1.022 && 0.853 & 0.973 & 0.994 && 1.020 & 1.022 & 1.017%
\rowsep{0.1}\\
\multicolumn{18}{c}{$N=1000$, $\rho_{\gamma x}=0$}\\
0.0 & 0.5 && 0.886 & 0.976 & 0.994 && 0.947 & 0.991 & 0.998 && 0.870 & 0.983 & 0.997 && 0.938 & 0.993 & 0.999\\
0.5 & 0.5 && 0.881 & 0.974 & 0.992 && 0.963 & 0.997 & 1.000 && 0.861 & 0.979 & 0.996 && 0.961 & 0.998 & 1.003\\
1.0 & 1.0 && 0.877 & 0.970 & 0.990 && 1.016 & 1.024 & 1.015 && 0.876 & 0.975 & 0.993 && 1.006 & 1.015 & 1.015%
\rowsep{0.1}\\
\multicolumn{18}{c}{$N=1000$, $\rho_{\gamma x}=0.5$}\\
0.0 & 0.5 && 0.886 & 0.976 & 0.994 && 0.947 & 0.991 & 0.998 && 0.868 & 0.982 & 0.997 && 0.938 & 0.993 & 0.999\\
0.5 & 0.5 && 0.880 & 0.972 & 0.991 && 0.963 & 0.998 & 1.000 && 0.864 & 0.978 & 0.995 && 0.961 & 0.998 & 1.003\\
1.0 & 1.0 && 0.877 & 0.970 & 0.989 && 1.017 & 1.024 & 1.015 && 0.895 & 0.978 & 0.993 && 1.006 & 1.015 & 1.015%
\rowsep{0.25}\\
\multicolumn{18}{c}{Conditional on $\kappa_i=\pm 1$}\\
\hline
\multicolumn{18}{c}{$N=100$, $\rho_{\gamma x}=0$}\\
0.0 & 0.5 && 0.744 & 0.935 & 0.989 && 0.826 & 0.948 & 0.980 && 0.712 & 0.934 & 0.982 && 0.758 & 0.929 & 0.973\\
0.5 & 0.5 && 0.773 & 0.974 & 1.032 && 0.885 & 0.997 & 1.026 && 0.771 & 0.966 & 0.994 && 0.881 & 0.987 & 1.008\\
1.0 & 1.0 && 0.806 & 1.042 & 1.113 && 1.037 & 1.120 & 1.129 && 0.861 & 0.991 & 1.002 && 0.997 & 1.045 & 1.049%
\rowsep{0.1}\\
\multicolumn{18}{c}{$N=100$, $\rho_{\gamma x}=0.5$}\\
0.0 & 0.5 && 0.745 & 0.939 & 0.994 && 0.826 & 0.948 & 0.980 && 0.719 & 0.939 & 0.984 && 0.758 & 0.929 & 0.973\\
0.5 & 0.5 && 0.778 & 0.981 & 1.040 && 0.888 & 1.000 & 1.028 && 0.782 & 0.972 & 0.996 && 0.884 & 0.989 & 1.009\\
1.0 & 1.0 && 0.810 & 1.048 & 1.119 && 1.043 & 1.125 & 1.131 && 0.847 & 0.989 & 1.002 && 1.003 & 1.048 & 1.051%
\rowsep{0.1}\\
\multicolumn{18}{c}{$N=1000$, $\rho_{\gamma x}=0$}\\
0.0 & 0.5 && 0.748 & 0.936 & 0.988 && 0.829 & 0.949 & 0.980 && 0.717 & 0.936 & 0.982 && 0.761 & 0.930 & 0.973\\
0.5 & 0.5 && 0.774 & 0.978 & 1.036 && 0.884 & 1.003 & 1.033 && 0.777 & 0.969 & 0.995 && 0.883 & 0.990 & 1.011\\
1.0 & 1.0 && 0.831 & 1.059 & 1.123 && 1.042 & 1.133 & 1.146 && 0.914 & 1.126 & 1.163 && 0.997 & 1.045 & 1.052%
\rowsep{0.1}\\
\multicolumn{18}{c}{$N=1000$, $\rho_{\gamma x}=0.5$}\\
0.0 & 0.5 && 0.748 & 0.939 & 0.993 && 0.829 & 0.949 & 0.980 && 0.723 & 0.939 & 0.984 && 0.761 & 0.930 & 0.973\\
0.5 & 0.5 && 0.771 & 0.976 & 1.035 && 0.881 & 0.998 & 1.028 && 0.785 & 0.995 & 1.024 && 0.879 & 0.988 & 1.009\\
1.0 & 1.0 && 0.823 & 1.049 & 1.113 && 1.032 & 1.121 & 1.132 && 0.918 & 1.105 & 1.136 && 0.993 & 1.041 & 1.047\\
\hline\hline
\multicolumn{18}{p{17.5cm}}{\footnotesize{Notes: The results are for equal weighted combinations of individual and pooled
forecasts and for individual and FE forecasts, and for combinations using oracle weights, which use the
disturbances and parameteres for the construction of the weights. 
For further details see the footnote of Table~\ref{tbl:MC_ARX_N100} in the paper.}}
\end{tabular}
\end{table}

%% file: tables/applications_appendix_April_2025.tex
\begin{sidewaystable}\thisfloatpagestyle{empty}
\caption{Results for the applications, including hierarchical Bayesian and equal weights forecasts}
\label{tbl:applications_msfe_appendix}
\centering
{\footnotesize
\begin{tabular}{llllllllllllllll}
\hline\hline
 & \multicolumn{3}{l}{Ratio of} && \multicolumn{3}{l}{Freq.\ beating} &&\multicolumn{3}{l}{Freq.\ smallest}&&\multicolumn{3}{l}{Freq.\ largest}\\
 & \multicolumn{3}{l}{ave. MSFE} &&\multicolumn{3}{l}{benchmark} &&\multicolumn{3}{l}{MSFE}&&\multicolumn{3}{l}{MSFE}\\
\cline{2-4}\cline{6-8}\cline{10-12}\cline{13-16}
Observations                 & all & $\kappa_i=0$ & $\kappa_i=\pm 1$ && all & $\kappa_i=0$ & $\kappa_i=\pm 1$&& all & $\kappa_i=0$ & $\kappa_i=\pm 1$ && all& $\kappa_i=0$ & $\kappa_i=\pm 1$ \\
\hline
\rowsep{-.5}\\
\multicolumn{16}{l}{House price inflation forecasts}\\
\hline
Individual              & 2.822 & 2.520 & 3.542 &&  --   &  --   &  --   && 0.008 & 0.218 & 0.122 && 0.564 & 0.238 & 0.381\rowsep{.2}\\
Pooled                  & 0.920 & 1.162 & 0.947 && 0.613 & 0.381 & 0.536 && 0.113 & 0.157 & 0.246 && 0.157 & 0.188 & 0.152\\
RE                      & 0.924 & 1.166 & 0.960 && 0.619 & 0.376 & 0.528 && 0.099 & 0.022 & 0.044 && 0.003 & 0.017 & 0.008\\
FE                      & 0.936 & 1.186 & 0.980 && 0.591 & 0.381 & 0.517 && 0.039 & 0.108 & 0.099 && 0.251 & 0.370 & 0.290\\
Emp.Bayes               & 0.901 & 0.955 & 0.881 && 0.942 & 0.519 & 0.652 && 0.083 & 0.066 & 0.055 && 0.000 & 0.039 & 0.017\\
Hier.Bayes (1)          & 0.939 & 0.972 & 0.939 && 0.939 & 0.550 & 0.674 && 0.022 & 0.052 & 0.028 && 0.000 & 0.025 & 0.022\\
Hier.Bayes (2)          & 0.978 & 0.988 & 0.976 && 0.939 & 0.550 & 0.671 && 0.003 & 0.047 & 0.028 && 0.006 & 0.039 & 0.030\\
Hier.Bayes (3)          & 0.908 & 0.970 & 0.914 && 0.914 & 0.500 & 0.622 && 0.257 & 0.105 & 0.094 && 0.003 & 0.017 & 0.041\\
Comb.\ (pool)           & 0.920 & 0.961 & 0.932 && 0.939 & 0.522 & 0.688 && 0.119 & 0.058 & 0.086 && 0.000 & 0.011 & 0.014\\
Comb.\ (FE)             & 0.937 & 0.977 & 0.940 && 0.917 & 0.494 & 0.677 && 0.014 & 0.064 & 0.064 && 0.011 & 0.033 & 0.041\\
indiv.weights           & 0.921 & 0.957 & 0.909 && 0.936 & 0.541 & 0.713 && 0.030 & 0.022 & 0.044 && 0.006 & 0.003 & 0.006\\
eq.weights (pool)       & 0.878 & 0.982 & 0.879 && 0.914 & 0.453 & 0.638 && 0.160 & 0.036 & 0.064 && 0.000 & 0.000 & 0.000\\
eq.weights (FE)         & 0.887 & 0.994 & 0.898 && 0.906 & 0.431 & 0.624 && 0.052 & 0.030 & 0.039 && 0.000 & 0.000 & 0.000\rowsep{.5}\\
\multicolumn{16}{l}{CPI inflation forecasts}\\
\hline
 Individual                 &15.501& 10.451& 11.295&&  --  &  --  &  --  && 0.000 & 0.064 & 0.037 && 0.433 & 0.134 & 0.278\rowsep{.2}\\
 Pooled                     &0.878 & 1.013 & 0.971 && 0.444 & 0.374 & 0.417 && 0.193 & 0.112 & 0.091 && 0.396 & 0.406 & 0.374\\
 RE                         &0.880 & 1.001 & 0.957 && 0.508 & 0.390 & 0.401 && 0.005 & 0.070 & 0.064 && 0.000 & 0.037 & 0.032\\
 FE                         &0.883 & 0.992 & 0.959 && 0.508 & 0.401 & 0.401 && 0.000 & 0.059 & 0.096 && 0.166 & 0.225 & 0.219\\
 Emp.Bayes                  &0.892 & 0.991 & 0.926 && 0.984 & 0.652 & 0.818 && 0.225 & 0.187 & 0.193 && 0.000 & 0.064 & 0.005\\
 Hier.Bayes (1)             &0.970 & 0.982 & 0.979 && 0.904 & 0.610 & 0.738 && 0.027 & 0.064 & 0.037 && 0.000 & 0.016 & 0.016\\
 Hier.Bayes (2)             &0.987 & 0.993 & 0.992 && 0.909 & 0.551 & 0.706 && 0.005 & 0.075 & 0.032 && 0.005 & 0.070 & 0.032\\
 Hier.Bayes (3)             &0.951 & 0.973 & 0.962 && 0.872 & 0.599 & 0.695 && 0.155 & 0.096 & 0.118 && 0.000 & 0.021 & 0.005\\
 Comb.\ (pool)              &0.930 & 0.987 & 0.953 && 0.733 & 0.481 & 0.572 && 0.037 & 0.064 & 0.064 && 0.000 & 0.016 & 0.016\\
 Comb.\ (FE)                &0.935 & 0.980 & 0.967 && 0.791 & 0.524 & 0.583 && 0.021 & 0.043 & 0.032 && 0.000 & 0.011 & 0.021\\
 indiv.weights              &0.897 & 0.972 & 0.931 && 0.973 & 0.695 & 0.813 && 0.128 & 0.102 & 0.118 && 0.000 & 0.000 & 0.000\\
 eq.weights (pool)          &0.895 & 0.984 & 0.945 && 0.802 & 0.519 & 0.647 && 0.166 & 0.053 & 0.096 && 0.000 & 0.000 & 0.000\\
 eq.weights (FE)            &0.900 & 0.982 & 0.942 && 0.856 & 0.535 & 0.647 && 0.037 & 0.011 & 0.021 && 0.000 & 0.000 & 0.000\\
 \hline\hline
\multicolumn{16}{p{20.1cm}}{\footnotesize{%
Notes:
The table reports the results for the methods in the paper and for five additional forecasts:
hierarchical Bayesian forecasts for three priors: (1) $\boldsymbol{S}_{\bar{\theta}}=\boldsymbol{I}_{K}10^{6}$, $\boldsymbol{S}%
_{\Sigma }=\boldsymbol{I}_{K}10$, (2)  $\boldsymbol{S}_{\bar{\theta}} = \boldsymbol{I}_{K}10^{2}$, $\boldsymbol{S}_{\Sigma }=\boldsymbol{I}_{K}10^{2}$, and (3) $\boldsymbol{S}_{\bar{\theta}} = \boldsymbol{I}_{K}$, $\boldsymbol{S}_{\Sigma }=\boldsymbol{I}_{K}$,
 the forecast that is an equal weighted average of the individual and the pooled forecasts and, finally,
the forecast that is an equal weighted average of the individual and the FE forecasts.
For further details see the footnote of Table~\ref{tbl:applications_msfe} in the paper. 
}}
\end{tabular}
}
\end{sidewaystable}

%% file: tables/DMtestStats_Appendix_7April25.tex




\begin{table}
\caption{Diebold-Mariano test statistics for equal predictive accuracy: equal weights forecasts}
\label{tbl:DM_appendix}\centering
\hspace*{-1.5cm}
\begin{tabular}{lrrrrr}
\hline\hline
  & hier.Bayes (1) & hier.Bayes (2) & hier.Bayes (3) & eq.weights(pool)& eq.weights(FE)\\
 \hline
\multicolumn{3}{l}{House Prices: all forecasts}\\
\hline
 Panel DM & $-$29.68 & $-$30.12 & $-$27.88 & $-$26.76 & $-$25.20\\
$\text{DM}<-1.96$/$\text{DM}>1.96$  & 204/1 & 213/0 & 182/1 & 158/0 & 148/1
\rowsep{0.25}\\
\multicolumn{3}{l}{CPI: all forecasts}\\
\hline
Panel DM   & $-$13.15 & $-$10.12 & $-$17.78 & $-$11.53 & $-$10.76\\
$\text{DM}<-1.96$/$\text{DM}>1.96$ & 112/8 & 90/3 & 94/8 & 79/14 & 79/9\\
\hline\hline
\multicolumn{6}{p{18.5cm}}{\footnotesize{Notes:
	The table reports the DM statistics for the three hierarchical Bayesian forecasts and the two equally weighted forecasts, 
	where the first combines individual and pooled forecasts and the second individual and FE forecasts.
	For further details see the footnote of Table~\ref{tbl:DM} in the paper.
}}
\end{tabular}
\end{table}
